\pdfoutput=1

\documentclass[english,11pt]{article}
\include{macros}

\pgfplotsset{compat=1.18}
\begin{document}

\title{\Large Benefits and Costs of Adaptive Sampling}

\author{Yu-Shiou Willy Lin\thanks{Department of IEOR, Columbia University. Email: \texttt{yl5782@columbia.edu}} \and Dae Woong Ham\thanks{Department of Technology and Operations, Ross School of Business. Email: \texttt{daewoong@umich.edu}} \and Iavor Bojinov\thanks{Harvard Business School. \texttt{Ibojinov@hbs.edu}} }

\date{}
\maketitle

\begin{abstract}

Multi-armed bandits are widely used for sequential experimentation in clinical trials, public policy treatment, recommendation systems, and online platforms. While regret minimization and valid inference from adaptively collected data have each been studied individually, one question remains: when does adaptivity \emph{improve estimation precision} relative to uniform designs, and how should inference be balanced against the online cost of experimentation? We first study arm-level mean estimation under mean-squared-error (MSE) objectives. We characterize when an adaptive Neyman allocation, which allocates samples according to arm variance, yields strict MSE improvements over uniform sampling in finite samples. When there is variance heterogeneity across arms, these improvements arise at modest sample sizes, clarifying that adaptivity can be preferable for inference not only asymptotically, but also in finite-sample settings. We then study a joint inference--regret objective that accounts for the cost of assigning units to inferior arms during experimentation. We propose three policies: the Plug-in Oracle Rate Policy, Static-Allocation Rate Policy, and Neyman-Adaptive Rate Policy, all of which attain the oracle-optimal asymptotic rate. The latter two offer computationally efficient and easy-to-implement alternatives by combining existing regret-minimizing algorithms with inference-oriented sampling through closed-form mixing rules, avoiding repeated numerical optimization while being managerially interpretable. We compare their finite-sample statistical performance through simulations and empirically calibrated applications. The results highlight the relative strengths of our policies across instances and their tradeoffs between joint inference--regret performance and computational cost.

\end{abstract}


% {\color{purple}

% Verification:
% \begin{enumerate}
%     \item Check the AE and the reviewer and make sure the citation makes sense and cite them properly.
%     \item Write letter on how this matches to the journal we are submitting to.
% \end{enumerate}
% }

\section{Introduction}\label{sec:intro}

Adaptive experimental design uses information collected during an experiment to guide future assignment decisions, fueling online experimentation to become central in decision-making across a wide range of operational settings, including recommendation systems, digital advertising, and clinical decision-making, where organizations repeatedly run experiments to improve products and treatments (\cite{bojinov2025experimentation, netflix2022experimentation}). At the same time, these experiments are rarely governed by a single objective. A practitioner may care both about experiment efficiency and about performance during the experiment itself, since assigning customers or patients to inferior alternatives can lower revenue, engagement, or increase avoidable medical costs while learning (\cite{wilson2025rar, netflix2025returnaware, ham2026design}). Adaptive experimentation has been shown to be very effective in tackling the experimentation efficiency and regret minimization problems separately.

Despite this relevance, adaptive experimentation is still often viewed as impractical in industry. Using experimentation efficiency as an example, consider a data scientist evaluating which recommendation system performs best for a streaming platform. If the primary objective is accurate estimation of treatment effects, then the natural statistical goal is to minimize the mean squared error (MSE) of the estimators. In principle, adaptive assignment can use online observations to allocate more samples where they are most informative. In practice, however, data scientists often worry that changing assignment probabilities over time will introduce bias if analyzed naively. Learning the allocation from noisy data and accounting for data-dependent assignment can also introduce additional variability. As a result, many experiments continue to rely on uniform sampling with the understanding that the penalty paid by an adaptive procedure does not outweigh the inferential gains. This raises a basic question: when do the precision gains from directing more samples toward arms that are harder to estimate, rather than those that require fewer observations, outweigh the statistical costs of adaptivity in finite samples?

To address these problems, we study the value of adaptive experimentation relative to non-adaptive designs under two objectives. We begin with the pure inference setting, where the goal is to improve statistical efficiency as measured by MSE. In this setting, we show that even a two-stage simple data-driven adaptive sampling design can outperform uniform sampling once the pilot learning phase is sufficiently informative. We then turn to a joint inference--regret setting, which captures the realistic motivations that experiments are often run in environments where continued assignment to suboptimal treatments carries real online cost (regret) but managers desire to learn an effect as well (inference). If the best-performing treatment were known, assigning every unit to it would minimize regret but provides no observations, thus no inference, on the other arms. Conversely, precisely estimating all arm means requires sampling inferior treatments as well, sacrificing expected performance during the experiment, demonstrating the tradeoff between regret and inference. To overcome this tradeoff, we propose three adaptive sampling rules,  Plug-in Oracle Rate Policy (PORP), Static-Allocation Rate Policy (SARP), and Neyman-Adaptive Rate Policy (NARP), that balance statistical accuracy with online performance, commonly formalized as regret in the multi-armed bandit literature. 

PORP directly computes the plug-in allocation but requires a numerical multiplier search at each non-forced update. Although PORP is closer to optimal, we argue that the computational burden can quickly be practically unattractive especially when adaptive allocation rules are updated repeatedly across a large experimentation program. For example, \citet{ham2026design} describe a Netflix prepaid-card experiment that adaptively allocated approximately 125,000 visitors over two months. The relevance of these per-decision savings extends beyond individual experiments; specifically, \citet{netflix2022experimentation} report that Netflix conducts thousands of streaming experiments annually, which only exacerbates the computational problem exponentially. These considerations motivate SARP and NARP, which avoid repeated numerical solvers by combining a suitable off-the-shelf regret-minimizing algorithm, such as Thompson sampling, with inference-oriented exploration through closed-form mixing rules. This modular construction also offers important managerial advantages. Reusing existing implementations is not only practically attractive from an engineering perspective but it also can simplify verification and maintenance, supporting reliable AI deployment as shown in recent case studies \citep{kadioglu2025opensource}. Furthermore, the explicit mixing rule by NARP and SARP provides an interpretable explanation of how inference and regret are balanced, addressing the stakeholder communication needs identified in practitioner interviews \citep{hong2020human}.

% Therefore, we introduce SARP and NARP, which avoid repeated numerical allocation solves through closed-form mixing rules that combine inference-oriented exploration with an existing regret-minimizing algorithm. SARP uses a fixed exploration profile, while NARP calibrates exploration using estimated variances and gaps.}

Under the stated assumptions, all three policies attain the $N^{-1/3}$ rate of the complete-information fixed-allocation benchmark. Their practical differences lie in finite-sample joint loss and computational cost, with PORP generally providing the strongest performance. In the reported simulations, SARP performs favorably when top-arm regret separation is evident, whereas NARP achieves lower joint loss than SARP in the configuration combining weak top-arm separation with heterogeneous variances. \footnote{Following standard terminology in the bandit literature, separation refers to a sufficiently large gap between the expected reward of the optimal arm and those of all suboptimal arms, making the optimal arm easier to distinguish.} PORP performs competitively with \texttt{ForcingBalance} \citep{ErraqabiLVBL17}, a state-of-the-art benchmark for balancing inference and regret, while SARP and NARP offer computationally cheaper updates. These comparisons highlight the policies' relative strengths beyond their common asymptotic rate. Rather than claiming to replace optimized problem-specific bandit algorithms \citep{ErraqabiLVBL17,simchi-levi_multi-armed_2025}, our aim is to provide a practical policy family that makes the inference--regret tradeoff explicit, with modular alternatives that practitioners can integrate into existing systems, easy-to-maintain, and have clear interpretation to stakeholders.

Taken together, our results show that adaptive designs can be justified on both statistical and operational grounds: they can improve inference relative to uniform designs, and can be extended to settings where learning must be balanced against performance during the experiment. We begin in Section~\ref{sec:setup} by formulating the problem and introducing the necessary notation. We then present our theoretical results for the inference objective and the joint inference--regret objective separately in Sections~\ref{sec:inf} and~\ref{sec:infreg}. Next, in Section~\ref{sec:simulation}, we provide numerical simulations illustrating the effectiveness of our findings and methods. We also perform a calibrated Monte Carlo experiment, which implements our proposed methods on a real world dataset from~\citet{choudhury2021virtualwatercoolers} in Section~\ref{sec:application}. Finally, we conclude in Section~\ref{sec:discussion}.

\subsection*{Related Literature}

\paragraph{Adaptive Experiment Design}

Our work sits within the growing literature that studies adaptive experiment design through the multi-arm bandit (MAB) lens. This perspective appears in both clinical and digital operation settings (\cite{berry2006bayesian, scott2010modern}), and has been developed in settings where assignments are updated sequentially using accrued data rather than fixed in advance (\cite{Chernoff1992, ECTA17527}). A separate line of work studies how to perform valid, yet not necessarily optimal, statistical inference under adaptive assignment despite the dependence between samples and assignment probabilities (\cite{doi:10.1073/pnas.2014602118, zhang2021statistical, dimakopoulou2021online, doi:10.1287/mnsc.2023.4921}). Our paper adopts this general viewpoint of experimentation as a sequential decision problem, but focuses on simple adaptive rules whose statistical and operational performance can be compared directly to non-adaptive benchmarks. Furthermore, our paper does not focus solely on regret minimization or on validity of inference under adaptive assignment. Instead, of simply focusing on valid statistical inference, we study whether adaptive designs are even worth pursuing if researchers are interested in inference or joint inference--regret objectives.

\paragraph{Statistical Efficiency}

A second closely related literature studies experimental design from the perspective of statistical efficiency. The interest in improving efficiency through design goes back at least to \cite{10.1111/j.2397-2335.1934.tb04184.x}, which proposed the statistically optimal Neyman allocation, and has been extensively developed in later work (\cite{semipar_effic, 10.1111/1468-0262.00442, 10.1214/12-AOS1008, 10.1093/jrsssb/qkad072}). Several recent papers focus specifically on adaptive designs that seek to approach Neyman-optimal allocation as the sample size increases~\citep{zhao2023adaptiveneymanallocation, dai2023clip, li2024optimal, kato2025efficientadaptiveexperimentaldesign}. Our pure-inference investigation differs in emphasis from this line. We do not focus on the asymptotic MSE optimality of a particular adaptive rule. Instead, building on the two-stage adaptive procedure in \cite{zhao2023adaptiveneymanallocation}, we ask when an adaptive design, using only a modest pilot phase, already outperforms the non-adaptive uniform benchmark in finite-sample settings.

\paragraph{Learning-Focused Bandits}

Our pure-inference objective is also closely related to the learning-focused bandit literature, including active learning, mean estimation, and pure exploration (\cite{audibert:hal-00654404, ANTOS20102712, 10.1007/978-3-642-24412-4_17, meanest}). These works study how adaptive sampling can improve learning quality when the objective is estimation or identification rather than cumulative reward maximization. This is therefore closer in spirit to our inference setting than classical regret-minimizing MAB. At the same time, our contribution differs in two ways. First, we are interested in explicit over-performance relative to a natural non-adaptive benchmark, namely uniform sampling, rather than primarily in regret guarantee or sample complexity. Second, we later extend the analysis beyond inference alone to account for the opportunity cost of suboptimal assignments during the experiment. 

\paragraph{Inference and Regret-Focused Bandits}

Our second contribution, which jointly considers experiment efficiency and regret minimization, is closely aligned with \cite{ErraqabiLVBL17}. Their work studies a reward--error tradeoff through a proxy-MSE objective, and their algorithm can be viewed as repeatedly solving a one-step version of the oracle allocation problem in Section~\ref{subsec:oracle_joint}: at each period, the next arm is selected by optimizing the marginal reward--error tradeoff. Our goal is not to establish dominance over such procedures; rather, we aim to derive simpler and more interpretable adaptive rules that attain the oracle asymptotic rate benchmark while remaining easy to deploy. \cite{simchi-levi_multi-armed_2025} also studies this multi-objective problem, proposes corresponding algorithms, and characterizes a Pareto tradeoff between regret minimization and inference efficiency. Relative to these works, our emphasis is on adaptive sampling rules that are intuitive to practitioners and expose the exploration--exploitation structure explicitly. Additionally, our objective function is different from the aforementioned papers, as detailed in Section~\ref{sec:infreg}.

\section{Problem Formulation and Notation Setup}\label{sec:setup} 

%\blue{Add motivation case here? Maybe specify this to link to each section?}

We study a stochastic $K$-armed bandit with horizon $N$. Each arm $i \in [K] := \{1,\dots,K\}$ is associated with an unknown reward distribution $\nu_i$, with arm-specific mean $\mu_i$ and variance $\sigma_i^2 > 0$. At each round $t \ge 1$, the learner selects an arm $I_t \in [K]$ and observes a reward $X_t \sim \nu_{I_t}$. The sampling rule may depend on the past history $\history_{t-1} := (I_1,X_1,\dots,I_{t-1},X_{t-1})$, and is therefore allowed to be adaptive. For each arm $i$, define $T_{i,t} := \sum_{s=1}^t \ind\{I_s = i\}$, and write $T_i := T_{i,N}$ for the total number of samples collected from arm $i$ by the end of the horizon. 

The primary statistical target is learning the vector of arm means $(\mu_1,\dots,\mu_K)$, and the sampling policy $\pi$ determines how the total budget $N$ is distributed across arms through the random sample counts $(T_1,\dots,T_K)$ satisfying $\sum_{i=1}^K T_i = N$. A sampling policy $\pi$ should be understood as a decision rule: at each round $t$, it maps the available history $\history_{t-1}$ to a distribution over arms. We write
\begin{equation}
p_{t,i}^\pi := \PP_\pi(I_t=i \mid \history_{t-1}), \qquad i \in [K],
\end{equation}
to be the conditional probability induced by policy $\pi$ of selecting arm $i$ at round $t$. Thus, $\pi$ denotes the policy as a whole, whereas $p_{t,i}^\pi$ denotes the particular assignment probability generated by that policy at a given time and history. We also write $p_{t,i}$ generically and suppress the dependence on the underlying policy when only the conditional assignment probability matters.

This distinction between $\pi$ and $p^{\pi}$ may not seem necessary for simple non-adaptive designs. For example, under uniform sampling $\pi^{\rm uni}$, the policy and its induced probabilities are effectively indistinguishable, since $p_{t,i}^{\pi^{\mathrm{uni}}} = 1/K$ for all $t \in [N]$ and $i \in [K]$. In contrast, under a sequential policy, the rule $\pi$ may depend on the observed data, and the realized probabilities $p_{t,i}^\pi$ therefore change over time with the history. Keeping these two objects separate is useful in later analysis, especially when the sampling rule is history-dependent and the resulting estimators or risk bounds depend explicitly on the conditional assignment probabilities.  

The formulation above specifies the common stochastic bandit model used throughout the paper. The additional ingredients needed to evaluate the problem at hand---including the estimator, sampling policy, and objective functional---depend on the design problem under consideration and will be introduced separately. We first study inference as the objective in Section~\ref{sec:inf}, where the focus is estimation of the arm means under adaptive data collection. We then turn to the joint inference--regret setting in Section~\ref{sec:infreg}, which augments the inference objective with a regret-based criterion and introduces the notation needed for optimality gaps and online loss.

\section{Adaptive Design for Inference Only}\label{sec:inf}

In this section, we study the design of an efficient experiment for estimating the arm means. For any policy $\pi$, given armwise estimators $\hat\mu_1, \dots, \hat\mu_K$, we define the armwise mean squared error by
\begin{equation}
\mathrm{MSE}(i,\pi) := \EE_\pi \left[(\hat\mu_i-\mu_i)^2\right],
\end{equation}
and our goal is to solve 
\begin{equation}
\min_{\pi} \sum_{i=1}^K \EE_\pi \bigl[(\hat\mu_i-\mu_i)^2\bigr] := \min_{\pi} \sum_{i=1}^K \mathrm{MSE}(i,\pi).
\end{equation}
That is, we seek to understand how the choice of sampling policy, for a given estimator, shapes the total mean squared error (MSE) across arms. Our main interest is whether an adaptive sampling rule can improve upon uniform sampling in finite sample, which serves as the natural non-adaptive baseline. As a reminder, this problem differs from the standard regret-minimizing bandit setting. It is well known that adaptive policies $\pi$ can achieve substantially better regret than non-adaptive exploration rules. However, far less attention has been paid to whether a similar advantage persists when the objective is statistical inference. The works that do study this question typically focus on optimal bandit design rather than directly on inference efficiency relative to non-adaptive baselines \cite{ANTOS20102712, 10.1007/978-3-642-24412-4_17}. We therefore begin with an oracle benchmark that isolates the underlying allocation problem in a static setting and clarifies the efficiency target that a data-driven adaptive design should aim to approximate. 

\subsection{Oracle Benchmark: Neyman Allocation}\label{subsec:ora_ney}

Consider an oracle benchmark in which the arm variances $\sigma_1^2,\dots,\sigma_K^2$ are known a priori by the experimenter and we estimate $\mu_i$ by the sample mean $\hat\mu_j^{\mathrm{sam}}$. In this case, as noted in \cite{ANTOS20102712}, there is no value in adopting adaptive sampling as the goal of adaptive sampling usually is to learn $\hat\sigma$ over time. Therefore, when observations are independent and arm $i$ receives $T_i$ samples under an oracle static case, $\EE\bigl[(\hat\mu_i^{\mathrm{sam}}-\mu_i)^2\bigr] = \frac{\sigma_i^2}{T_i}.$ Hence the total sum of MSE reduces to the static allocation problem
\begin{equation}\label{eq:oracle_inf}
\min_{T_1,\dots,T_K} \sum_{i=1}^K \frac{\sigma_i^2}{T_i} \qquad \text{subject to } \sum_{i=1}^K T_i = N.
\end{equation}
The optimizer to~\ref{eq:oracle_inf} is the classical Neyman allocation, which we denote as $\pi^{\mathrm{Ney}}$ and assigns samples in proportion to the arm standard deviations, i.e. $T_i = N \sigma_i / \sum_{i}^K \sigma_i$, and leads to the total sum of MSE to be
\begin{equation}\label{eq:Ney_mse}
    \sum_{i=1}^K \mathrm{MSE}(i,\pi^{\mathrm{Ney}}) = \sum_{i=1}^K \EE_{\pi^{\mathrm{Ney}}}\bigl[(\hat\mu_i^{\mathrm{sam}} - \mu_i)^2\bigr] = \frac{1}{N} (\sum_{i=1}^K \sigma_i)^2.
\end{equation}

However, when the arm variances are unknown, the oracle allocation is no longer directly available and the variances must be learned, hinting at the benefits of how adaptive policies can help inference. As mentioned above, many practitioners in this case turn directly toward uniform sampling $\pi^{\mathrm{uni}}$, i.e. $p_{t,i}^{\pi^{\mathrm{uni}}} = 1/K$ for all $t \in [N]$ and $i \in [K]$, which is equivalent to setting $T_j = \tfrac{N}{K}$ in a completely randomized trial. By simple algebra, the total sum of MSE under uniform sampling is
\begin{equation}\label{eq:unif_mse}
    \sum_{i=1}^K \mathrm{MSE}(i,\pi^{\mathrm{uni}}) = \sum_{i=1}^K \EE_{\pi^{\mathrm{uni}}}\bigl[(\hat\mu_i^{\mathrm{sam}}-\mu_i)^2\bigr] = \sum_{i=1}^K \frac{\sigma_i^2}{N/K} = \frac{K}{N}\sum_{i=1}^K \sigma_i^2.
\end{equation}

Using Cauchy-Schwarz inequality we can show \eqref{eq:Ney_mse} $\le$ \eqref{eq:unif_mse}, which is theoretically consistent as the Neyman allocation is optimal. Nonetheless, the gap between these two benchmarks highlights the main question of this section: 
\begin{center}
    \textit{When the variances are unknown, can one construct a data-driven sampling policy that improves upon uniform sampling and move toward the efficiency of Neyman allocation in finite sample even after accounting for adaptively estimated $\sigma_i$?}
\end{center} 
This naturally raises the possibility of adaptive sampling, where early observations are used to guide the remaining budget toward a more efficient allocation. As mentioned in the related literature, there have been works that show that asymptotically adaptive policies can be designed to approach Neyman Allocation (\cite{dai2023clip, zhao2023adaptiveneymanallocation}). However, to the best of our knowledge, no works exist that show when adaptive scheme can achieve better performance over uniform sampling in finite-samples. On one hand, we know if our estimates of $\sigma_i$ are sufficiently precise, the Neyman allocation should outperform the uniform allocation. On the other hand, our MSE suffers because of the randomness of estimating (potentially imprecise) $\sigma_i$ from an adaptive design. In this case, we study whether the benefit obtained from getting closer to Neyman allocation outweighs the price one pays for an adaptive design.

\subsection{Sampling Policy and Estimator Choice}\label{subsec:post_pilot_adaptive_pcipw}
Before presenting our main result, there are two technical issues to be addressed. The first concerns the sampling policy. When the arm variances are unknown, the Neyman allocation must be learned from the data, which naturally suggests sequential adaptation. However, under a fully adaptive policy, the resulting allocation process is difficult to analyze tractably in finite samples. To keep the MSE calculation tractable, we consider a specific post-pilot adaptive policy class. The experiment begins with a balanced pilot stage of size \(N_1\), where each arm receives $n_1 := N_1/K$ pilot observations.\footnote{We assume $N_1$ is divisible by $K$ for simplicity; the case where $N_1$ is not divisible generalizes straightforwardly.} The purpose of this stage is not to optimize the allocation, but to provide a non-adaptive source of information from which the policy can begin learning the variance structure. After this pilot stage, the policy may use any assignment probabilities that are predictable with respect to the observed history. That is, for every \(t>N_1\), the assignment probability $p^\pi_t=(p_{t,1},\ldots,p_{t,K})$ for policy $\pi$ is allowed to depend on the history \(\mathcal \history_{t-1}\). Let $\mathcal F_t:=\sigma(\mathcal H_t)$ denote
the sigma-algebra generated by the observed history. This post-pilot class contains a range of adaptive designs. If \(p_t\) is fixed immediately after the pilot, then the class contains the two-stage adaptive Neyman design analyzed in Section~\ref{subsec:twostage_closed_form}. If \(p_t\) is updated at a finite collection of update times, the class contains multi-stage adaptive designs. If \(p_t\) is recomputed after every observation, the class also contains fully sequential adaptive designs. 

% We remark that this technical simplification does not affect the conclusion of our results. Formally, if one  shows a specific adaptive sampling scheme outperforms the uniform sampling scheme, then there exists another (potentially more complex) adaptive scheme that can outperform the uniform. 

The second issue is the choice of estimator. In the oracle benchmark, the sample mean leads directly to the variance-driven Neyman objective. Under adaptive assignment, however, the clean~\eqref{eq:oracle_inf} no longer holds, and different estimators induce different finite-sample MSE decompositions. A natural first candidate is the Horvitz--Thompson estimator
\[
\hat\mu_i^{\mathrm{HT}} := \frac{1}{N} \sum_{t=1}^N \frac{\ind\{I_t=i\}X_t}{p_{t,i}}, \qquad i\in[K].
\]
However, as we show in Lemma~\ref{lem:mse_twostage_ht}, the corresponding MSE depends on the second moments \(\mu_i^2+\sigma_i^2\). This dependence is undesirable for our purpose because the Neyman benchmark is driven by the variances \(\sigma_i^2\), not by the second moments. We therefore use a pilot-centered inverse-propensity weighted estimator,
\begin{equation}\label{eq:pcipw_general}
    \hat\mu_i^{\mathrm{PCIPW}} = \hat\mu_i^{\mathrm{pilot}} + \frac{1}{N} \sum_{t=1}^{N} \frac{ \ind\{I_t=i\} \bigl(X_t-\hat\mu_i^{\mathrm{pilot}}\bigr) }{p_{t,i} },
\end{equation}
where the pilot mean is
\[
\hat\mu_i^{\mathrm{pilot}} := \frac{1}{n_1} \sum_{t=1}^{N_1} \ind\{I_t=i\}X_t.
\]

\begin{algorithm}[t]
\caption{General Post-Pilot Adaptive Policy with Pilot-Centered IPW Estimation}
\label{alg:post_pilot_general_pcipw}
\begin{algorithmic}[1]
\STATE \textbf{Input:} horizon \(N\), pilot budget \(N_1\) with \(K\mid N_1\), and post-pilot allocation rules \(\{\varphi_t\}_{t=N_1+1}^N\)
\STATE Set \(n_1\gets N_1/K\)
\STATE \textbf{Balanced pilot:} pull each arm \(i\in[K]\) exactly \(n_1\) times during rounds \(1,\ldots,N_1\)
\STATE Compute the pilot mean
\[
\hat\mu_i^{\mathrm{pilot}} := \frac{1}{n_1} \sum_{t=1}^{N_1} \ind\{I_t=i\}X_t, \quad i\in[K].
\]
\FOR{\(t=N_1+1,\ldots,N\)}
    \STATE Compute a predictable assignment vector
    \[
    p_t = \varphi_t(\mathcal \history_{t-1}), \quad p_{t,i}>0\ \text{for all }i\in[K].
    \]
    \STATE Sample \(I_t\sim p_t\) and observe reward \(X_t\)
\ENDFOR
\STATE For each arm \(i\in[K]\), output $\hat\mu_i^{\mathrm{PCIPW}}$ for all $i \in [K]$
\end{algorithmic}
\end{algorithm}

We summarize the general policy class and estimator choice in Algorithm~\ref{alg:post_pilot_general_pcipw}. The only structural requirement in Algorithm~\ref{alg:post_pilot_general_pcipw} is measurability of \(p_t\). They may be fixed after the pilot, updated on a small number of blocks, or updated sequentially after every observation. We next present a lemma that shows all such policies share the same MSE decomposition.

\begin{lemma}[Post-pilot PCIPW MSE decomposition]\label{lem:mse_twostage_pcipw}
Consider any policy in Algorithm~\ref{alg:post_pilot_general_pcipw}. Then, for each arm \(i\in[K]\),
\begin{align*}
\mathrm{MSE}(i,\pi) = \underbrace{ \frac{N_1}{N} \mathrm{MSE}(i,\pi^{\rm uni}) }_{\text{Pilot Stage MSE}} + \underbrace{ \frac{1}{N^2} \sum_{t=N_1+1}^{N} \EE\left[ \frac{\sigma_i^2}{p_{t,i}} \right] }_{\text{Adaptive variance}} + \underbrace{ \frac{1}{N^2} \sum_{t=N_1+1}^{N} \EE\left[ (\mu_i-\hat\mu_i^{\mathrm{pilot}})^2 \left( \frac{1}{p_{t,i}}-1 \right) \right] }_{\text{Interaction penalty}}.
\end{align*}
\end{lemma}

The proof of Lemma~\ref{lem:mse_twostage_pcipw}  is given in Appendix~\ref{sec:mse_pf}. Lemma~\ref{lem:mse_twostage_pcipw} separates the effect of adaptive allocation into three pieces. The first piece is the MSE one obtains from using a uniform sampling scheme to assign $n_1$ samples to each arm uniformly. The second term is the expected MSE price one pays for using an adaptive scheme $p_{t, i}$ that are random and data-dependent. The third interaction penalty is the cost of centering at an estimated pilot mean and then reweighting the centered residual by \(1/p_{t,i}\). This penalty is small when the pilot mean is accurate, but it can be amplified when the adaptive allocation assigns very small probability to an arm.

The decomposition gives a finite-sample identity for adaptive inference. Adaptive sampling is beneficial when the reduction in the variance term from moving toward Neyman allocation is large enough to dominate the interaction penalty induced by learning the allocation from data. We now rewrite this tradeoff as a direct comparison with uniform sampling.

\begin{theorem}[Adaptive improvement decomposition]
\label{thm:general_inf_improvement}
Consider any policy $\pi$ in Algorithm~\ref{alg:post_pilot_general_pcipw}, and set $N_2:=N-N_1,$ $ S:=\sum_{i=1}^K\sigma_i,$ and $V:=\sum_{i=1}^K\sigma_i^2.$
For each post-pilot round \(t>N_1\), define the allocation-error and the interaction term to be
\[
\mathcal D_t := \sum_{i=1}^K \frac{(\sigma_i-Sp_{t,i})^2}{p_{t,i}}, \quad \mathcal B_t := \sum_{i=1}^K (\mu_i-\hat\mu_i^{\mathrm{pilot}})^2 \left( \frac{1}{p_{t,i}}-1 \right).
\]
Then
\[
\sum_{i=1}^K
\mathrm{MSE}(i,\pi)
- \sum_{i=1}^K
\mathrm{MSE}(i,\pi^{\rm uni})
=
-\frac{N_2}{N^2}(KV-S^2) + \frac{1}{N^2} \sum_{t=N_1+1}^{N} \EE[\mathcal D_t+\mathcal B_t].
\]
Consequently, the adaptive policy $\pi$ weakly improves over uniform sampling if and only if
\[
\frac{1}{N_2} \sum_{t=N_1+1}^{N} \EE[\mathcal D_t+\mathcal B_t] \le KV-S^2.
\]
\end{theorem}

The proof of Theorem~\ref{thm:general_inf_improvement} is in Appendix~\ref{app:pf_general_inf_improvement}. Theorem~\ref{thm:general_inf_improvement} gives us an intuitive understanding toward the benefits and costs of adaptive sampling for inference objective. The term \(KV-S^2\) is the amount of variance reduction available from moving from uniform allocation to oracle Neyman allocation. It is large when the standard deviations are heterogeneous. The terms \(\mathcal D_t\) and \(\mathcal B_t\) are the finite-sample costs of using a data-driven adaptive allocation: \(\mathcal D_t\) measures how far the realized assignment probabilities are from oracle Neyman allocation, while \(\mathcal B_t\) measures the interaction between pilot mean error and inverse-propensity reweighting.

Theorem~\ref{thm:general_inf_improvement} is deliberately stated for a general post-pilot adaptive policy. In the next subsection, however, we analyze a specific two-stage plug-in Neyman adaptive scheme, for which the expectation of \(\mathcal D_t+\mathcal B_t\) can be computed closed-form under Gaussian rewards.

\subsection{Case Study: Two-Stage Adaptive Neyman} \label{subsec:twostage_closed_form}

The term two-stage adaptive Neyman policy is termed by~\cite{zhao2023adaptiveneymanallocation}, and can be obtained from Algorithm~\ref{alg:post_pilot_general_pcipw}. Formally, we estimate the Neyman allocation after the pilot stage via
\[
\hat p_i^{\mathrm{Ney}} := \frac{\hat\sigma_i}{\sum_{j=1}^K\hat\sigma_j}, \text{ where } \hat\sigma_i^2 := \frac{1}{n_1-1} \sum_{t:I_t=i,\;t\le N_1} (X_t-\hat\mu_i^{\mathrm{pilot}})^2, \quad i\in[K].
\]
as the post-pilot assignment probability that is fixed after the pilot and $\hat \sigma_i$ is the typical sample standard deviation estimator using $n_1$ samples. The resulting two-stage adaptive Neyman policy, denoted as \(\pi^{\rm AN}\) then sets
\[
p_{t,i}^{\pi^{\rm AN}} =
\begin{cases}
1/K, & t\le N_1,\\
\hat p_i^{\mathrm{Ney}}, & t>N_1.
\end{cases}
\]
Equivalently, in the notation of Algorithm~\ref{alg:post_pilot_general_pcipw}, we are setting $\varphi_t(\mathcal \history_{t-1}) = \hat p^{\mathrm{Ney}},$ for all $ t > N_1.$ Although, for theoretical tractability, this does not utilize the full power of adaptivity, we view this approach conservative to the fully general adaptive approach. In other words, if a single adaptive update can already improve over uniform sampling, then more fully adaptive variants may plausibly yield further gains as long as the interaction penalty is not too large. Furthermore, our paper aims to show how even practically simple adaptive schemes can still outperform the uniform sampling in finite-samples as opposed to asymptotic rates. We now show a closed form condition for the improvement upon uniform allocation for Two-stage Adaptive Neyman Policy.

\begin{prop}[Exact condition over uniform sampling]\label{prop:inf_suff}
Assume that for each arm \(i\in[K]\), rewards are i.i.d. normal with mean \(\mu_i\) and variance \(\sigma_i^2\). Following the setting in Theorem~\ref{thm:general_inf_improvement} and denote $\nu := n_1 - 1 = \frac{N_1}{K}-1.$ Let
\[
\beta_\nu := \EE[\sqrt{F_{\nu,\nu}}] = \frac{ \Gamma\left(\frac{\nu+1}{2}\right)\Gamma\left(\frac{\nu-1}{2}\right)}{ \Gamma\left(\frac{\nu}{2}\right)^2
}.
\]
Then the total MSE of the two-stage pilot-centered IPW estimator satisfies
\begin{equation}
\label{eq:refined_condition_beta}
\sum_{i=1}^K \mathrm{MSE}(i,\pi^{\rm AN}) \le \sum_{i=1}^K \mathrm{MSE}(i,\pi^{\rm uni}) \iff \beta_\nu \le \frac{(K-1)V}{S^2-V} \cdot \frac{1}{1+\frac{K}{N_1}}.
\end{equation}
\end{prop}

The full proof of Proposition~\ref{prop:inf_suff} is in Appendix~\ref{app:pf_inf_suff}. Because the left-hand side of the condition is difficult to interpret directly, we apply Jensen's inequality to obtain the following more transparent sufficient condition.

\begin{corollary}
\label{cor:pilot_size_condition}
Under the setup of Proposition~\ref{prop:inf_suff}, a sufficient condition for the two-stage adaptive Neyman allocation $\pi^{\rm AN}$ to have total MSE no greater than that of uniform sampling $\pi^{\rm uni}$ is
\[
1 < \left(1+\frac{K}{N_1}\right) \sqrt{\frac{N_1-K}{N_1-3K}} \le \frac{(K-1)V}{S^2-V}, \qquad 3K < N_1 \le N.\]
\end{corollary}

Proposition~\ref{prop:inf_suff} and Corollary~\ref{cor:pilot_size_condition} gives a clear intuition for when adaptive Neyman allocation can improve upon uniform sampling. The pilot size $N_1$ measures how much information is collected before the second-stage allocation is chosen, capturing the strength of the learning phase. The terms $V$ and $S^2-V$ summarize the variance structure across arms, and measure how heterogeneous the arm variances are. When all arms have similar variances, the right-hand side is close to its smallest possible value, making the inequality harder to satisfy for a fixed pilot size. Indeed, when variances are similar, there is little reason to allocate them differently as uniform sampling is already optimal. Therefore, an adaptive Neyman allocation has little advantage over uniform sampling. By contrast, when the variances differ substantially across arms, there is more to gain from reallocating effort toward the high-variance arms, and the sufficient condition becomes easier to satisfy; in that case the right-hand side becomes larger, so the same left-hand side is more likely to fall below it. 

Lastly, for any fixed instance with arm variances that are not all equal, this condition is satisfied for sufficiently large $N_1$, so the two-stage adaptive Neyman policy outperforms uniform sampling, i.e., if practitioners have large amount of samples, they should prefer adaptive sampling even for inference. Corollary~\ref{cor:pilot_size_condition} shows that the pilot-dependent term converges to one as $N_1 \rightarrow \infty$, while the right-hand side is strictly greater than one in this case. We empirically show later in Section~\ref{subsection:inference_sim} how this two-stage adaptive sampling outperforms the uniform sampling even under moderate sample sizes.

\section{Adaptive Design for Inference and Regret}\label{sec:infreg}
So far, our focus has been purely inferential. In many experiments, especially adaptive ones, however, estimation accuracy is not the only objective. While the experiment is running, the experimenter may also care about \emph{online performance}, namely,  avoiding real customers from being allocated to inferior versions of the firm's product. This naturally leads to a joint inference--regret problem. For simplicity, in this section, we will use the sample mean estimator for the inference objective. Changing the sample mean to other estimators likely does not affect this section's main asymptotic rate result, since the main asymptotic result is governed by the sampling frequencies rather than by the particular bias-corrected form of the estimator. 

% We therefore work with the sample mean, in contrast to Pilot-Centered IPW, which would add technical complexity without changing the main message. %This question has also being tackled by \cite{ErraqabiLVBL17} and \cite{simchi-levi_multi-armed_2025}, yet both work proposed bandit algorithm that remains non-intuitive to practitioners. 

To formalize this problem, Let $i^\star \in \arg\max_{i \in [K]} \mu_i$ denote an optimal arm, and for each arm $i$ define the suboptimality gap $\Delta_i := \mu_{i^\star} - \mu_i$. Throughout this section, we will assume for simplicity that the optimal arm is
unique, so that \(\Delta_i > 0\) for all \(i \neq i^\star\). As typically done in bandit literature, the cumulative and average regret up to horizon $N$ is defined as 
\[
R_N := \sum_{i=1}^K \Delta_i T_i, \quad \bar R_N := \frac{1}{N} \sum_{i=1}^K \Delta_i T_i,
\] 
which measures the performance loss induced by allocating samples to suboptimal arms. This is the usual pseudo-regret formulation in stochastic bandits: rather than comparing realized rewards pathwise through the random observations $X_t$, it compares the expected reward of the chosen arm to that of an optimal arm through the population means $\mu_i$. This makes the criterion depend on the sampling design itself, rather than on additional outcome noise, and is therefore the natural quantity to combine with an inference objective. 

We now measure performance through a single objective that balances statistical accuracy and online performance. Fix a weight \(\lambda \in (0,1)\) that determines the two objectives relative importance. For the inferential term, we use \(\sqrt{\mathrm{MSE}}\), which is the Root MSE (RMSE), rather than MSE itself, since the regret term grows linearly whereas MSE is quadratic in the estimation error. This puts the two terms on a more comparable scale. Our objective $\mathcal J_N$, therefore, is 
\begin{equation}\label{eq:joint_obj_main}
\mathcal J_N(\pi) := \lambda \sum_{i=1}^K \sqrt{\mathrm{MSE}(i,\pi)} + (1-\lambda) \EE_\pi[\bar R_N],
\end{equation}

Equivalently, writing out the mean squared error term explicitly, we have
\begin{equation*}\label{eq:joint_obj_main_expanded}
\mathcal J_N(\pi) = \lambda \sum_{i=1}^K \sqrt{\EE_\pi\bigl[(\hat\mu_i-\mu_i)^2\bigr]} + (1-\lambda) \EE_\pi \left[\frac{1}{N}\sum_{i=1}^K \Delta_i T_i\right].
\end{equation*}
Our goal in this section is to analyze $ \mathcal J_N(\pi)$ for different policies $\pi$. We remark, as mentioned in the introduction, these objectives have tradeoff. On one hand, the optimal oracle regret is to pick the best arm every round to incur zero regret. Unfortunately, this leads to suboptimal inference as we learn nothing about any other arm. On the other hand, precisely estimating all arm means requires sampling every arm, including suboptimal ones, which incurs regret. In this case, the question remains which policy could be optimal to balance both managerially important objectives.

Our objective is closely related to that of
\citet{ErraqabiLVBL17}, but differs in its normalization with respect to the horizon. They combine average reward with an estimation criterion equal to \(\sqrt{N}\) times the average armwise RMSE, so that, for a fixed tradeoff weight, the resulting target allocation fractions are independent of \(N\). We instead combine the unscaled final armwise RMSE with average pseudo-regret. Consequently, our oracle problem is horizon-dependent: as the total budget grows, a given allocation fraction yields increasingly precise estimates, allowing the optimal design to place more mass on the best arm while retaining a vanishing but sufficient fraction of pulls for the suboptimal arms. \citet{simchi-levi_multi-armed_2025} study a different multi-objective formulation. They evaluate policy--estimator pairs using cumulative regret and the worst pairwise gap-estimation error over a class of bandit instances, and characterize the Pareto frontier at the level of rates in the horizon \(N\), while also providing finite-horizon guarantees. By contrast, we study an instance-specific, horizon-dependent linear scalarization of summed armwise RMSE and average pseudo-regret for a fixed preference weight \(\lambda\) to provide guidance to practitioners on how to set exact allocations across arms.

\subsection{Oracle Benchmark}\label{subsec:oracle_joint}

We start with the oracle fixed-allocation problem, with both $\sigma$ and $\Delta$ fully known. Let $\mathcal S_K := \Bigl\{p\in\mathbb R_+^K : \sum_{i=1}^K p_i = 1\Bigr\}$ denote the simplex of allocation vectors. For \(p=(p_1,\dots,p_K)\in\mathcal S_K\), we interpret \(p_i\) as the target fraction of the horizon assigned to arm \(i\), so that \(T_i = N p_i\). Under such a static allocation, the joint objective~\eqref{eq:joint_obj_main} becomes
\begin{equation}\label{eq:fixedalloc_obj_joint_clean} 
\mathcal J_N(p) = \lambda \sum_{i=1}^K \frac{\sigma_i}{\sqrt{N p_i}} + (1-\lambda)\sum_{i=1}^K p_i \Delta_i, \quad p \in \mathcal S_K.
\end{equation}
The first term is the oracle RMSE under fixed counts, while the second is the average regret induced by allocating a fraction \(p_i\) of the horizon to arm \(i\). The next lemma characterizes the oracle optimizer of~\eqref{eq:fixedalloc_obj_joint_clean}, serving as a benchmark throughout the rest of the section.

\begin{lemma}[Oracle fixed-allocation benchmark]\label{lem:oracle_fixedalloc_joint_clean}
Assume \(\lambda\in(0,1)\), \(\sigma_i>0\) for all \(i\in[K]\), and a unique optimal arm \(i^\star\). Then the objective in \eqref{eq:fixedalloc_obj_joint_clean} admits a unique minimizer \(p^\star(N)\in\mathcal S_K\). It is characterized by
\begin{equation}\label{eq:q_star_closed_clean}
p_i^\star(N) = \left[ \frac{\lambda \sigma_i} {2\sqrt{N}\bigl((1-\lambda)\Delta_i-\gamma_N^\star\bigr)} \right]^{2/3}, \quad i\in[K],
\end{equation}
where the multiplier \(\gamma_N^\star\) is chosen so that \(\sum_{i=1}^K p_i^\star(N)=1\) and all denominators are positive. Moreover, as \(N\to\infty\),
\begin{equation}\label{eq:q_star_asymp_clean}
p_{i^\star}^\star(N)=1-\Theta(N^{-1/3}), \quad p_i^\star(N) = \left( \frac{\lambda \sigma_i}{2(1-\lambda)\Delta_i} \right)^{2/3} N^{-1/3}(1+o(1)), \quad i\neq i^\star.
\end{equation}
Consequently,
\begin{equation}\label{eq:oracle_opt_rate_clean}
\mathcal J_N\bigl(p^\star(N)\bigr)=\Theta(N^{-1/3}).
\end{equation}
\end{lemma}

The proof for Lemma~\ref{lem:oracle_fixedalloc_joint_clean} is shown in Section~\ref{subsec:oracle_joint_pf}. Lemma~\ref{lem:oracle_fixedalloc_joint_clean} gives the best benchmark one could hope for: even with full knowledge of the instance ($\sigma_i, \Delta_i$), the optimal joint tradeoff decays only at rate \(N^{-1/3}\). This rate therefore becomes the natural target for any adaptive procedure. The structure of the oracle solution is also intuitive. Since regret pushes mass toward the best arm, while inference requires each arm to be sampled, the optimizer places almost all of its mass on the best arm, but still allocates a vanishing fraction of order \(N^{-1/3}\), equivalently $N^{2/3}$ pulls, to each suboptimal arm. The theorem shows this is the correct amount to keep the inferential cost at an optimal minimum. 

\subsection{PORP: Plug-in Oracle Rate Policy}\label{sec:porp}

The oracle benchmark in Lemma~\ref{lem:oracle_fixedalloc_joint_clean} suggests a direct data-driven policy: estimate the unknown means and standard deviations adaptively, solve the empirical version of the oracle fixed-allocation problem, and sample according to the resulting plug-in allocation. We follow this strategy and call this adaptive scheme the plug-in oracle rate policy (PORP). We present Algorithm~\ref{alg:porp} and its theoretical result in Theorem~\ref{thm:porp}, which shows the rate optimal, ie. $\Theta(N^{-1/3})$, of this algorithm.

\begin{theorem} \label{thm:porp}
Consider a stochastic $K$-armed bandit with a unique optimal arm $i^\star$ and suboptimality gaps $ \Delta_i:=\mu_{i^\star}-\mu_i>0,$ for $ i\neq i^\star.$ Assume that $\sigma_i^2>0$ for every $i\in[K]$, and that each centered reward $X_{i,s}-\mu_i$ is $\bar\nu$-sub-Gaussian. Fix
$\lambda\in(0,1)$, a warm-start size $m_0\ge 2$, and a forced-exploration
constant $\alpha>0$.

Let $\pi^{\rm PORP}$ denote Algorithm~\ref{alg:porp}. Then there exist constants
$C<\infty$ and $N_0<\infty$, depending only on the bandit instance and on
$(K,\bar\nu,\lambda,m_0,\alpha)$, such that
\[
\mathcal J_N(\pi^{\rm PORP})\le C N^{-1/3}, \quad \forall N\ge N_0  \implies \mathcal J_N(\pi^{\rm PORP})=O(N^{-1/3}),
\]
so PORP attains the same asymptotic rate as the oracle fixed-allocation
benchmark in Lemma~\ref{lem:oracle_fixedalloc_joint_clean}.
\end{theorem}

\begin{algorithm}[ht!]
\caption{Plug-in Oracle Rate Policy (PORP)} 
\label{alg:porp}
\begin{algorithmic}[1]
\STATE \textbf{Input:} warm-start size $m_0\ge 2$, objective weight
$\lambda\in(0,1)$, forced-exploration constant $\alpha>0$
\STATE In the first $m_0K$ rounds, pull each arm exactly $m_0$ times

\FOR{$t=m_0K+1,m_0K+2,\dots$}
    \STATE Let $U_t\in\arg\min_{i\in[K]}T_i(t-1)$ be a least-sampled arm
    \STATE Define the minimum-sampling quota $g(t):=m_0+\left\lceil \alpha\sqrt t\right\rceil .$

    \IF{$T_{U_t}(t-1)<g(t)$}
        \STATE \textbf{Forced exploration:} pull $I_t=U_t$
    \ELSE
                \STATE Compute empirical means $\hat\mu_{i,t-1}$ and empirical standard
        deviations $\hat\sigma_{i,t-1}$ for all $i\in[K]$.

        \STATE Define the empirical best arm and empirical gaps
        \[
        \hat i^\star_{t-1}\in\arg\max_{i\in[K]}\hat\mu_{i,t-1},
        \qquad
        \hat\Delta_{i,t-1}
        :=
        \bigl(
        \hat\mu_{\hat i^\star_{t-1},t-1}
        -
        \hat\mu_{i,t-1}
        \bigr)_+ .
        \]

        \IF{$\min_{i\in[K]}\hat\sigma_{i,t-1}=0$}
            \STATE Set $\hat p_{i,t}=1/K$ for all $i\in[K]$.
        \ELSE
            \STATE Find the unique multiplier $
            \hat\gamma_t$
            satisfying the empirical form of~\eqref{eq:q_star_closed_clean} and set
            \[
            \hat p_{i,t} := \left[ \frac{\lambda\hat\sigma_{i,t-1}} {2\sqrt t\bigl((1-\lambda)\hat\Delta_{i,t-1}-\hat\gamma_t\bigr)}   \right]^{2/3}, \quad i\in[K].
            \]
        \ENDIF

        \STATE Pull $I_t\sim \hat p_t$.
    \ENDIF
\ENDFOR
\end{algorithmic}
\end{algorithm}

Algorithm~\ref{alg:porp} shares the same basic idea of finding the optimal allocation as in \citet{ErraqabiLVBL17}'s \texttt{ForcingBalance} algorithm. In both algorithms, every arm is sampled often enough to ensure estimates of means or variances are stable and avoid self-confirming. The difference is what happens after this stabilization step. In \texttt{ForcingBalance}, forced sampling is part of a target-tracking procedure: once the minimum-sampling requirement is met, the algorithm computes an estimated optimal allocation for its objective and pulls arms to keep the realized allocation close to that estimated target. PORP instead uses the forced quota only to make the plug-in estimates reliable. After the quota is satisfied, it samples from the current plug-in solution to our horizon-dependent oracle problem. This distinction reflects the different structure of the two objectives: \texttt{ForcingBalance} tracks a static allocation target, whereas in our objective the oracle allocation changes with the horizon and assigns suboptimal arms only order $N^{2/3}$ pulls by time $N$. Thus, PORP uses forcing as a stabilization device, while the substantive exploration schedule is determined by the oracle rate implied by our inference--regret objective.

We evaluate the performance and computational cost of Algorithm~\ref{alg:porp} in Section~\ref{sec:simulation}. PORP directly implements the plug-in allocation principle but, like the approach of \citet{ErraqabiLVBL17}, relies on repeated numerical allocation updates. Furthermore, the computational burden is exacerbated by having to solve for the allocation multiplier $\gamma^*_N$ at every time point. Lastly, the resulting $\hat p_{i,t}$ does not explicitly separate inference-oriented exploration from regret-oriented exploitation. This not only makes the algorithm harder for managers to interpret but also can not utilize existing regret-minimizing tools. These considerations motivate our next two methods, SARP and NARP, which combine the two components through closed-form mixing rules, which allows practitioners to reuse existing regret-minimizing algorithms, and avoids repeated numerical allocation.

\subsection{SARP: Static-Allocation Rate Policy}

Based on the preceding discussion, we ask a practical question:

\begin{quote}
\textit{Can one design an adaptive policy that is easy for practitioners to implement using existing tools and explain to stakeholders, while still attaining the oracle benchmark's \(N^{-1/3}\) rate for the joint inference--regret objective?}
\end{quote}

The answer is yes. The key idea is to separate the two roles of the policy: exploration across all arms supports inference, while an existing regret-minimizing algorithm, such as Thompson sampling (\cite{pmlr-v23-agrawal12}), controls regret. At round \(t\), the policy explores with probability of order \(t^{-1/3}\) and follows the regret-minimizing algorithm otherwise. This separation makes the role of each component explicit, while the closed-form mixing probability allows practitioners to reuse existing algorithms without solving an additional optimization problem. This modular structure addresses the implementation and interpretability goals outlined in the introduction: practitioners can separately verify and maintain the added exploration and mixing logic, while explaining to stakeholders how the two components balance inference and regret. To establish guarantees for this construction, we impose the following expected-regret condition on the exploitation algorithm:

\begin{definition}[Expected regret control]
\label{def:ALG-Exp-control-clean}
We say an exploitation algorithm \(\ALG\) satisfies expected regret control if there exists a constant \(C_{\ALG}>0\) such that, for every horizon \(n\ge 1\),
\begin{equation}\label{eq:ALG-Exp-regret-clean}
\EE \left[R_n^{\ALG}\right] \le C_{\ALG}\sqrt{K n \log n}.
\end{equation}
\end{definition}

This is a standard regret guarantee, satisfied by many familiar bandit algorithms, including UCB or Thompson Sampling (\cite{LAI19854,pmlr-v23-agrawal12, lattimore2020bandit}). The advantage of Definition~\ref{def:ALG-Exp-control-clean} is that the theorem below does not depend on the fine details of the exploitation rule; any algorithm with a bound of the form~\eqref{eq:ALG-Exp-regret-clean} can be leveraged. We are now ready to state the main result of this subsection, introducing the Static Allocation Rate Policy (SARP) in Algorithm~\ref{alg:sarp} and its main theorem guarantee.

\begin{algorithm}[ht!]
\caption{Static-Allocation Rate Policy (SARP)}
\label{alg:sarp}
\begin{algorithmic}[1]
\STATE \textbf{Input:} Constant $c_x>0$, an exploitation algorithm $\ALG$, static exploration distribution $p^0 \in \mathcal{S}_K$ with $p_i^0 > 0, \forall i \in [K]$
\STATE In the first $K$ rounds, pull each arm exactly once
\FOR{$t = K+1, K+2, \dots$}
    \STATE Define the exploration probability $x_t := \min\{1, c_x t^{-1/3}\}$
    \STATE With probability $x_t$, \textbf{explore} by drawing arm $I_t$ from $p^0$
    \STATE With probability $1-x_t$, \textbf{exploit} by following the algorithm $\ALG$
\ENDFOR
\end{algorithmic}
\end{algorithm}

\begin{theorem}[SARP]\label{thm:joint_sarp_rate}
Consider a stochastic \(K\)-armed bandit with a unique optimal arm \(i^\star\) and suboptimality gaps $\Delta_i := \mu_{i^\star}-\mu_i > 0$ for $i\neq i^\star$. Assume that for each arm $i\in[K]$, the centered rewards $X_{i,s}-\mu_i$ are $\bar\nu$-sub-Gaussian. Fix \(\lambda \in (0,1)\), an exploration constant \(c_x>0\), and a full-support exploration distribution \(p^0=(p_1^0,\dots,p_K^0)\in\mathcal S_K\) with \(p_i^0>0\) for all \(i\in[K]\). 

Let \(\pi^{\rm SARP}\) denote the policy executed by Algorithm~\ref{alg:sarp}. Assume $\ALG$ satisfies Definition~\ref{def:ALG-Exp-control-clean} under the intermittent execution induced by Algorithm~\ref{alg:sarp}. Then there exist constants \(C<\infty\) and \(N_0<\infty\), such that
\begin{equation}\label{eq:sarp_rate}
\mathcal J_N(\pi^{\rm SARP}) \le C N^{-1/3}, \quad \forall N\ge N_0.
\end{equation}
Equivalently, $ \mathcal J_N(\pi^{\rm SARP}) = O(N^{-1/3}),$ meaning SARP is asymptotically rate optimal, i.e., it achieves the same asymptotic rate as the oracle benchmark in Lemma~\ref{lem:oracle_fixedalloc_joint_clean}.
\end{theorem}

The theorem is deliberately general so that managers have the flexibility to choose $p^0$ and $\ALG$ based on engineering/implementation constraints and stakeholder demands. In line with the goals outlined in the introduction, practitioners can, for example, combine uniform exploration with a familiar regret-minimizing algorithm such as Thompson sampling and still attain the fixed-allocation oracle's $N^{-1/3}$ rate under the stated assumptions. This allows managers to clearly explain what the algorithm is doing: ``often doing uniform sampling for inferences and often doing Thompson sampling for regret minimization''. Theorem~\ref{thm:joint_sarp_rate} then quantifies exactly how ``often'' the algorithm needs to mix the two while maintain asymptotic optimality. Throughout the numerical simulations and empirical applications, unless otherwise specified, we implement SARP with uniform sampling as its baseline exploration distribution $p^0$.

\subsection{NARP: Neyman-Adaptive Rate Policy}

Theorem~\ref{thm:joint_sarp_rate} is useful as a conceptual benchmark: it shows that the oracle rate is attainable by a very simple and broadly applicable adaptive strategy. At the same time, the theorem is intentionally coarse. The exploration distribution \(p^0\) is fixed in advance, and the result only guarantees the correct asymptotic order. It does not attempt to improve upon constants, i.e., finite-sample performance, nor does it adapt to the instance structure learned during the experiment.

%\blue{In simulation NARP doesn't always improve in finite sample. I'm 80 percent sure if $N$ is large enough it will always outperform but it may not kick in reasonable sample size. We need to rephrase this probably. Also it seems like the exploration part, ie. lower RMSE, is what NARP helps perform best. Rewrite to make this line clear.}

On the other hand, the oracle benchmark in Lemma~\ref{lem:oracle_fixedalloc_joint_clean} suggests that the best allocation should not explore all arms under an arbitrary fixed probability. Rather, it should tilt its inferential effort toward high-variance arms while still favoring high-mean arms to control regret. Therefore, in this section we present an alternative algorithm that gives managers a specific choice of an exploration algorithm that performs better than an arbitrary uniform sampling scheme in finite samples, evidenced later in Section~\ref{sec:simulation}. Our proposed policy interpolates between the two natural extremes:
\begin{itemize}
    \item a regret-oriented policy, such as Thompson Sampling or algorithms that fit Definition~\ref{def:ALG-Exp-control-clean}; and
    \item an inference-oriented policy, such as the Neyman Allocation from the previous section.
\end{itemize}

Hence we propose a policy that refines the simple SARP construction. NARP retains SARP's exploration--exploitation structure but calibrates the exploration intensity and distribution using estimated variances and gaps. It thereby provides an instance-adaptive mixing rule without PORP's repeated numerical allocation solve.

Before stating the algorithm, it is helpful to build intuition from an oracle interpolation between two extreme allocations: the rooted-Neyman allocation \(p^{\rm rN}\), which optimizes statistical efficiency, and the oracle best-arm vector \(e_{i^\star}\), which places all mass on the optimal arm. For \(x\in[0,1]\), define
\[
p^{(x)} := (1-x)e_{i^\star} + x p^{\rm rN}, \quad p_i^{\rm rN} := \frac{\sigma_i^{2/3}}{\sum_j^k \sigma_j^{2/3}}.
\]
Along this path, the parameter \(x\) controls only the \emph{amount} of exploration, while the exploration \emph{shape} is fixed by \(p^{\rm rN}\). From the oracle MSE calculation in Section~\ref{subsec:ora_ney}, we can do similar calculation to show the rooted-Neyman rule allocates mass proportionally to \(\sigma_i^{2/3}\) for RMSE optimization.

We now examine how the joint objective~\eqref{eq:joint_obj_main} behaves along this one-dimensional path. For every suboptimal arm \(i\neq i^\star\), we have \(p_i^{(x)}=x p_i^{\rm rN}\), so the regret contribution is linear in \(x\):
\[
(1-\lambda)\sum_{i=1}^K p_i^{(x)}\Delta_i = (1-\lambda)x\sum_{i\neq i^\star}\Delta_i p_i^{\rm rN} = (1-\lambda)Lx, \quad L:=\sum_{i\neq i^\star}\Delta_i p_i^{\rm rN}.
\]
At the same time, the dominant suboptimal-arm RMSE contribution is of order
\[
\frac{\lambda}{\sqrt N}\sum_{i\neq i^\star}\frac{\sigma_i}{\sqrt{x p_i^{\rm rN}}} = \frac{\lambda M}{\sqrt N}x^{-1/2}, \quad M:=\sum_{i\neq i^\star}\sigma_i (p_i^{\rm rN})^{-1/2}.
\]
Thus, after fixing the rooted-Neyman exploration profile, the remaining one-dimensional tradeoff is roughly
\[
\frac{\lambda M}{\sqrt N}x^{-1/2} + (1-\lambda)Lx.
\]
Optimizing this expression yields an exploration level of order
\[
x_N^\star \asymp \Bigl(\frac{\lambda M}{2(1-\lambda)L}\Bigr)^{2/3}N^{-1/3}.
\]
This directly motivates the exploration rule used in the algorithm: we replace the unknown oracle quantities \(M\) and \(L\) by their plug-in estimates \(\hat M_{t-1}\) and \(\hat L_{t-1}\), and uses the same \(t^{-1/3}\) scaling. To implement this scaling online without requiring the final
horizon in advance, we use the current round $t$ in place of $N$. We are now ready to present Neyman-Adaptive Rate Policy (NARP) in Algorithm~\ref{alg:narp}, together with its theoretical guarantee in Theorem~\ref{thm:NARP}.

\begin{algorithm}[ht!]
\caption{Neyman-Adaptive Rate Policy (NARP)}
\label{alg:narp}
\begin{algorithmic}[1]
\STATE \textbf{Input:} warm-start size $m_0\ge 2$, mixing weight $\lambda\in(0,1)$, forced-exploration constant $\alpha>0$
\STATE In the first $m_0K$ rounds, pull each arm exactly $m_0$ times

\FOR{$t = m_0K+1, m_0K+2, \dots$}
    \STATE Let $U_t \in \arg\min_{i\in[K]} T_i(t-1)$ be a least-sampled arm
    \STATE Define the minimum-sampling quota $g(t):= m_0 + \left\lceil \alpha \sqrt t \right\rceil$

    \IF{$T_{U_t}(t-1) < g(t)$}
        \STATE \textbf{Forced exploration:} pull arm $I_t = U_t$
    \ELSE
        \STATE Compute empirical means $\hat\mu_{i,t-1}$ and empirical standard deviations $\hat\sigma_{i,t-1}$ for all $i\in[K]$

        \STATE Define the estimated best arm, the plug-in gaps, rooted-Neyman allocation for all $i \in [K]$
        \[
        \hat i^\star_{t-1}\in\arg\max_{i\in[K]} \hat\mu_{i,t-1}, \quad \hat\Delta_{i,t-1} := \bigl(\hat\mu_{\hat i^\star_{t-1},t-1}-\hat\mu_{i,t-1}\bigr)_+, \quad \hat p_{i,t-1}^{\rm rN} := \frac{\hat\sigma_{i,t-1}^{2/3}}{\sum_{k=1}^K \hat\sigma_{k,t-1}^{2/3}}
        \]

        \STATE Define the plug-in functionals
        \[
        \hat M_{t-1} := \sum_{i\neq \hat i^\star_{t-1}} \hat\sigma_{i,t-1}\bigl(\hat p_{i,t-1}^{\rm rN}\bigr)^{-1/2}, \quad \hat L_{t-1} := \sum_{i\neq \hat i^\star_{t-1}} \hat\Delta_{i,t-1}\hat p_{i,t-1}^{\rm rN}
        \]

        \STATE Define
        \[
        x_t := \min\Biggl\{ 1, \Bigl( \frac{\lambda \hat M_{t-1}}{2(1-\lambda)\hat L_{t-1}} \Bigr)^{2/3} t^{-1/3} \Biggr\}
        \]
        with the convention $x_t = 1$ if $\hat L_{t-1}=0$

        \STATE With probability $x_t$, explore by drawing $A_t$ from $\hat p_{t-1}^{\rm rN}$
        \STATE With probability $1-x_t$, exploit by following the algorithm $\ALG$
    \ENDIF
\ENDFOR
\end{algorithmic}
\end{algorithm}

\begin{theorem}[NARP]\label{thm:NARP}
Consider a stochastic $K$-armed bandit with a unique optimal arm $i^\star$ and suboptimality gaps $\Delta_i := \mu_{i^\star}-\mu_i > 0,$ for all $i\neq i^\star.$ Assume that for each arm $i\in[K]$, the centered rewards $X_{i,s}-\mu_i$ are $\bar\nu$-sub-Gaussian. Fix $\lambda\in(0,1)$, an integer warm-start size $m_0\ge 2$, and a forced-exploration constant $\alpha>0$.

We denote Algorithm~\ref{alg:narp} by $\pi^{\rm NARP}$. Assume $\ALG$ satisfies Definition~\ref{def:ALG-Exp-control-clean} under the intermittent execution induced by Algorithm~\ref{alg:narp}. Then there exist constants $C<\infty$ and $N_0> m_0K$, \footnote{These constants depend only on the bandit and algorithm instance: $\{K, \bar\nu, \lambda, m_0, \alpha\}$ and on the regret constant of $\ALG$.} , such that
\[
\mathcal J_N\bigl(\pi^{\rm NARP}\bigr)\le C N^{-1/3}, \qquad \forall\, N\ge N_0.
\]
Equivalently, $ \mathcal J_N(\pi^{\rm NARP}) = O(N^{-1/3}),$ meaning NARP is asymptotically rate optimal, i.e., it achieves the same asymptotic rate as the oracle benchmark in Lemma~\ref{lem:oracle_fixedalloc_joint_clean}.
\end{theorem}

Although Theorem~\ref{thm:NARP} establishes the same \(N^{-1/3}\) rate as SARP, the plug-in quantities \(\hat M_{t-1}\) and \(\hat L_{t-1}\) are the online analogues of the two oracle coefficients that govern the interpolation between inference-driven and regret-driven allocation. Thus, NARP should be viewed as a practical and interpretable refinement of SARP rather than as a claim of global optimality among all adaptive bandit procedures (\cite{ErraqabiLVBL17, simchi-levi_multi-armed_2025}). Comparing Algorithm~\ref{alg:sarp} and Algorithm~\ref{alg:narp}, the key distinction is that both use an exploration rate of the correct order \(t^{-1/3}\), but SARP leaves the leading constant unspecified, taking the form \(c_x t^{-1/3}\), whereas NARP adaptively calibrates this constant from the data. In addition, SARP permits any fixed exploration distribution \(p^0\), while NARP uses a plug-in rooted-Neyman allocation, which targets the lower-RMSE side of the inference--regret tradeoff by assigning exploratory pulls according to the estimated variance structure. Section~\ref{sec:simulation} examines when this additional calibration improves the joint objective relative to SARP and how its performance compares with PORP and \texttt{ForcingBalance}.

% This finer oracle-calibrated behavior is not captured in the theorem statement, but it provides the main heuristic for why NARP can improve over SARP in finite-sample simulations for certain instances.

%\blue{We show the overperformance of NARP upon SARP in the simulation. Next, we also need to comment on why exactly we want the $\hat M, \hat L$. This is tracking optimality but our theorem is simply showing rate optimal. Are there strong enough reason for us to use the proof technique we use right now? Or should we focus on the claim of optimal tracking nature? But what can we claim in the statement in this sense?} 

\section{Numerical Simulation}\label{sec:simulation}

In this section, we provide simulations that validate our theoretical finding in Section~\ref{sec:inf} and \ref{sec:infreg}, and further show finite-sample empirical performances across our proposed policies.

\subsection{Inference Only Simulation}
\label{subsection:inference_sim}
%\blue{Add comparison between estimators to show our point in the $\hat \mu$ of IPW? -> Nice to have but not first order.}

In this section, we focus on the Two-Stage Adaptive Neyman policy introduced in Section~\ref{subsec:twostage_closed_form}, focusing on the inference-only objective in Theorem~\ref{thm:general_inf_improvement}. Before turning to Monte Carlo simulations, we use Proposition~\ref{prop:inf_suff} to compute the minimum pilot size \(N_1^{\min}\) required for Algorithm~\ref{alg:post_pilot_general_pcipw} to have lower total MSE than uniform sampling under different variance configurations. As a reminder, Corollary~\ref{cor:pilot_size_condition} shows the adaptive sampling scheme always outperforms the uniform sampling scheme as $N_1 \rightarrow \infty$. Yet, we show below that even in small to modest sample regimes the adaptive scheme still outperforms the uniform sampling scheme. 
We also report
\[
\text{Oracle Gain} := 100 \times \left(\frac{KV/N - S^2/N}{KV/N}\right)
= 100\times\left(1 - \frac{S^2}{KV}\right),
\]
which measures the relative reduction in total MSE achieved by oracle Neyman allocation compared with uniform allocation, as derived in~\eqref{eq:Ney_mse} and~\eqref{eq:unif_mse}. This quantity is not the realized gain from adaptive sampling; rather, it summarizes the amount of variance heterogeneity and the maximum possible efficiency improvement. By contrast, \(N_1^{\min}\) is the theorem-implied pilot threshold above which the two-stage adaptive design is guaranteed to improve upon uniform sampling. Table~\ref{tab:inference_thresholds_thm} reports these calculations and identifies the regimes in which adaptive allocation is expected to be more beneficial.

\begin{table}[htbp!]
\centering
\begin{tabular}{c l c c}
\hline
$K$ & $\boldsymbol{\sigma}$ & $N_1^{\min}$ & Oracle Gain (\%) \\
\hline
4  & $[1,1,1,5]$ & 12   & 42.86 \\
4  & $[3,3,3,4]$ & 260 & 1.74 \\
6  & $[1,1,2,2,3,6]$ & 24   & 31.82 \\
6  & $[1,1,1,3,3,3]$ & 42   & 20.00 \\
10 & $[1,1,2,2,3,3,4,4,5,5]$ & 80   & 18.18 \\
10 & $[1,1,1,1,1,1,1,1,1,10]$ & 30   & 66.88 \\
\hline
\end{tabular}
\caption{Theorem-guided inference thresholds under different arm-variance configurations. The column \(N_1^{\min}\) reports the minimum pilot size required by Proposition~\ref{prop:inf_suff} for the two-stage adaptive Neyman design to improve upon uniform sampling. Oracle Gain reports the relative total-MSE reduction of oracle Neyman allocation compared with uniform allocation.}
\label{tab:inference_thresholds_thm}
\end{table}

The table illustrates two distinct forces behind the benefit of adaptive allocation. First, the Oracle Gain is larger when the arm variances are more heterogeneous, because this is precisely the setting in which Neyman allocation differs most from uniform allocation. In the nearly-homogeneous variance case, the optimal Neyman allocation and uniform allocation are close to equivalent, so there is little gain from adaptivity strategies as reflected in configuration \(K=4\) and \(\boldsymbol\sigma=[3,3,3,4]\) (second row of Table~\ref{tab:inference_thresholds_thm}). Second, the adaptive procedure can realize this advantage only after the pilot stage estimates the variance structure accurately enough. Thus, \(N_1^{\min}\) should be interpreted as the required learning-phase size: when the variances are similar, the oracle gain is small and more pilot samples may be needed before adaptation is guaranteed to outperform uniform sampling.

We now turn to a simulation study, using \(2000\) Monte Carlo replications throughout this section. We begin by examining Algorithm~\ref{alg:post_pilot_general_pcipw} under varying pilot sizes \(N_1\). The goal is to confirm the effect of the length of the pilot phase on total MSE of the resulting two-stage adaptive Neyman procedure. As in Section~\ref{sec:inf}, we compare against uniform sampling and oracle Neyman as the main benchmarks. They are not included as columns in the tables since, for a fixed total budget \(N\), their total MSE does not depend on the pilot size \(N_1\).

\begin{table}[htbp]
\centering
\setlength{\tabcolsep}{4pt}
\renewcommand{\arraystretch}{1.1}

% Panel (a)
\begin{subtable}[t]{0.48\linewidth}
    \vspace{0pt}
    \centering
    \caption{Strong heterogeneity, \(K=4\)}
    \label{tab:k4_strong_adaptive_centered}
    \small
    \begin{tabular}{@{}rrrr@{}}
        \toprule
        \(N_1\) & \(A\) & \(A-U\) & Gain (\%) \\
        \midrule
         8 & 0.476 &  0.253 & -113.21 \\
        12 & 0.196 & -0.027 &   12.09 \\
        16 & 0.177 & -0.047 &   20.86 \\
        20 & 0.157 & -0.067 &   29.90 \\
        24 & 0.157 & -0.066 &   29.59 \\
        28 & 0.151 & -0.073 &   32.45 \\
        \bottomrule
    \end{tabular}
\end{subtable}%
\hfill
% Panel (b)
\begin{subtable}[t]{0.48\linewidth}
    \vspace{0pt}
    \centering
    \caption{Moderate heterogeneity, \(K=10\)}
    \label{tab:k10_moderate_adaptive_centered}
    \small
    \begin{tabular}{@{}rrrr@{}}
        \toprule
        \(N_1\) & \(A\) & \(A-U\) & Gain (\%) \\
        \midrule
         50 & 0.608 &  0.058 & -10.53 \\
         60 & 0.575 &  0.025 &  -4.49 \\
         70 & 0.553 &  0.003 &  -0.53 \\
         80 & 0.539 & -0.011 &   1.98 \\
         90 & 0.531 & -0.019 &   3.51 \\
        100 & 0.521 & -0.029 &   5.26 \\
        110 & 0.514 & -0.036 &   6.60 \\
        \bottomrule
    \end{tabular}
\end{subtable}

\caption{
Total MSE of two-stage adaptive Neyman with pilot-centered IPW (2-AN) under Gaussian rewards, based on \(2{,}000\) Monte Carlo replications. Here, \(A\) denotes 2-AN MSE, \(U\) denotes uniform MSE, and Gain is \(100(U-A)/U\). Panel (a) uses \(\boldsymbol{\sigma}=(1,1,1,5)\) and \(N=500\), with \(U=0.2234\) and oracle Neyman MSE \(0.1279\). Panel (b) uses \(\boldsymbol{\sigma}=(1,1,2,2,3,3,4,4,5,5)\) and \(N=2000\), with \(U=0.5500\) and oracle Neyman MSE \(0.4399\).
}
\label{tab:inference_pilot_comparison}
\end{table}

We report results for two configurations that vary the heterogeneity of the variances $\sigma$in  Table~\ref{tab:inference_pilot_comparison}. In panel~(\subref{tab:k4_strong_adaptive_centered}), the adaptive procedure begins to outperform uniform sampling at \(N_1=12\), while in panel~(\subref{tab:k10_moderate_adaptive_centered}), the improvement begins around \(N_1=80\), consistent with the thresholds in Table~\ref{tab:inference_thresholds_thm}. In both configurations, short pilots lead to worse performance than uniform sampling, while larger pilots generally improve precision over the range considered. The overall pattern is consistent across both: as the pilot size \(N_1\) increases, the two-stage adaptive Neyman procedure improves steadily, reflecting the fact that a longer pilot phase yields more accurate variance estimates and therefore a better plug-in Neyman allocation in stage two. More importantly, Table~\ref{tab:inference_pilot_comparison} shows that even in moderate sample sizes (within 50-100 samples) an adaptive sampling scheme's benefit on inference outweighs the cost.

Overall, these experiments reinforce the theoretical message of this subsection: adaptive allocation is not automatically beneficial when the pilot phase is too short or when the underlying variance is homogeneous. However, it becomes quickly (even in small samples) effective once the pilot budget is large enough to recover the underlying variance structure with reasonable accuracy. 

\subsection{Inference and Regret Simulation}

We now examine the finite-sample performance of Algorithms~\ref{alg:porp}, \ref{alg:sarp}, and~\ref{alg:narp}, focusing on differences in joint loss and computational cost that the common oracle $N^{-1/3}$ rate does not capture. We compare their empirical behavior against three benchmarks: uniform sampling, the complete-information oracle fixed-allocation benchmark from Lemma~\ref{lem:oracle_fixedalloc_joint_clean}, and the \texttt{ForcingBalance} algorithm of \citet{ErraqabiLVBL17}.\footnote{ To be precise, \citet{ErraqabiLVBL17} study a different objective, using a proxy inference term associated with a static-allocation problem together with a $\sqrt N$ scaling, whereas our performance metric is based on summed armwise RMSE and average regret. We therefore remove the $\sqrt N$ scaling in our implementation of \texttt{ForcingBalance} and evaluate all policies under our objective.} Throughout this section, we use 1{,}000 Monte Carlo replications and fix \(\lambda=0.5\). For implementation, SARP uses \(c_x=1\) and a uniform exploration distribution \(p_0\); PORP and NARP both use \(m_0=2\) and \(\alpha=1\); and in both SARP and NARP, the exploitation component is Gaussian Thompson Sampling.

\begin{table}[htbp!]
\centering
\small
\begin{tabular}{lccc}
\toprule
Policy & Sum RMSE & Avg.\ Regret & Joint Loss \\
\midrule
Oracle & 1.99 & 0.77 & 1.38 \\
\texttt{ForcingBalance} & 2.06 & 0.78 & 1.42 \\
PORP & 1.89 & 1.01 & 1.45 \\
NARP & 1.75 & 1.41 & 1.58 \\
SARP & 3.10 & 0.45 & 1.78 \\
Uniform & 1.79 & 1.75 & 1.77 \\
\bottomrule
\end{tabular}
\caption{Policy comparison at horizon \(N=1000\) for the 8-arm configuration with \(\boldsymbol{\mu}=[0,0.5,1,1.5,2,2.5,3,3.5]\), \(\boldsymbol{\sigma} = [1,1,2,2,3,3,4,4]\), and equal weight \(\lambda=0.5\).}
\label{tab:staircase8arm_t1000}
\end{table}

Table~\ref{tab:staircase8arm_t1000} shows the distinct roles of the policies. PORP performs close to both the oracle benchmark and \texttt{ForcingBalance}: at $N=1{,}000$, its estimated joint loss is $1.45$, compared with $1.42$ for \texttt{ForcingBalance}, while its reported runtime is approximately $8.6\%$ lower (Table~\ref{tab:runtime_staircase8arm_mc1000}). On the other hand, NARP obtains substantially lower summed RMSE than SARP, but this inferential gain comes with higher regret. This reflects the intended distinction between the two policies: SARP uses a fixed rate-controlled exploration schedule and therefore tilts more strongly toward exploitation, while NARP uses plug-in estimates of the variance and gap structure to preserve a more informative allocation. At this horizon, the table therefore shows: full plug-in methods are closest to the oracle, while NARP improves substantially over SARP in estimation accuracy within the simpler interpolation class.

\begin{figure}[htbp!]
    \centering
    \includegraphics[width=.8\linewidth]{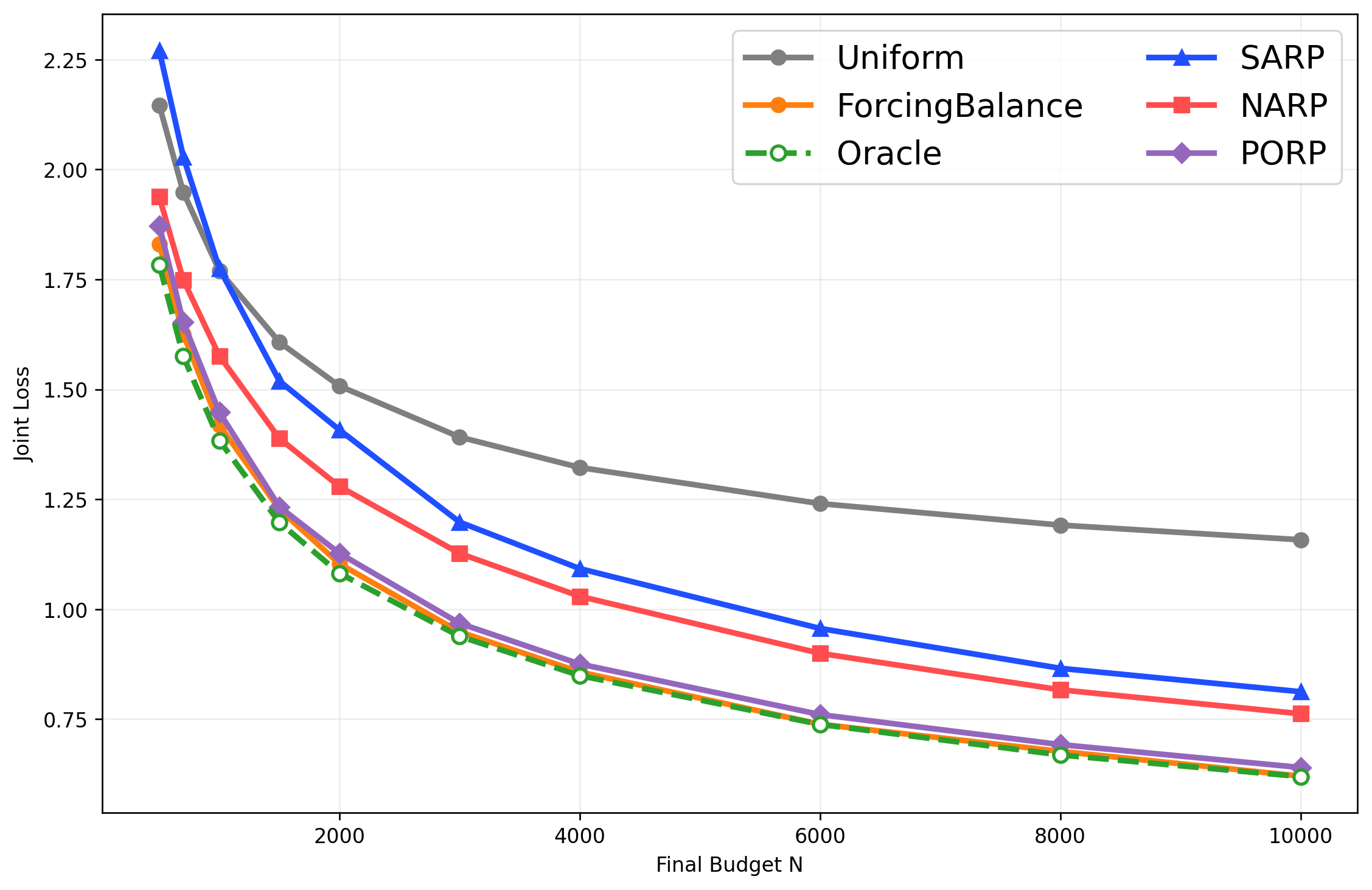}
    \caption{Joint objective as a function of the total budget \(N\) for the 8-arm configuration with \(\boldsymbol{\mu} = [0,0.5,1,1.5,2,2.5,3,3.5]\), \(\boldsymbol{\sigma}=[1,1,2,2,3,3,4,4]\), and objective weight \(\lambda=0.5\).}
    \label{fig:staircase_8arm_obj}
\end{figure}

Figure~\ref{fig:staircase_8arm_obj} plots the joint objective as the total budget \(N\) increases from $500$ to $10{,}000$. The first observation is that PORP, SARP, and NARP all exhibit the same qualitative decay pattern as the oracle benchmark, consistent with the \(N^{-1/3}\) scaling in Theorems~\ref{thm:porp}, \ref{thm:joint_sarp_rate}, and~\ref{thm:NARP}. The second observation is that PORP remains close to the oracle, reinforcing its role as the objective-aligned plug-in benchmark. The third observation is that NARP remains below SARP across the horizon range, showing that adaptive calibration improves on the fixed exploration schedule in this instance. We next examine the computational cost of obtaining these improvements.

\begin{table}[htbp!]
\centering
\small
\begin{tabular}{lccc}
\toprule
Policy & $N=500$ & $N=1000$ & $N=2000$ \\
\midrule
SARP & 10.75 & 21.96 & 61.41  \\
NARP & 26.10  & 55.72 & 192.19  \\
PORP & 46.74 & 150.35 & 272.03 \\
\texttt{ForcingBalance} & 63.46  & 164.46 & 353.17  \\
\bottomrule
\end{tabular}
\caption{Running time, in seconds, for 1{,}000 Monte Carlo replications on the 8-arm configuration in Figure~\ref{fig:staircase_8arm_obj}.}
\label{tab:runtime_staircase8arm_mc1000}
\end{table}

Table~\ref{tab:runtime_staircase8arm_mc1000} shows that the performance
differences above come with different computational costs. PORP is faster than \texttt{ForcingBalance} at all three horizons. This difference is largely because the KKT conditions for our oracle allocation reduce the PORP update to a one-dimensional multiplier search, whereas \texttt{ForcingBalance} is implemented as a full allocation optimization followed by a tracking step. At the same time, PORP is still more expensive than SARP and NARP because it solves the plug-in allocation problem repeatedly. SARP is the cheapest
method because its exploration probability follows the deterministic \(t^{-1/3}\) schedule and requires no online calibration. NARP adds the cost of updating empirical estimates, but it avoids the repeated full allocation solve required by PORP and \texttt{ForcingBalance}. Thus, NARP occupies a middle ground: it retains oracle-guided adaptation that improves finite-sample inference relative to SARP, while preserving much of the simplicity and modularity compared to PORP. Because the joint objective combines inferential accuracy and regret, we
next decompose performance into its two components.

\begin{figure}[htbp!]
    \centering
    \includegraphics[width=1\linewidth]{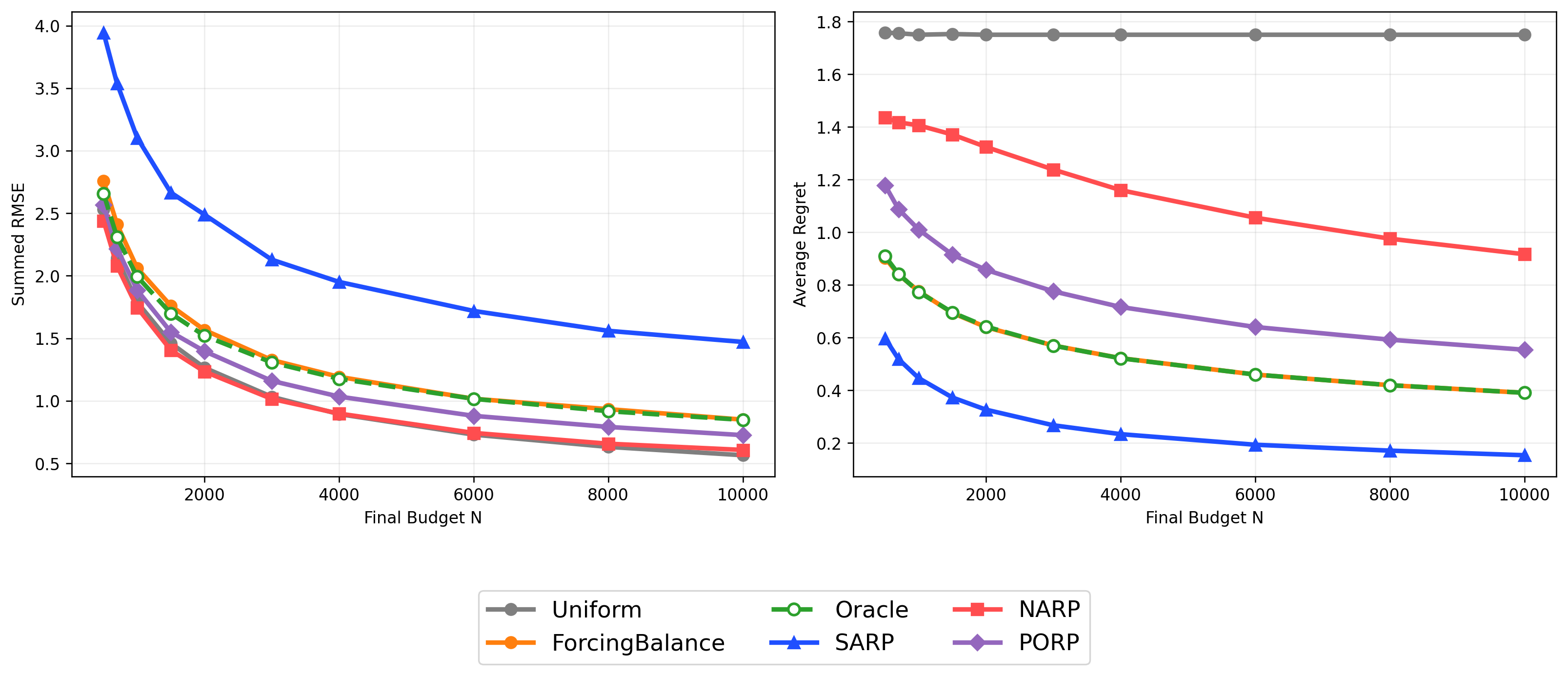}
    \caption{Performance decomposition as a function of the total budget \(N\) for the 8-arm configuration with \(\boldsymbol{\mu} =[0,0.5,1,1.5,2,2.5,3,3.5]\), \(\boldsymbol{\sigma} = [1,1,2,2,3,3,4,4]\), and objective weight \(\lambda=0.5\). The left panel plots summed RMSE, and the right panel plots average regret.}
    \label{fig:staircase_8arm_decompose}
\end{figure}

The decomposition in Figure~\ref{fig:staircase_8arm_decompose} explains the mechanism behind the joint-objective comparisons. Uniform sampling and NARP achieves good RMSE performance because they allocate substantial information across arms at the cost of higher regret. SARP has the opposite tradeoff: its exploitation component controls regret aggressively, but its less optimal exploration leads to much worse inferential accuracy. PORP and \texttt{ForcingBalance} occupy a more balanced region, controlling regret much better than uniform sampling while maintaining substantially better inference than SARP. Summing up, NARP improves over SARP by giving up some of SARP's regret advantage, but recovers a large amount of inferential precision. PORP, in contrast, shows what is obtained when the empirical oracle allocation is solved more directly rather than approximated through the simpler interpolation rule.

\begin{figure}[htbp!]
    \centering
    \includegraphics[width=1\linewidth]{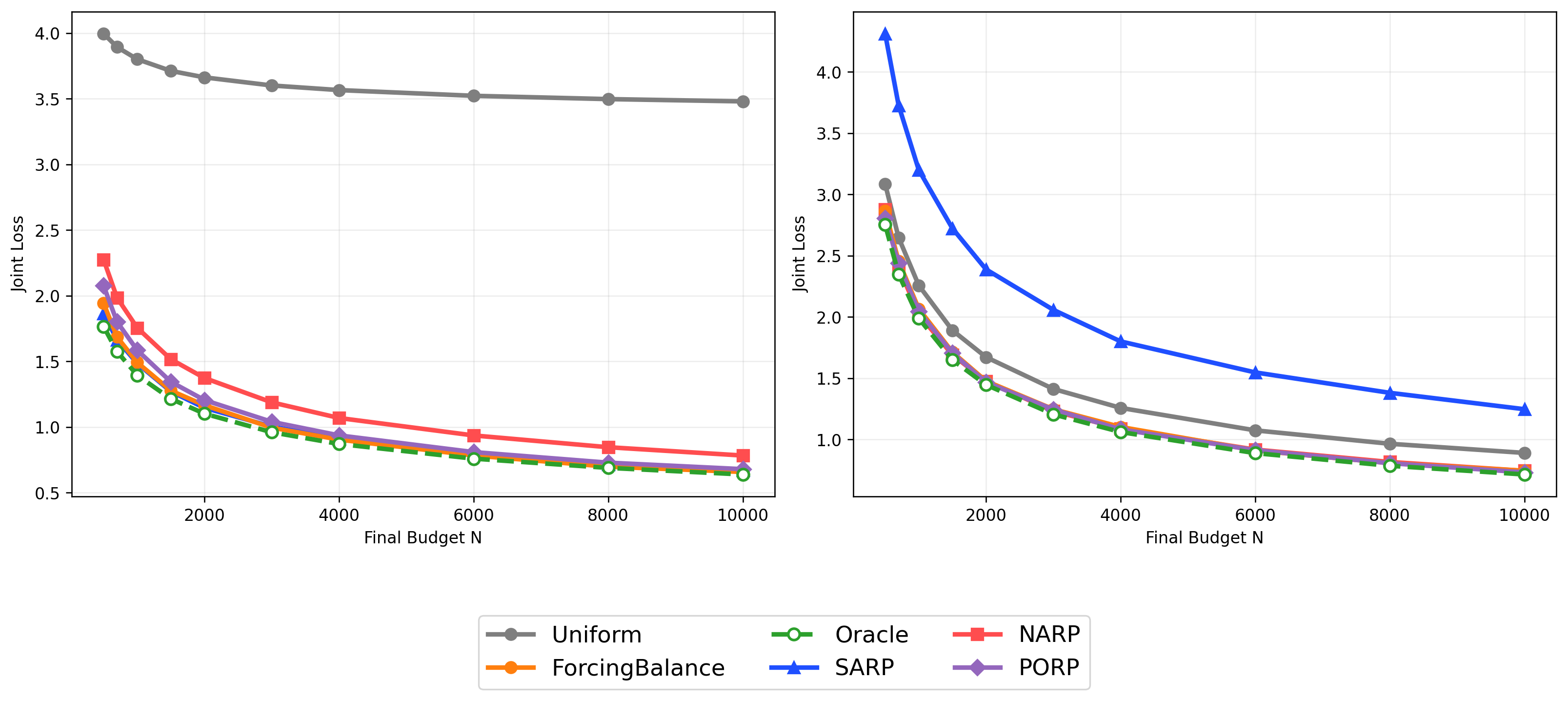}
    \caption{Joint objective as a function of the total budget \(N\) for two configurations with weight $\lambda = 0.5$. Panel (a) shows the 6-arm configuration with \(\boldsymbol{\mu}=[0,1,2,3,4,10]\), \(\boldsymbol{\sigma}=[1,1,2,2,3,3]\). Panel (b) shows the 10-arm configuration with \(\boldsymbol{\mu}=[0,0.2,0.4,0.6,0.8,1,1.1,1.2,1.25,1.3]\), \(\boldsymbol{\sigma}=[1,1,2,2,3,3,5,5,8,10]\). The 6-arm setting has a clearly separated best arm and moderate variance heterogeneity, while the 10-arm setting combines weak top-arm separation with large variance heterogeneity.}
    \label{fig:6_10arm_comparison}
\end{figure}

Finally, Figure~\ref{fig:6_10arm_comparison} shows that the relative performance of the policies depends on the underlying instance. Panel~(a) considers an easier 6-arm configuration with a clearly separated best arm and moderate variance heterogeneity. In this setting, the inferential cost of more aggressive exploitation is smaller, so SARP achieves lower joint loss than NARP. Panel~(b) considers a harder 10-arm configuration with weak separation among the top arms and substantial variance heterogeneity. In this setting, aggressive exploitation can leave too little information for accurate arm-level estimation, so the benefit of adaptive exploration becomes much more visible. PORP, \texttt{ForcingBalance}, and NARP remain close to the oracle benchmark, whereas SARP is noticeably worse. These comparisons suggest that NARP's additional exploration is most valuable when its inferential gains outweigh the associated regret cost, while SARP's less optimal exploration can be preferable when the cost of regret is higher than inference. 

Overall, these simulations show that PORP, SARP, and NARP serve different purposes in our framework. PORP is the full plug-in oracle-rate benchmark: it closely follows the oracle objective and competes with \texttt{ForcingBalance}, while being computationally more efficient because its allocation update reduces to a scalar multiplier search rather than a full optimization-and-tracking step. However, PORP still requires repeatedly recomputing an instance-dependent allocation, so it is less lightweight than the rate-based policies. SARP and NARP are designed with a different practical goal: they modularize the design by allowing the experimenter to combine an inference-oriented exploration rule with any standard regret-minimizing algorithm, and then control the interpolation between the two. SARP is the cheapest and most exploitation-oriented version, making it attractive when regret dominates or the inferential problem is easy. NARP adds plug-in calibration of the exploration rate and variance-sensitive exploration shape, which improves inference in harder instances while avoiding the repeated allocation solve required by PORP and \texttt{ForcingBalance}. Thus, the main contribution of SARP and NARP is not to dominate full plug-in allocation methods in every finite sample, but to provide a transparent, modular, and computationally light way to operate on the inference--regret frontier while still attaining the oracle rate.

\section{Application: Virtual Interactions in Organizations}\label{sec:application}

Having evaluated our methods in controlled simulations, we next examine their effectiveness using a real-world dataset. Our main application uses data from the randomized field experiment of \citet{choudhury2021virtualwatercoolers}.\footnote{The partner organization conducted the original experiment and surveys as part of its internship program, without a separate research consent process, and anonymized the data before sharing them with the researchers. The present study uses these existing data to calibrate Monte Carlo simulations and involves no new participant recruitment or intervention.} We also report a parallel application in Appendix~\ref{apx:gym_app} using data from \citet{milkman2021megastudies}, which studies behavioral interventions designed to encourage exercise.

\subsection{Dataset Description}

The dataset we use comes from a randomized field experiment on virtual interactions during a fully remote summer internship program conducted in 2020 at a large global organization (\citet{choudhury2021virtualwatercoolers}). The internship program, in a typical year, combined apprenticeship-based work with in-person networking and social activities. In response to the COVID-19 pandemic, the organization moved the 2020 program to a shortened, fully remote five-week format. Executives were concerned that interns would have fewer opportunities to work with their cohort, learn about the organization's culture, and receive mentoring from senior employees. As part of its broader effort to redesign the internship for a remote setting, the organization developed and implemented a randomized field experiment of several structured forms of virtual interaction, with the research team advising on the experimental design. The organization's operational concern was to maintain interns' access to socialization and mentoring with peers and senior managers. The research objective was to estimate the causal effects of these virtual-interaction formats on the performance and career outcomes of remote interns.

The experiment included 1{,}370 interns in the organization's North American internship program, spanning 16 program cities and multiple organizational divisions. During weeks 2--5 of the internship, interns were randomly assigned to one of five experimental virtual interaction conditions, which we treat as the arms of a pseudo-bandit: a passive control condition, an asynchronous question-and-answer forum with senior managers, an intern group project, an intern-only virtual water cooler, and an intern--senior-manager virtual water cooler. After data cleaning and only including the interns with a supervisor performance rating, we have 1,330 interns left. Additional details on the experimental conditions and summary statistics are provided in Appendix~\ref{apx:virtual_data}.

This setting maps naturally to the two objectives studied in this paper. An organization may want to learn the performance consequences of several alternative program designs while also avoiding excessive assignment to programs that appear to lose promising employees during the experiment. Our application therefore asks the following counterfactual design question: \textit{if the organization were to deploy these virtual-interaction programs sequentially across future internship cohorts, how would adaptive assignment affect the precision with which the program effects are estimated, and what performance cost would be incurred during the learning process?}

\subsection{Experiment Setup}

Our primary reward is the intern's final supervisor performance rating. The original rating scale takes values \(1\), \(2\), and \(3\), where \(1\) denotes outstanding performance, \(2\) denotes good performance, and \(3\) denotes performance needing improvement. We reverse-code the rating by defining
\[
Y=4-\text{Original Rating},
\]
so that larger values correspond to better performance. This yields a bounded reward \(Y\in\{1,2,3\}\), measured directly in performance-rating units. Let $D_j=\{Y_{j1},\ldots,Y_{jn_j}\}$
denote the collection of observed performance outcomes in arm \(j\), where \(n_j\) is the number of observed outcomes in that arm. Conditional on the observed dataset, we define the empirical reward distribution as $ \widehat F_j := \frac{1}{n_j} \sum_{i=1}^{n_j}\delta_{Y_{ji}},$ where \(\delta_y\) denotes a point mass at \(y\). The corresponding empirical mean and variance are
\[
\mu_j^{\mathrm{emp}} := \frac{1}{n_j}\sum_{i=1}^{n_j}Y_{ji}, \qquad \left(\sigma_j^{\mathrm{emp}}\right)^2 := \frac{1}{n_j} \sum_{i=1}^{n_j} \left(Y_{ji}-\mu_j^{\mathrm{emp}}\right)^2,
\]
which can be used to determine the oracle allocation.

At each round \(t\), after the policy selects arm \(I_t=i\), the synthetic reward is generated according to $X_t\mid I_t=i,\mathcal D \sim \widehat F_i,$ independently across pulls conditional on the selected arms and the observed dataset \(\mathcal D=(D_1,\ldots,D_K)\). Equivalently, each pull samples an observed outcome from the corresponding arm with replacement. Each feasible policy observes only its own sequential arm assignments and realized rewards and not the empirical oracle, exactly as in the bandit formulation of Section~\ref{sec:setup}. The full dataset is used only by the evaluator to construct the environment and the complete-information benchmarks. This exercise should therefore not be interpreted as a re-estimation of the causal effects in the original field experiment. Rather, it is an empirically calibrated Monte Carlo experiment which allows us to compare how alternative adaptive assignment policies would perform in the same environment, highlighting the effect of our proposed policies.

To evaluate the policies in terms of MSE and pseudo-regret, we use a Monte Carlo design with \(B=1{,}000\) replications. In each replication \(b=1,\ldots,B\), we generate independent reward streams from the empirical arm distributions, reset the policy history, and run each policy for horizons \(N\) ranging from 250 to 5{,}000. Because rewards are sampled with replacement, \(N\) may exceed the number of observations in the original field experiment. Thus, \(N\) represents the number of interns in a hypothetical future sequential rollout, rather than the number of unused observations in the original dataset. We set the inference--regret tradeoff parameter to \(\lambda=0.1\) as an illustrative specification. In this application, the arm variances are similar, making the inference-optimal allocation close to uniform, while the empirical reward gaps are modest. Placing greater weight on regret therefore makes the operational consequences of adaptive assignment more visible. More generally, \(\lambda\) is a practitioner-chosen parameter that can be adjusted to reflect the relative importance of estimation precision versus regret performance during the experiment. Focusing on the joint inference--regret objective introduced in Section~\ref{sec:infreg}, we compare uniform allocation, PORP, SARP, NARP, \texttt{ForcingBalance}, and the complete-information oracle. The complete-information oracle treats each arm's empirical distribution as the population reward distribution for this resampling experiment and uses the corresponding empirical means and variances to compute the optimal fixed allocation for each horizon $N$. For SARP and NARP, we use Thompson Sampling as the regret-minimizing exploitation component while the SARP continues to utilize the uniform sampling scheme as its base exploration policy.

\subsection{Results}

We first report the joint objective for each policy and horizon in Figure~\ref{fig:app_virtual_jointloss}. Among the feasible policies, PORP consistently achieves the lowest joint loss. More broadly, the adaptive policies increasingly outperform uniform allocation as the budget grows, demonstrating the value of adjusting assignments based on the outcomes observed during the experiment.

\begin{figure}[htbp!]
    \centering
    \includegraphics[width = 0.8\linewidth]{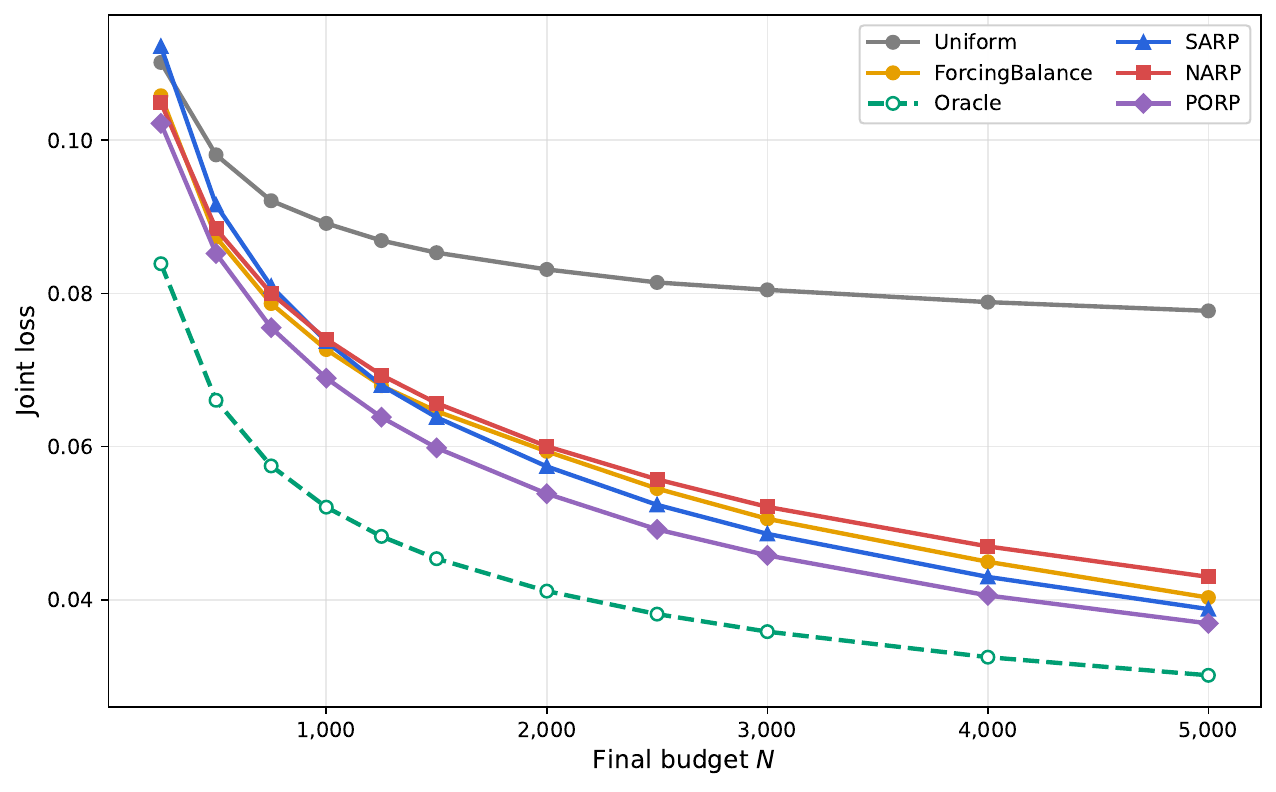}
    \caption{Joint inference--regret performance. The figure reports the Monte Carlo estimate of the joint objective across \(B=1{,}000\) replications as the total budget \(N\) increases. The objective places weight \(\lambda=0.1\) on summed RMSE and weight \(1-\lambda=0.9\) on average pseudo-regret.}
    \label{fig:app_virtual_jointloss}
\end{figure}

In addition, we report summed RMSE and average pseudo-regret in Figure~\ref{fig:app_virtual_decomp}. This decomposition reveals whether an improvement in the joint objective arises from better inference, lower regret, or both. Uniform allocation achieves the lowest summed RMSE because it maintains balanced coverage across the five arms, yet it also incurs the highest linear pseudo-regret. In contrast, the adaptive policies accept some loss in estimation precision in exchange for lower pseudo-regret. PORP provides the most favorable overall balance. Among the mixing policies, SARP achieves lower joint loss at the larger displayed budgets, where its lower regret outweighs NARP's greater inferential precision. By contrast, when inferential precision outweight regret differences, NARP achieves the lowest estimated joint loss, as we show in Appendix~\ref{apx:gym_app}.

\begin{figure}[htbp!]
    \centering
    \includegraphics[width = 1\linewidth]{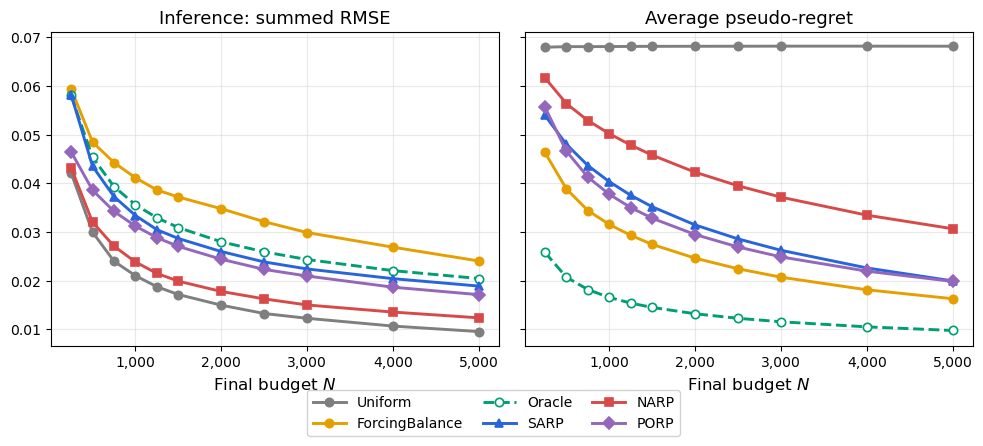}
    \caption{Decomposition of performance in the virtual-interaction application across \(B=1{,}000\) Monte Carlo replications. The left panel reports summed RMSE across the five treatment arms, measuring the precision with which the empirical arm means are estimated. The right panel reports average pseudo-regret, measuring the performance loss from assigning interns to arms other than the empirically best arm. The scale of y-axis for both plots are $\lambda$ adjusted.}
    \label{fig:app_virtual_decomp}
\end{figure}

Table~\ref{tab:computation-time} reports average computation time per decision, which includes a policy update from the observed history and assignment of the next intern to an experimental condition. SARP is fastest at $75.0\,\mu\mathrm{s}$ because it uses a fixed exploration distribution and a closed-form exploration rate. NARP additionally updates estimated means, variances, gaps, and the rooted-Neyman allocation, requiring $107.0\,\mu\mathrm{s}$. Both avoid repeated numerical allocation solves, whereas PORP and \texttt{ForcingBalance} require $276.5$ and $320.9\,\mu\mathrm{s}$ per decision, respectively. Relative to \texttt{ForcingBalance}, SARP and NARP reduce computation time by approximately 76.6\% and 66.7\%, highlighting the computational advantage of closed-form mixing. PORP achieves a smaller 13.8\% reduction while delivering the lowest estimated joint loss among feasible policies in this application.

\begin{table}[!t]
    \centering
    \begin{tabular}{lc}
        \toprule
        Policy & Average time ($\mu$s/decision) \\
        \midrule
        SARP                    & 75.0  \\
        NARP                    & 107.0 \\
        PORP                    & 276.5 \\
        \texttt{ForcingBalance} & 320.9 \\
        \bottomrule
    \end{tabular}
    \caption{Average computation time per sequential assignment decision, which consists of updating the policy using the observed history. Times are reported in microseconds (\(1 \ \mu\mathrm{s}=10^{-6}\) seconds). SARP and NARP use closed-form interpolation rules, whereas PORP and \texttt{ForcingBalance} solve a numerical allocation problem.}
    \label{tab:computation-time}
\end{table}

\section{Discussion}\label{sec:discussion}
In summary, this paper studies adaptive experimentation through two complementary lenses: statistical efficiency and online performance. We show that in the pure-inference setting, adaptive allocation can improve upon uniform sampling in finite samples once the pilot phase is informative enough, which is achieved in moderate sample sizes. We then show that in the joint inference--regret setting, simple adaptive policies can attain the natural oracle rate while expressing different operating points on the same tradeoff frontier. The resulting message is not that one adaptive rule uniformly dominates all others, but that the value of adaptivity depends on the structure of the instance, on the experimenter's objective, and operational constraints (interpretability, computational costs, etc.). By making this dependence explicit, our framework helps clarify when adaptive sampling is beneficial, when it is worth paying for additional exploration, and how that choice can be encoded in a policy that remains interpretable and practically deployable. Overall, although this paper does not take a stance on which proposed adaptive policy is best, we still argue that adaptivity is generally helpful than non-adaptive uniform sampling in both inference and inference-and-regret settings. 

\newpage
\section*{Acknowledgement}
The authors thank Jinglong Zhao and Chonghuan Wang for helpful discussions. 

{
\bibliographystyle{ims}
\bibliography{bib}
}

\newpage

\appendix

\section{Technical Lemmas}
\begin{lemma}[MSE Decomposition of Two-Stage Horvitz--Thompson]\label{lem:mse_twostage_ht}
Under the setup of Algorithm~\ref{alg:post_pilot_general_pcipw}, but set estimator to be 
\[
\hat\mu_i^{\mathrm{HT}} := \frac{1}{N}\sum_{t=1}^N \frac{\ind\{I_t=i\}X_t}{p_{t,i}}, \quad i\in[K],
\]
which is the global Horvitz--Thompson estimator over horizon $N$. Then the total MSE for arm $i$ decomposes as
\[
\mathrm{MSE}(i, \hat \pi^{\rm AN}) = \frac{N_1}{N} \mathrm{MSE}(i, \pi^{\rm uni}) + \frac{1}{N^2}\sum_{t=N_1+1}^{N} \EE \left[\frac{\mu_i^2+\sigma_i^2}{p_{t,i}}-\mu_i^2\right] \ \text{ under HT estimator.}
\]
\end{lemma}

\begin{proof}
Define the HT score
\[
Z_{t,i}:=\frac{\ind\{I_t=i\}X_t}{p_{t,i}}.
\]
Then
\[
\hat\mu_i^{\mathrm{HT}}=\frac{1}{N}\sum_{t=1}^N Z_{t,i} = \frac{1}{N}\left(\sum_{t=1}^{N_1}Z_{t,i}+\sum_{t=N_1+1}^{N}Z_{t,i} \right).
\]

Under the balanced pilot design, the first-stage contribution satisfies
\[
\sum_{t=1}^{N_1}Z_{t,i} = K\sum_{t=1}^{N_1}\ind\{I_t=i\}X_t = N_1\hat\mu_i^{\mathrm{pilot}},
\]
where $ \hat \mu_i^{\mathrm{pilot}} := \frac{1}{n_1}\sum_{t=1}^{N_1}\ind\{I_t=i\}X_t.$
Hence
\[
\hat\mu_i^{\mathrm{HT}}-\mu_i = \frac{1}{N}\left(
N_1(\hat\mu_i^{\mathrm{pilot}}-\mu_i) + \sum_{t=N_1+1}^{N}(Z_{t,i}-\mu_i)
\right).
\]

Squaring and taking expectations gives
\begin{align*}
\mathrm{MSE}^{\mathrm{HT}}(i)
&=
\EE\!\left[\left(\frac{N_1}{N}(\hat\mu_i^{\mathrm{pilot}}-\mu_i)\right)^2\right]
+
\EE \left[\left(\frac{1}{N}\sum_{t=N_1+1}^{N}(Z_{t,i}-\mu_i)\right)^2\right] \\
&\qquad
+
\frac{N_1}{N^2}
\EE \left[
(\hat\mu_i^{\mathrm{pilot}}-\mu_i)\sum_{t=N_1+1}^{N}(Z_{t,i}-\mu_i)
\right].
\end{align*}

The cross-term vanishes by iterated expectation, since for each $t>N_1$,
\[
\EE[Z_{t,i}\mid \mathcal F_{t-1}] = \frac{\EE[\ind\{I_t=i\}X_t\mid \mathcal F_{t-1}]}{p_{t,i}} = \mu_i.
\]
Therefore,
\[
\EE \left[ (\hat\mu_i^{\mathrm{pilot}}-\mu_i)\sum_{t=N_1+1}^{N}(Z_{t,i}-\mu_i) \right]=0.
\]

For the pilot term,
\[
\EE \left[\left(\frac{N_1}{N}(\hat\mu_i^{\mathrm{pilot}}-\mu_i)\right)^2\right]
=
\frac{N_1^2}{N^2}\Var(\hat\mu_i^{\mathrm{pilot}})
=
\frac{N_1}{N}\,\mathrm{MSE}_{\mathrm{unif},N}(i).
\]

For the adaptive term, the martingale-difference cross-products vanish, so
\[
\EE\!\left[\left(\frac{1}{N}\sum_{t=N_1+1}^{N}(W_{t,i}-\mu_i)\right)^2\right]
=
\frac{1}{N^2}\sum_{t=N_1+1}^{N}\EE[\Var(W_{t,i}\mid \mathcal F_{t-1})].
\]
Now
\[
\Var(W_{t,i}\mid \mathcal F_{t-1}) = \EE[Z_{t,i}^2\mid \mathcal F_{t-1}] - \mu_i^2.
\]
Since
\[
\EE[Z_{t,i}^2\mid \mathcal F_{t-1}] = \frac{1}{p_{t,i}^2}\EE[\ind\{I_t=i\}X_t^2\mid \mathcal F_{t-1}] = \frac{\mu_i^2+\sigma_i^2}{p_{t,i}},
\]
we obtain
\[
\Var(Z_{t,i}\mid \mathcal F_{t-1}) = \frac{\mu_i^2+\sigma_i^2}{p_{t,i}}-\mu_i^2 = \frac{\sigma_i^2}{p_{t,i}}+\mu_i^2\left(\frac{1}{p_{t,i}}-1\right).
\]
Substituting this into the adaptive term yields the claimed decomposition.
\end{proof}

\begin{lemma}[Multiplicative Chernoff lower-tail bound]
\label{lem:chernoff-lower-tail}
Let $Z_1,\dots,Z_n$ be independent Bernoulli random variables, and define
\[
S_n:=\sum_{j=1}^n Z_j, \quad \mu:=\EE[S_n].
\]
Then for every $\delta\in(0,1)$,
\[
\PP(S_n\le (1-\delta)\mu)\le \exp\!\left(-\frac{\delta^2\mu}{2}\right).
\]
\end{lemma}

\begin{lemma}[Sub-Gaussian Hoeffding inequality, Theorem 2.7.3 in \cite{vershynin2026hdp}]\label{lem:subg-hoeff}
For an arm $i\in[K]$, let
\[
\hat\mu_{i,n}:=\frac1n\sum_{s=1}^n X_{i,s}, \quad \mu_i:=\EE[X_{i,1}].
\]
Assume that $X_{i,1},\dots,X_{i,n}$ are independent and that each centered variable $X_{i,s}-\mu_i$ is $\bar\nu$-sub-Gaussian. Then there exists a universal constant $c_\mu>0$ such that for every $n\ge 1$ and every $\varepsilon>0$,
\[
\PP \left(\left|\hat\mu_{i,n}-\mu_i\right|>\varepsilon\right) \le 2\exp \left(-c_\mu n \varepsilon^2/\bar\nu^2\right).
\]
\end{lemma}

\begin{lemma}[Bernstein concentration for centered second moments]\label{lem:bernstein-second-moment}
For an arm $i\in[K]$, define
\[
\sigma_i^2:=\Var(X_{i,1}), \quad \hat v_{i,n}:=\frac1n\sum_{s=1}^n (X_{i,s}-\mu_i)^2.
\]
Assume that $X_{i,1},\dots,X_{i,n}$ are independent and that each centered variable $X_{i,s}-\mu_i$ is $\bar\nu$-sub-Gaussian. Then there exists a universal constant $c_v>0$ such that for every $n\ge 1$ and every $\varepsilon>0$,
\[
\Pr\!\left(\left|\hat v_{i,n}-\sigma_i^2\right|>\varepsilon\right)
\le
2\exp\!\left(
-c_v\, n \min\left\{\frac{\varepsilon^2}{\bar\nu^4},\frac{\varepsilon}{\bar\nu^2}\right\}
\right).
\]
\end{lemma}

\begin{proof}
Let
\[
Z_s := (X_{i,s}-\mu_i)^2-\sigma_i^2.
\]
Since $X_{i,s}-\mu_i$ is $\bar\nu$-sub-Gaussian, Lemma 2.8.5 in \cite{vershynin2026hdp} implies that $(X_{i,s}-\mu_i)^2$ is sub-exponential with $\psi_1$-norm of order $\bar\nu^2$, and hence so is $Z_s$. Applying Corollary 2.9.2 in \cite{vershynin2026hdp} to the independent mean-zero sub-exponential variables $Z_1,\dots,Z_n$ with coefficients $a_s\equiv 1/n$ gives our result.
\end{proof}

\begin{corollary}[Empirical standard deviation concentration]\label{lem:sd-concentration}

For an arm $i\in[K]$, define
\[
\hat\sigma_{i,n}^2:=\frac1n \sum_{s=1}^n (X_{i,s}-\hat\mu_{i,n})^2, \quad \sigma_i^2:=\Var(X_{i,1}).
\]
Assume that $X_{i,1},\dots,X_{i,n}$ are independent, that each centered variable $X_{i,s}-\mu_i$ is $\bar\nu$-sub-Gaussian, and that $\sigma_i>0$. Then there exists a universal constant $c_\sigma>0$ such that for every $n\ge 1$ and every $\varepsilon\in(0,1/2]$,
\[
\PP \left(\left|\frac{\hat\sigma_{i,n}}{\sigma_i}-1\right|>\varepsilon\right) \le 4\exp \left(-c_\sigma  n  \sigma_i^4\varepsilon^2/\bar\nu^4\right).
\]
\end{corollary}

\begin{proof}
Define
\[
\hat v_{i,n}:=\frac1n\sum_{s=1}^n (X_{i,s}-\mu_i)^2.
\]
Using
\[
\hat\sigma_{i,n}^2 = \frac1n\sum_{s=1}^n (X_{i,s}-\hat\mu_{i,n})^2 = \frac1n\sum_{s=1}^n (X_{i,s}-\mu_i)^2-(\hat\mu_{i,n}-\mu_i)^2 = \hat v_{i,n}-(\hat\mu_{i,n}-\mu_i)^2,
\]
we have $|\hat\sigma_{i,n}^2-\sigma_i^2| \le |\hat v_{i,n}-\sigma_i^2|+(\hat\mu_{i,n}-\mu_i)^2.$ Now, if $\left|\frac{\hat\sigma_{i,n}}{\sigma_i}-1\right|>\varepsilon,$ then either $\hat\sigma_{i,n}>(1+\varepsilon)\sigma_i$ or $\hat\sigma_{i,n}<(1-\varepsilon)\sigma_i$. In either case,
\[
\left|\frac{\hat\sigma_{i,n}^2}{\sigma_i^2}-1\right| > \varepsilon,
\]
because $ (1+\varepsilon)^2-1=2\varepsilon+\varepsilon^2>\varepsilon$ and $1-(1-\varepsilon)^2=2\varepsilon-\varepsilon^2\ge \varepsilon$ for every $\varepsilon\in(0,1]$. Therefore
\[
\left\{\left|\frac{\hat\sigma_{i,n}}{\sigma_i}-1\right|>\varepsilon\right\} \subseteq \left\{|\hat\sigma_{i,n}^2-\sigma_i^2|>\sigma_i^2\varepsilon\right\}.
\]

Combining this inclusion with the previous decomposition yields
\[
\left\{\left|\frac{\hat\sigma_{i,n}}{\sigma_i}-1\right|>\varepsilon\right\} \subseteq \left\{|\hat v_{i,n}-\sigma_i^2| > \frac12 \sigma_i^2\varepsilon\right\} \cup \left\{(\hat\mu_{i,n}-\mu_i)^2>\frac12 \sigma_i^2\varepsilon\right\}.
\]
Hence, by a union bound,
\begin{align*}
\PP \left(\left|\frac{\hat\sigma_{i,n}}{\sigma_i}-1\right|>\varepsilon\right) &\le \PP \left(|\hat v_{i,n}-\sigma_i^2|>\frac12 \sigma_i^2\varepsilon\right) + \PP \left(|\hat\mu_{i,n}-\mu_i|>\sigma_i\sqrt{\frac{\varepsilon}{2}}\right).
\end{align*}

Applying Lemma~\ref{lem:bernstein-second-moment} and Lemma~\ref{lem:subg-hoeff}, we obtain
\begin{align*}
\PP \left(\left|\frac{\hat\sigma_{i,n}}{\sigma_i}-1\right|>\varepsilon\right) &\le 2\exp \left( -c_v n \min\left\{ \frac{\sigma_i^4\varepsilon^2}{4\bar\nu^4}, \frac{\sigma_i^2\varepsilon}{2\bar\nu^2} \right\} \right) + 2\exp \left( -c_\mu n \frac{\sigma_i^2\varepsilon}{2\bar\nu^2} \right).
\end{align*}

It remains to compare the linear and quadratic terms. Since $X_{i,s}-\mu_i$ is $\bar\nu$-sub-Gaussian, its variance satisfies $\sigma_i^2=\Var(X_{i,s})\le \bar\nu^2.$ Also, $\varepsilon\in(0,1/2]$, therefore
\[
\frac{\sigma_i^4\varepsilon^2}{4\bar\nu^4} \le \frac{\sigma_i^2\varepsilon}{4\bar\nu^2} \le \frac{\sigma_i^2 \varepsilon}{2\bar\nu^2}.
\]
Hence the quadratic term is the smaller one, so $\min\left\{ \frac{\sigma_i^4\varepsilon^2}{4\bar\nu^4}, \frac{\sigma_i^2\varepsilon}{2\bar\nu^2} \right\} = \frac{\sigma_i^4\varepsilon^2}{4\bar\nu^4},$ and similarly $\frac{\sigma_i^2\varepsilon}{2\bar\nu^2} \ge \frac{\sigma_i^4\varepsilon^2}{2\bar\nu^4} \ge \frac{\sigma_i^4\varepsilon^2}{4\bar\nu^4}.$ Thus both exponentials are bounded by the same quadratic scale, giving
\[
\PP \left(\left|\frac{\hat\sigma_{i,n}}{\sigma_i}-1\right|>\varepsilon\right) \le 4\exp \left(-c_\sigma  n \frac{\sigma_i^4\varepsilon^2}{\bar\nu^4}\right),
\]
where
\[
c_\sigma:=\min\left\{\frac{c_v}{4},\frac{c_\mu}{4}\right\}.
\]
This proves the claim.
\end{proof}

\begin{lemma}[Freedman inequality, Theorem 1.1 from \cite{tropp2011freedman}]\label{lem:freed_ineq}
Let $(Y_k)_{k\ge 0}$ be a real-valued martingale adapted to $(\mathcal F_k)$, with differences $X_k:=Y_k-Y_{k-1}$. Assume that, almost surely, $X_k\le R$ for all  $k\ge 1.$ Define the predictable quadratic variation $ W_k:=\sum_{j=1}^k \EE[X_j^2\mid\mathcal F_{j-1}].$ Then, for all $t\ge 0$ and $\sigma^2>0$,
\[
\PP \left(\exists\, k\ge 0:\ Y_k\ge t \ \text{and}\ W_k\le \sigma^2\right) \le \exp \left(-\frac{t^2/2}{\sigma^2+Rt/3}\right).
\]
\end{lemma}

\begin{theorem}[Doob's $L^2$ maximal inequality, Theorem 4.4.4 from \cite{durrett2019probability} with $p=2$]
\label{thm:doob-l2}
Let $(M_n,\mathcal F_n)_{n\ge 0}$ be a square-integrable martingale. Then for every integer $m\ge 1$,
\[
\EE \left[\max_{0\le n\le m} |M_n|^2\right] \le 4 \EE[|M_m|^2].
\]
\end{theorem}

\begin{lemma}[MSE bound from a deterministic pull lower bound] \label{lem:dyadic-mse-pullorder}
Fix an arm $i\in[K]$. Under the independent-sequences model, let
\[
S_{i,n}:=\sum_{s=1}^n (Y_{i,s}-\mu_i), \quad \hat\mu_i(N)-\mu_i=\frac{S_{i,T_i(N)}}{T_i(N)}.
\]
Let $\mathcal E$ be any event such that on $\mathcal E$ one has the deterministic lower bound $T_i(N)\ge n_0$ for some integer $n_0\ge 1$. Then there exists a universal constant $C < \infty$ such that
\[
\EE \left[(\hat\mu_i(N)-\mu_i)^2 \ind\{\mathcal E\}\right] \le C \frac{\sigma_i^2}{n_0}.
\]
\end{lemma}

\begin{proof}
We use a dyadic peeling argument. The reason for this is that the pull count $T_i(N)$ is random, so we cannot directly treat $\tfrac{S_{i,T_i(N)}^2}{T_i(N)^2}$ as if the denominator were deterministic. The event $\mathcal E$ only guarantees the lower bound $T_i(N)\ge n_0$, but does not fix the exact realized size of $T_i(N)$. We therefore partition $\mathcal E$ into dyadic slices according to the magnitude of $T_i(N)$, so that on each slice the denominator can be controlled by a deterministic scale.

Let $ m_0 := \lfloor \log_2 n_0\rfloor,$ so that $2^{m_0}\le n_0<2^{m_0+1}.$ For each integer $m\ge m_0$, define the dyadic slice
\[
\mathcal E_m:=\mathcal E\cap \{2^m \le T_i(N)<2^{m+1}\},
\]
then $\ind\{\mathcal E\}=\sum_{m=m_0}^{\infty}\ind\{\mathcal E_m\}.$ This decomposition allows us to control the random denominator $T_i(N)$ on each slice by the deterministic scale $2^m$.

Now
\begin{equation*}
    \EE \left[(\hat\mu_i(N)-\mu_i)^2\ind\{\mathcal E\}\right] = \sum_{m=m_0}^{\infty} \EE \left[ \frac{S_{i,T_i(N)}^2}{T_i(N)^2}\ind\{\mathcal E_m\} \right].    
\end{equation*}

On the event $\mathcal E_m$, we have $T_i(N)\in[2^m,2^{m+1})$, so
\[
\frac{S_{i,T_i(N)}^2}{T_i(N)^2}\ind\{\mathcal E_m\} \le \frac{1}{2^{2m}} \Bigl(\max_{1\le n<2^{m+1}} |S_{i,n}|\Bigr)^2 \ind\{\mathcal E_m\}.
\]
Therefore
\begin{equation}\label{eq:dyadic_decomp}
    \EE \left[(\hat\mu_i(N)-\mu_i)^2\ind\{\mathcal E\}\right] \le \sum_{m=m_0}^{\infty} 2^{-2m} \EE \left[ \max_{1\le n<2^{m+1}} |S_{i,n}|^2 \right].
\end{equation}

Now $(S_{i,n})_{n\ge 0}$ is an $L^2$-martingale with respect to its canonical filtration $\mathcal G_{i,n}:=\sigma(Y_{i,1},\dots,Y_{i,n}),$ and since the increments $Y_{i,s}-\mu_i$ are independent, mean zero, and have variance $\sigma_i^2$, we have $ \EE[|S_{i,2^{m+1}}|^2]=2^{m+1}\sigma_i^2.$ Now by Theorem~\ref{thm:doob-l2}, we have
\[
\EE \left[ \max_{1\le n<2^{m+1}} |S_{i,n}|^2 \right] \le 4\cdot 2^{m+1}\sigma_i^2 = 8\cdot 2^m \sigma_i^2.
\]

Substituting this into~\eqref{eq:dyadic_decomp} gives
\begin{align*}
\EE \left[(\hat\mu_i(N)-\mu_i)^2\ind\{\mathcal E\}\right] &\le \sum_{m=m_0}^{\infty} 2^{-2m}\cdot 8\cdot 2^m \sigma_i^2 \\
&=
8\sigma_i^2 \sum_{m=m_0}^{\infty} 2^{-m}
=
8\sigma_i^2 \cdot 2^{-m_0+1}
=
16\sigma_i^2\,2^{-m_0}.
\end{align*}

Finally, since $m_0=\lfloor \log_2 n_0\rfloor$, we have $2^{-m_0}\le 2 / n_0.$ Hence
\[
\EE \left[(\hat\mu_i(N)-\mu_i)^2\ind\{\mathcal E\}\right] \le 32 \frac{\sigma_i^2}{n_0}.
\]
This proves the claim with $C=32$.
\end{proof}

\section{Proofs of Main Results}

\subsection{Proof of Lemma~\ref{lem:mse_twostage_pcipw}}\label{sec:mse_pf}

\begin{proof}
Fix an arm \(i\in[K]\). Let \(\hat\mu_i^{\mathrm{pilot}}\) denote the pilot sample mean based on the first \(N_1\) rounds, where stage~1 uses uniform allocation. To make the decomposition exact, write the final Pilot-centered IPW estimator in score form:
\[
\hat\mu_i^{\mathrm{PCIPW}} = \frac{1}{N}\sum_{t=1}^N Z_{t,i}, \quad Z_{t,i} := \hat\mu_i^{\mathrm{pilot}} + \frac{\ind \{A_t=i\}\bigl(X_t-\hat\mu_i^{\mathrm{pilot}}\bigr)}{p_{t,i}}.
\]
For \(t\le N_1\), the stage-1 allocation is uniform, so \(p_{t,i}=1/K\).

We first simplify the pilot contribution. Since \(\hat\mu_i^{\mathrm{pilot}}\) is the sample mean of the pilot observations from arm \(i\), the residuals sum to zero, ie. $ \sum_{t=1}^{N_1}\ind \{A_t=i\}\bigl(X_t-\hat\mu_i^{\mathrm{pilot}}\bigr)=0.$ Hence
\[
\sum_{t=1}^{N_1} Z_{t,i} = \sum_{t=1}^{N_1}\hat\mu_i^{\mathrm{pilot}} + K\sum_{t=1}^{N_1}\ind \{A_t = i\} \bigl(X_t -\hat \mu_i^{\mathrm{pilot}}\bigr) = N_1 \hat \mu_i^{\mathrm{pilot}}.
\]
Therefore,
\[
\hat\mu_i^{\mathrm{PCIPW}} - \mu_i = \frac{N_1}{N}\bigl(\hat\mu_i^{\mathrm{pilot}}-\mu_i\bigr) + \frac{1}{N}\sum_{t=N_1+1}^N \bigl(Z_{t,i}-\mu_i\bigr).
\]

Squaring and taking expectations gives
\begin{align*}
\mathrm{MSE}(i,\pi^{\rm AN})
&= \EE \left[ \left( \frac{N_1}{N}(\hat\mu_i^{\mathrm{pilot}}-\mu_i) + \frac{1}{N}\sum_{t=N_1+1}^N (Z_{t,i}-\mu_i) \right)^2 \right] \\
&= \EE \left[\left(\frac{N_1}{N}(\hat\mu_i^{\mathrm{pilot}}-\mu_i)\right)^2\right] + \EE \left[\left(\frac{1}{N}\sum_{t=N_1+1}^N (Z_{t,i}-\mu_i)\right)^2\right] \\
& \quad + \frac{2N_1}{N^2} \EE \left[ (\hat\mu_i^{\mathrm{pilot}}-\mu_i) \sum_{t=N_1+1}^N (Z_{t,i}-\mu_i) \right].
\end{align*}

We now show that the cross-term vanishes. For each \(t>N_1\), \(\hat\mu_i^{\mathrm{pilot}}-\mu_i\) is \(\mathcal F_{t-1}\)-measurable, so by the tower property,
\begin{align*}
\EE \left[ (\hat\mu_i^{\mathrm{pilot}}-\mu_i)(Z_{t,i}-\mu_i) \right] &= \EE \left[ (\hat\mu_i^{\mathrm{pilot}}-\mu_i) \EE[Z_{t,i}-\mu_i\mid \mathcal F_{t-1}] \right].
\end{align*}
But conditional on \(\mathcal F_{t-1}\),
\begin{align*}
\EE[Z_{t,i}\mid \mathcal F_{t-1}] &= \hat\mu_i^{\mathrm{pilot}} + \frac{1}{p_{t,i}} \EE \left[ \ind \{A_t=i\}(X_t - \hat\mu_i^{\mathrm{pilot}})\middle| \mathcal F_{t-1} \right] \\
&= \hat\mu_i^{\mathrm{pilot}} + \frac{1}{p_{t,i}} p_{t,i}  \EE\left[ X_t-\hat\mu_i^{\mathrm{pilot}} \middle| A_t=i,\mathcal F_{t-1} \right] \\
&= \hat\mu_i^{\mathrm{pilot}}+(\mu_i-\hat\mu_i^{\mathrm{pilot}}) =\mu_i.
\end{align*}
Thus \(\EE[Z_{t,i}-\mu_i\mid \mathcal F_{t-1}]=0\), so the cross-term is zero. We are left with two terms.

For the first term, since stage~1 is uniform, arm \(i\) receives \(N_1/K\) pilot samples, and therefore
\[
\Var(\hat\mu_i^{\mathrm{pilot}})=\frac{\sigma_i^2}{N_1/K}.
\]
Hence
\[
\EE \left[\left(\frac{N_1}{N}(\hat\mu_i^{\mathrm{pilot}}-\mu_i)\right)^2\right] = \frac{N_1^2}{N^2}\Var(\hat\mu_i^{\mathrm{pilot}}) = \frac{N_1^2}{N^2}\cdot \frac{\sigma_i^2}{N_1/K} = \frac{N_1}{N} \mathrm{MSE}(i,\pi^{\rm uni}).
\]

For the second term, define \(\tilde Z_{t,i}:=Z_{t,i}-\mu_i\). Then
\begin{align*}
\EE \left[\left(\frac{1}{N}\sum_{t=N_1+1}^N \tilde Z_{t,i}\right)^2\right] &= \frac{1}{N^2} \EE\left[ \sum_{t=N_1+1}^N \tilde Z_{t,i}^2 + 2\sum_{N_1+1\le s<t\le N}\tilde Z_{s,i}\tilde Z_{t,i} \right].
\end{align*}
The cross-terms are again zero by previous analysis: for \(s<t\), \(\tilde Z_{s,i}\) is \(\mathcal F_{t-1}\)-measurable and \(\EE[\tilde Z_{t,i}\mid\mathcal F_{t-1}]=0\), so $ \EE[\tilde Z_{s,i}\tilde Z_{t,i}] = \EE \left[\tilde Z_{s,i} \EE[\tilde Z_{t,i}\mid\mathcal F_{t-1}]\right] =0.$
Therefore,
\[
\EE \left[\left(\frac{1}{N}\sum_{t=N_1+1}^N \tilde Z_{t,i}\right)^2\right] = \frac{1}{N^2}\sum_{t=N_1+1}^N \EE[\tilde Z_{t,i}^2].
\]
Since \(\EE[\tilde Z_{t,i}\mid\mathcal F_{t-1}]=0\),
\[
\EE[\tilde Z_{t,i}^2] = \EE \left[\Var(Z_{t,i}\mid\mathcal F_{t-1})\right].
\]

It remains to compute this conditional variance. A direct expansion yields
\begin{align*}
\EE[Z_{t,i}^2\mid \mathcal F_{t-1}] &= (\hat\mu_i^{\mathrm{pilot}})^2 + 2\hat\mu_i^{\mathrm{pilot}} \EE \left[ \frac{\ind \{A_t=i\}(X_t-\hat\mu_i^{\mathrm{pilot}})}{p_{t,i}} \middle| \mathcal F_{t-1} \right] \\
&\quad + \EE \left[ \frac{\ind \{A_t=i\}(X_t-\hat\mu_i^{\mathrm{pilot}})^2}{p_{t,i}^2} \middle| \mathcal F_{t-1} \right].
\end{align*}
Conditional on \(\mathcal F_{t-1}\), both \(\hat\mu_i^{\mathrm{pilot}}\) and \(p_{t,i}\) are fixed, so the middle term equals $2\hat\mu_i^{\mathrm{pilot}}(\mu_i-\hat\mu_i^{\mathrm{pilot}}).$ For the last term, by bias-variance decomposition,
\begin{align*}
\EE \left[ \ind \{A_t=i\}(X_t-\hat\mu_i^{\mathrm{pilot}})^2 \middle| \mathcal F_{t-1} \right] &= p_{t,i}  \EE \left[(X_t - \hat\mu_i^{\mathrm{pilot}})^2 \mid A_t=i, \mathcal F_{t-1}\right] \\
&= p_{t,i}\Bigl(\sigma_i^2+(\mu_i - \hat\mu_i^{\mathrm{pilot}})^2\Bigr).
\end{align*}
Hence
\[
\EE[Z_{t,i}^2\mid \mathcal F_{t-1}] = (\hat\mu_i^{\mathrm{pilot}})^2 + 2\hat\mu_i^{\mathrm{pilot}}(\mu_i -\hat\mu_i^{\mathrm{pilot}}) + \frac{\sigma_i^2 + (\mu_i-\hat\mu_i^{\mathrm{pilot}})^2}{p_{t,i}}.
\]
Subtracting \(\mu_i^2\) gives
\begin{align*}
\Var(Z_{t,i}\mid \mathcal F_{t-1}) &= \EE[Z_{t,i}^2\mid \mathcal F_{t-1}] - \mu_i^2 \\
&= \frac{\sigma_i^2}{p_{t,i}} + (\mu_i-\hat\mu_i^{\mathrm{pilot}})^2\left(\frac{1}{p_{t,i}}-1\right).
\end{align*}
Therefore,
\[
\EE \left[\left(\frac{1}{N}\sum_{t=N_1+1}^N (Z_{t,i}-\mu_i)\right)^2\right] = \frac{1}{N^2}\sum_{t=N_1+1}^N \EE \left[\frac{\sigma_i^2}{p_{t,i}}\right] + \frac{1}{N^2}\sum_{t=N_1+1}^N \EE \left[ (\mu_i-\hat \mu_i^{\mathrm{pilot}})^2\left(\frac{1}{p_{t,i}}-1\right) \right].
\]

Combining the pilot term and the adaptive term proves
\begin{align*}
\mathrm{MSE}(i,\pi^{\rm AN}) &= \frac{N_1}{N}\,\mathrm{MSE}(i,\pi^{\rm uni}) + \frac{1}{N^2}\sum_{t=N_1+1}^{N} \EE \left[\frac{\sigma_i^2}{p_{t,i}}\right] + \frac{1}{N^2}\sum_{t=N_1+1}^{N} \EE \left[ (\mu_i-\hat\mu_i^{\mathrm{pilot}})^2 \left(\frac{1}{p_{t,i}}-1\right) \right].
\end{align*}
\end{proof}

\subsection{Proof of Theorem~\ref{thm:general_inf_improvement}}
\label{app:pf_general_inf_improvement}

\begin{proof}
Throughout the proof, write \(p_{t,i}\) for the assignment probability induced by the policy \(\pi\). By Lemma~\ref{lem:mse_twostage_pcipw}, for each arm \(i\in[K]\),
\[
\mathrm{MSE}(i,\pi) = \frac{N_1}{N}\mathrm{MSE}(i,\pi^{\rm uni}) + \frac{1}{N^2}\sum_{t=N_1+1}^{N} \EE\left[\frac{\sigma_i^2}{p_{t,i}}\right] + \frac{1}{N^2}\sum_{t=N_1+1}^{N} \EE\left[ (\mu_i- \hat\mu_i^{\mathrm{pilot}})^2 \left(\frac{1}{p_{t,i}}-1\right) \right].
\]
When we sum over \(i\in[K]\), the last summation is exactly \(\mathcal B_t\). We next rewrite the adaptive-variance
summation in terms of \(\mathcal D_t\). Since \(\sum_{i=1}^K p_{t,i}=1\), we have the
pathwise identity
\begin{align*}
\mathcal D_t = \sum_{i=1}^K \frac{(\sigma_i-Sp_{t,i})^2}{p_{t,i}} = \sum_{i=1}^K \frac{\sigma_i^2}{p_{t,i}} -2S\sum_{i=1}^K \sigma_i +S^2\sum_{i=1}^K p_{t,i} = \sum_{i=1}^K \frac{\sigma_i^2}{p_{t,i}}-S^2.
\end{align*}
Therefore,
\[
\sum_{i=1}^K \frac{\sigma_i^2}{p_{t,i}} = S^2+\mathcal D_t.
\]
Also, under uniform sampling,
\[
\sum_{i=1}^K \mathrm{MSE}(i,\pi^{\rm uni}) = \sum_{i=1}^K \frac{\sigma_i^2}{N/K} = \frac{K}{N}V.
\]
Substituting these two identities into the previous display yields
\begin{align*}
\sum_{i=1}^K \mathrm{MSE}(i,\pi)
&= \frac{N_1}{N}\sum_{i=1}^K \mathrm{MSE}(i,\pi^{\rm uni}) + \frac{1}{N^2}\sum_{t=N_1+1}^{N} \EE\left[S^2+\mathcal D_t+\mathcal B_t\right] \\
&= \frac{N_1}{N}\sum_{i=1}^K \mathrm{MSE}(i,\pi^{\rm uni}) + \frac{N_2S^2}{N^2} + \frac{1}{N^2}\sum_{t=N_1+1}^{N} \EE\left[\mathcal D_t+\mathcal B_t\right].
\end{align*}
Subtracting the uniform benchmark gives
\begin{align*}
\sum_{i=1}^K \mathrm{MSE}(i,\pi) - \sum_{i=1}^K \mathrm{MSE}(i,\pi^{\rm uni})
&= -\frac{N_2}{N} \sum_{i=1}^K \mathrm{MSE}(i,\pi^{\rm uni}) + \frac{N_2S^2}{N^2} + \frac{1}{N^2}\sum_{t=N_1+1}^{N} \EE\left[\mathcal D_t+\mathcal B_t\right] \\
&= -\frac{N_2}{N}\cdot \frac{K}{N}V + \frac{N_2S^2}{N^2} + \frac{1}{N^2}\sum_{t=N_1+1}^{N} \EE\left[\mathcal D_t+\mathcal B_t\right] \\
&= -\frac{N_2}{N^2}(KV-S^2) + \frac{1}{N^2}\sum_{t=N_1+1}^{N} \EE\left[\mathcal D_t+\mathcal B_t\right].
\end{align*}
This proves the stated decomposition.

Finally, assuming \(N_2>0\), the policy \(\pi\) weakly improves over uniform sampling if and only if the left-hand side above is nonpositive, which is equivalent to
\[
\frac{1}{N_2}\sum_{t=N_1+1}^{N}
\EE\left[\mathcal D_t+\mathcal B_t\right] \le KV-S^2.
\]
This proves the consequent statement.
\end{proof}

\subsection{Proof of Proposition~\ref{prop:inf_suff}}
\label{app:pf_inf_suff}

\begin{proof}
Let \(N_2:=N-N_1\). For the two-stage adaptive Neyman policy, we have $p_{t,i}=\hat p_i^{\rm Ney},$ for all $t>N_1,$ so the post-pilot quantities \(\mathcal D_t\) and \(\mathcal B_t\) do not depend on \(t\). Write the common values as \(\mathcal D\) and \(\mathcal B\). By Theorem~\ref{thm:general_inf_improvement},
\[
\sum_{i=1}^K \mathrm{MSE}(i,\pi^{\rm AN}) - \sum_{i=1}^K \mathrm{MSE}(i,\pi^{\rm uni}) = \EE[\mathcal D+\mathcal B]-(KV-S^2).
\]
Since \(N_2>0\), it remains only to compute \(\EE[\mathcal D+\mathcal B]\).

First, for $\mathcal D$, we have
\[
\EE[\mathcal D] = \sum_{i=1}^K
\frac{(\sigma_i-S\hat p_i^{\rm Ney})^2}{\hat p_i^{\rm Ney}} = \sum_{i=1}^K \frac{\sigma_i^2}{\hat p_i^{\rm Ney}} - S^2 
\]

Also, under the Gaussian pilot sample, the pilot mean \(\hat\mu_i^{\rm pilot}\) is independent of
\((\hat\sigma_1,\ldots,\hat\sigma_K)\), and hence independent of \(\hat p_i^{\rm Ney}\). Since each arm has \(n_1=N_1/K\) pilot samples,
\[
\EE\left[(\mu_i-\hat\mu_i^{\rm pilot})^2\right] = \frac{\sigma_i^2}{n_1} = \frac{K\sigma_i^2}{N_1}.
\]
Therefore,
\[
\EE[\mathcal B] = \frac{K}{N_1} \sum_{i=1}^K \sigma_i^2 \left( \EE\left[\frac{1}{\hat p_i^{\rm Ney}}\right]-1 \right) = \frac{K}{N_1}(\sum_{i=1}^K \EE [\frac{\sigma_i^2}{\hat p_i^{\rm Ney}}] - V).
\]
Combining the two displays gives
\[
\EE[\mathcal D+\mathcal B] = \left(1+\frac{K}{N_1}\right)\sum_{i=1}^K \frac{\sigma_i^2}{\hat p_i^{\rm Ney}} - S^2 - \frac{K}{N_1}V.
\]

It remains to evaluate \(A\). By the plug-in Neyman definition,
\[
\frac{1}{\hat p_i^{\rm Ney}} = \frac{\sum_{\ell=1}^K \hat\sigma_\ell}{\hat\sigma_i} = 1+\sum_{j\ne i}\frac{\hat\sigma_j}{\hat\sigma_i}.
\]
For \(j\ne i\), the Gaussian pilot assumption gives $
\nu\hat\sigma_i^2/ \sigma_i^2\sim \chi_\nu^2,$ where $\nu:=n_1-1,$
with independence across arms. Hence
\[
\frac{\hat\sigma_j}{\hat\sigma_i} = \frac{\sigma_j}{\sigma_i} \sqrt{ \frac{\nu\hat\sigma_j^2/\sigma_j^2}{\nu\hat\sigma_i^2/\sigma_i^2} } = \frac{\sigma_j}{\sigma_i}\sqrt{F_{\nu,\nu}}.
\]
Thus, for \(\nu>1\),
\[
\EE\left[\frac{\hat\sigma_j}{\hat\sigma_i}\right] = \frac{\sigma_j}{\sigma_i}\beta_\nu,
\qquad
\beta_\nu
:=
\EE[\sqrt{F_{\nu,\nu}}]
=
\frac{
\Gamma\left(\frac{\nu+1}{2}\right)
\Gamma\left(\frac{\nu-1}{2}\right)
}{
\Gamma\left(\frac{\nu}{2}\right)^2
}.
\]
Therefore,
\[
\EE\left[\frac{1}{\hat p_i^{\rm Ney}}\right]
=
1+\beta_\nu\sum_{j\ne i}\frac{\sigma_j}{\sigma_i}.
\]

Substituting this expression into the decomposition above yields
\begin{align*}
\sum_{i=1}^K \mathrm{MSE}(i,\pi^{\rm AN}) - \sum_{i=1}^K \mathrm{MSE}(i,\pi^{\rm uni})  = \left(1+\frac{K}{N_1}\right)\beta_\nu(S^2-V) - (K-1)V.
\end{align*}
Summing up, the two-stage adaptive Neyman policy weakly improves over uniform sampling if and only if
\[
\left(1+\frac{K}{N_1}\right)\beta_\nu(S^2-V) \le (K-1)V.
\]
For \(K\ge 2\) with positive arm variances, \(S^2-V>0\). Hence this is equivalent to
\[
\beta_\nu \le \frac{(K-1)V}{S^2-V} \cdot \frac{1}{1+\frac{K}{N_1}},
\]
which is exactly condition~\eqref{eq:refined_condition_beta}.
\end{proof}

\subsection{Proof of Lemma~\ref{lem:oracle_fixedalloc_joint_clean}}\label{subsec:oracle_joint_pf}

\begin{proof}
We study the minimization of
\[
\mathcal J_N(p)=\lambda \sum_{i=1}^K \frac{\sigma_i}{\sqrt{N p_i}}+(1-\lambda)\sum_{i=1}^K p_i\Delta_i, \quad p\in\mathcal S_K.
\]

Since $\sigma_i>0$ for all $i\in[K]$, the term $p_i\mapsto \sigma_i/\sqrt{N p_i}$ diverges to $+\infty$ as $p_i\downarrow 0$. Hence any minimizer must lie in the interior of the simplex. Moreover, each map $p_i\mapsto \frac{\sigma_i}{\sqrt{N p_i}}$ is strictly convex on $(0,\infty)$, while the regret term is linear. Therefore $\mathcal J_N$ is strictly convex on the simplex, so it admits a unique minimizer $p^\star(N)\in\mathcal S_K$.

We now characterize this minimizer by the KKT conditions. Since the minimizer is interior, there are no active inequality constraints, and the Lagrangian is
\[
\mathcal L(p,\alpha) = \lambda \sum_{i=1}^K \frac{\sigma_i}{\sqrt{N p_i}} + (1-\lambda)\sum_{i=1}^K p_i\Delta_i + \alpha\Bigl(1-\sum_{i=1}^K p_i\Bigr).
\]
Differentiating with respect to $p_i$ gives
\[
\frac{\partial \mathcal L}{\partial p_i} = -\frac{\lambda \sigma_i}{2\sqrt N} p_i^{-3/2} + (1-\lambda)\Delta_i - \alpha.
\]
At the optimum, this is $0$ for every $i\in[K]$, hence
\[
\frac{\lambda \sigma_i}{2\sqrt N} p_i^{-3/2} = (1-\lambda)\Delta_i-\alpha.
\]
Solving for $p_i$ yields
\[
p_i^\star(N) = \left[
\frac{\lambda \sigma_i}{2\sqrt N\bigl((1-\lambda)\Delta_i-\alpha_N^\star\bigr)} \right]^{2/3}, \quad i\in[K],
\]
where the multiplier $\alpha_N^\star$ is chosen so that $\sum_{i=1}^K p_i^\star(N)=1$. This proves~\eqref{eq:q_star_closed_clean}.

We next study the asymptotic behavior as $N\to\infty$. Since $\Delta_{i^\star}=0$, the positivity of the denominator for $i=i^\star$ forces $ -\alpha_N^\star>0.$ We denote $a_N:=-\alpha_N^\star>0.$ Then
\[
p_{i^\star}^\star(N) = \left(\frac{\lambda \sigma_{i^\star}}{2\sqrt N a_N}\right)^{2/3},
\]
and for every suboptimal arm $i\neq i^\star$,
\[
p_i^\star(N) = \left(
\frac{\lambda \sigma_i}{2\sqrt N\bigl((1-\lambda)\Delta_i+a_N\bigr)} \right)^{2/3}.
\]

We will now show the rate of the optimal arm and $a_N=\frac{\lambda \sigma_{i^\star}}{2}N^{-1/2}(1+o(1))$, which can help us prove the asymptotic rate for the pull rate for suboptimal arms. Indeed, from the normalization condition, $1 = p_{i^\star}^\star(N) + \sum_{i\neq i^\star} p_i^\star(N).$ For every $i\neq i^\star$, since $\Delta_i>0$ is fixed and $a_N>0$,
\[
p_i^\star(N) \le \left( \frac{\lambda \sigma_i}{2(1-\lambda)\Delta_i} \right)^{2/3}N^{-1/3} \implies \sum_{i\neq i^\star} p_i^\star(N)=O(N^{-1/3}).
\]
Therefore
\[
p_{i^\star}^\star(N)=1-O(N^{-1/3}) \implies p_{i^\star}^\star(N)\to 1.
\]
Using the explicit formula for $p_{i^\star}^\star(N)$, this means $\left(\frac{\lambda \sigma_{i^\star}}{2\sqrt N\,a_N}\right)^{2/3}\to 1, $ which is equivalent to
\[
a_N=\frac{\lambda \sigma_{i^\star}}{2}N^{-1/2}(1+o(1)).
\]

Substituting this into the suboptimal-arm formula, for each $i\neq i^\star$ we obtain
\[
p_i^\star(N) = \left( \frac{\lambda \sigma_i}{2\sqrt N\bigl((1-\lambda)\Delta_i+a_N\bigr)} \right)^{2/3} = \left( \frac{\lambda \sigma_i}{2(1-\lambda)\Delta_i} \right)^{2/3} N^{-1/3}(1+o(1)),
\]
since $a_N=o(1)$ while each $\Delta_i>0$ is fixed. This proves the second part of~\eqref{eq:q_star_asymp_clean}. 

It remains to evaluate the objective at $p^\star(N)$. For the regret term,
\[
(1-\lambda)\sum_{i=1}^K p_i^\star(N)\Delta_i
=
(1-\lambda)\sum_{i\neq i^\star} p_i^\star(N)\Delta_i
=
\Theta(N^{-1/3}),
\]
because each suboptimal arm contributes a constant multiple of $N^{-1/3}$.

For the RMSE term,
\[
\lambda \sum_{i=1}^K \frac{\sigma_i}{\sqrt{N p_i^\star(N)}}.
\]
For the optimal arm, $
\frac{\sigma_{i^\star}}{\sqrt{N p_{i^\star}^\star(N)}} = \Theta(N^{-1/2}),$ since $p_{i^\star}^\star(N)\to 1$. For each suboptimal arm $i\neq i^\star$,
\[
\frac{\sigma_i}{\sqrt{N p_i^\star(N)}}
=
\Theta(N^{-1/3}),
\]
because $p_i^\star(N)=\Theta(N^{-1/3})$. Hence
\[
\lambda \sum_{i=1}^K \frac{\sigma_i}{\sqrt{N p_i^\star(N)}}=\Theta(N^{-1/3}),
\]
the dominant contribution coming from the suboptimal arms.

Combining the RMSE and regret terms yields
\[
\mathcal J_N\bigl(p^\star(N)\bigr)=\Theta(N^{-1/3}),
\]
which proves~\eqref{eq:oracle_opt_rate_clean}.
\end{proof}

\subsection{Proof of Theorem~\ref{thm:porp}}

\begin{proof}
Fix a horizon $N$. Define $ \Delta_{\min}:=\min_{i\neq i^\star}\Delta_i, \ \Delta_{\max}:=\max_{i\neq i^\star}\Delta_i, \ \sigma_{\min}:=\min_{i\in[K]}\sigma_i .$ By assumption, $\Delta_{\min}>0$ and $\sigma_{\min}>0$. 

\paragraph{Checkpoint phase concentration.}

Let $m_t:=\min_{i\in[K]}T_i(t).$ By the same block argument used in the proof of Theorem~\ref{thm:NARP}, there exists a constant $C_g<\infty$, depending only on $(K,m_0,\alpha)$, such that for all $t\ge m_0K+1$,
\begin{equation}
\label{eq:porp_min_count_nofloor}
m_t\ge \alpha\sqrt t-C_g .
\end{equation}
Moreover, the total number of forced-exploration rounds up to time $N$ is $O(\sqrt{N})$. Indeed, forced pulls can occur only to keep the least-sampled arms above the quota $g(t)=m_0+\lceil \alpha\sqrt t\rceil$, so by time $N$ the forced mechanism can allocate only $O(\sqrt N)$ pulls to each arm.

Let $s_N:=\lceil N^{1/2}\rceil.$ Choose
\[
A_0:= \max\left\{ \frac{\bar\nu^2}{\Delta_{\min}^2}, \frac{\bar\nu^4}{\sigma_{\min}^4} \right\}, \quad n_{\rm stab}(N):= \left\lceil C_{\rm stab}A_0\log(8KN^4)\right\rceil ,
\]
where $C_{\rm stab}$ is a sufficiently large universal constant.

By the forced-exploration floor, for all sufficiently large $N$, $T_i(s_N)\ge n_{\rm stab}(N),$ for all $ i\in[K],$ Since counts are nondecreasing, every arm has at least $n_{\rm stab}(N)$ samples throughout the post-checkpoint window $u=s_N+1,\dots,N$. By the the same argument as Theorem~\ref{thm:NARP}'s proof, with probability at least $1-N^{-3}$, for every $u\in\{s_N+1,\dots,N\}$ and every arm $i\in[K]$,
\[
|\hat\mu_{i,u-1}-\mu_i|\le \frac14\Delta_{\min},
\qquad
|\hat\sigma_{i,u-1}-\sigma_i|\le \frac12\sigma_{\min}.
\]
Denote this event by $\mathcal E_{\rm stab}(N)$.

On $\mathcal E_{\rm stab}(N)$, for every post-checkpoint time
$u\in\{s_N+1,\dots,N\}$, all empirical standard deviations used by the algorithm satisfy $ \frac12\sigma_i \le \hat\sigma_{i,u-1} \le \frac32\sigma_i.$ In particular, all empirical standard deviations are strictly positive after the checkpoint. Therefore, on $\mathcal E_{\rm stab}(N)$, the empirical plug-in optimization problem has a unique interior KKT solution at every non-forced post-checkpoint round. Also, on $\mathcal E_{\rm stab}(N)$, the empirical best arm equals the true best arm, ie. $ \hat i^\star_{u-1}=i^\star.$ Summing up, for every $i\neq i^\star$,
\begin{equation}
\label{eq:porp_gap_sd_stable_nofloor}
\frac12\Delta_i \le \hat\Delta_{i,u-1} \le \frac32\Delta_i, \quad \frac12\sigma_i \le \hat\sigma_{i,u-1} \le \frac32\sigma_i .
\end{equation}

\paragraph{Plug-in KKT Order}

Fix a non-forced round $u\in\{s_N+1,\dots,N\}$ and work on $\mathcal E_{\rm stab}(N)$. Since all empirical standard deviations are positive, the empirical minimizer is interior and satisfies the KKT formula. Since $\hat\Delta_{i^\star,u-1}=0$, the empirical multiplier is negative. Write
\[
a_u:=-\hat\alpha_u>0 .
\]
Then
\[
\hat p_{i^\star,u} = \left[ \frac{\lambda\hat\sigma_{i^\star,u-1}} {2\sqrt u\,a_u} \right]^{2/3}, \ \ \  \text{and } \hat p_{i,u} = \left[ \frac{\lambda\hat\sigma_{i,u-1}} {2\sqrt u\bigl((1-\lambda)\hat\Delta_{i,u-1}+a_u\bigr)} \right]^{2/3}, \quad \text{for }i \neq i^*.
\]

First, using $a_u>0$ and \eqref{eq:porp_gap_sd_stable_nofloor}, for every
$i\neq i^\star$,
\[
\hat p_{i,u} \le \left[ \frac{\lambda(3\sigma_i/2)}{2\sqrt u(1-\lambda)(\Delta_i/2)} \right]^{2/3} = \overline c_i u^{-1/3} \implies \sum_{i\neq i^\star}\hat p_{i,u} \le \overline C u^{-1/3}.
\]
For all sufficiently large $u$, this gives $ \hat p_{i^\star,u} = 1-\sum_{i\neq i^\star}\hat p_{i,u} \ge \frac12 .$ Using the best-arm KKT equation, this implies
\[
\left[ \frac{\lambda\hat\sigma_{i^\star,u-1}} {2\sqrt u\,a_u} \right]^{2/3} \ge \frac12 \implies a_u\le C_a u^{-1/2}
\]
for some finite constant $C_a$.

Now, again using \eqref{eq:porp_gap_sd_stable_nofloor}, for every
$i\neq i^\star$,
\[
(1-\lambda)\hat\Delta_{i,u-1}+a_u \le \frac32(1-\lambda)\Delta_i+C_a u^{-1/2} \le C_{\Delta,i}
\]
for all sufficiently large $u$. Therefore
\[
\hat p_{i,u} \ge \left[ \frac{\lambda(\sigma_i/2)}{2\sqrt u\,C_{\Delta,i}} \right]^{2/3} = \underline c_i u^{-1/3}.
\]
Consequently, for all $i \neq i^*$, we have $\hat p_{i,u} = O(u^{-1/3})$ and $\hat p_{i^\star,u}\ge \frac12.$

\paragraph{Lower Bounding Counts}

For each arm $i$, define
\[
Z_{u,i}:= \mathbf 1\{u\text{ is non-forced and }I_u=i\}, \quad u=s_N+1,\dots,N .
\]
Conditional on the past, $Z_{u,i}$ is Bernoulli with mean $\hat p_{i,u}$ on non-forced rounds and zero on forced rounds. Recall the number of forced rounds after the checkpoint is $O(\sqrt N)$. Combining this with $\hat p_{i,u} = O(u^{-1/3})$, for each $i\neq i^\star$,
\[
\sum_{u=s_N+1}^{N} \mathbb E[Z_{u,i}\mid\mathcal F_{u-1}] \ge \underline c_i \sum_{\substack{u=s_N+1\\ u\text{ non-forced}}}^{N}u^{-1/3} \ge c_i N^{2/3}
\]
for some $c_i>0$ and all sufficiently large $N$. The forced rounds remove only $O(\sqrt N)$ terms, which is negligible relative to $\Theta(N^{2/3})$. Similarly, by $\hat p_{i^*,u} \ge \frac{1}{2}$, $\sum_{u=s_N+1}^{N} \mathbb E[Z_{u,i^\star}\mid\mathcal F_{u-1}] \ge c_\star N$ for some $c_\star>0$. A standard lower-tail martingale bound for adapted Bernoulli variables then implies that, for some constants $c_{\rm cnt}>0$ and $c_{\rm pull}>0$,
\[
\mathcal E_{\rm cnt}(N)
:=
\left\{
T_i(N)\ge c_{\rm pull}N^{2/3}
\text{ for all }i\in[K]
\right\}
\]
satisfies
\begin{equation}
\label{eq:porp_count_prob_nofloor}
\mathbb P\bigl(\mathcal E_{\rm cnt}(N)^c\cap\mathcal E_{\rm stab}(N)\bigr) \le K\exp(-c_{\rm cnt}N^{2/3}).
\end{equation}
Summing up,
\begin{equation}
\label{eq:porp_good_prob_nofloor}
\mathbb P\bigl((\mathcal E_{\rm stab}(N)\cap\mathcal E_{\rm cnt}(N))^c\bigr) \le N^{-3}+K\exp(-c_{\rm cnt}N^{2/3}).
\end{equation}

The following steps to show the oracle rate of inference and regret will be identical to Theorem~\ref{thm:NARP} or~\ref{thm:joint_sarp_rate}, and hence we will omit.
\end{proof}

\subsection{Proof of Theorem~\ref{thm:joint_sarp_rate}}

\begin{proof}
Fix a horizon $N\ge 1$, and for each arm $i\in[K]$ define $\sigma_i^2:=\Var(X_{i,1}),$ which is finite since $X_{i,s}-\mu_i$ is $\bar\nu$-sub-Gaussian. We first prove that the exploration branch alone gives each arm order $N^{2/3}$ samples with high probability, which yields the desired RMSE rate; we then show that the regret contribution is also of order at most $N^{-1/3}$.

\paragraph{Inference bound.}
For each arm $i\in[K]$ and each round $t>K$, define
\[
Z_{t,i}:=\ind\{\text{round }t\text{ uses the exploration branch and selects arm }i\}.
\]
By the construction of Algorithm~\ref{alg:sarp}, for each fixed $i$, the variables $(Z_{t,i})_{t=K+1}^{\infty}$ are independent Bernoulli random variables with
\[
\PP (Z_{t,i}=1)=x_t p_i^0, \quad x_t:=\min\{1,c_x t^{-1/3}\}.
\]

We focus on the exploration pulls, and show that these suffice to give inference optimality. Define the number of explicit exploration pulls of arm $i$ up to horizon $N$ by $T_i^{\mathrm{exp}}(N):=\sum_{t=K+1}^N Z_{t,i}.$ We denote its mean as
\[
\bar n_i:=\EE[T_i^{\mathrm{exp}}(N)] = \sum_{t=K+1}^N x_t p_i^0 = p_i^0\sum_{t=K+1}^N x_t.
\]
We now show that $\bar n_i$ is of order $N^{2/3}$, which is precisely the sample-size scale needed for the RMSE term to be of order $N^{-1/3}$. Let $t_\star:=\lceil c_x^3\rceil,$ and $s_0:=\max\{K+1,t_\star\}.$ Then for every $t\ge s_0$, we have $x_t=c_x t^{-1/3}$, so the clipping only affects finitely many initial rounds and is asymptotically negligible. Therefore, for every $N\ge s_0$,
\[
\sum_{t=K+1}^N x_t \ge c_x\sum_{t=s_0}^N t^{-1/3}.
\]
Using the integral comparison for the decreasing positive function $x\mapsto x^{-1/3}$,
\[
\sum_{t=s_0}^N t^{-1/3} \ge \int_{s_0}^{N+1} x^{-1/3}\,dx = \frac32\bigl((N+1)^{2/3}-s_0^{2/3}\bigr).
\]
Now choose $N_0:=\left\lceil 2^{3/2}s_0\right\rceil.$ Then for every $N\ge N_0$, one has $s_0^{2/3}\le \tfrac12 N^{2/3}$, and therefore
\[
\sum_{t=K+1}^N x_t \ge \frac{3c_x}{2}\bigl(N^{2/3}-s_0^{2/3}\bigr) \ge \frac{3c_x}{4}N^{2/3}.
\]
Consequently, for every $N\ge N_0$,
\[
\bar n_i \ge \frac{3c_x p_i^0}{4}N^{2/3} =: \kappa_i N^{2/3}, \quad \kappa_i:=\frac{3c_x p_i^0}{4}>0.
\]

However, a lower bound on the expectation  $\bar n_i$ alone is not sufficient, since the realized count $T_i$ may still be pathologically small on some sample paths and thereby spoil the final RMSE rate. We therefore introduce a high-probability event on which the realized exploration counts are simultaneously bounded below by a constant fraction of their means.

Define the good event $ \mathcal E_{\mathrm{cnt}}(N) := \Bigl\{ T_i^{\mathrm{exp}}(N)\ge \tfrac12 \bar n_i\text{ for all } i\in[K] \Bigr\}.$ Since $T_i^{\mathrm{exp}}(N)$ is a sum of independent Bernoulli variables, by the multiplicative Chernoff bound stated in Lemma~\ref{lem:chernoff-lower-tail}, for every $i\in[K]$ and every $N\ge N_0$,
\[
\PP \left(T_i^{\mathrm{exp}}(N)<\frac12\bar n_i\right) \le \exp(-\bar n_i/8) \le \exp(-\kappa_i N^{2/3}/8).
\]

On this event $\mathcal E_{\mathrm{cnt}}(N)$, the exploration branch alone already guarantees that each arm is sampled at least on the order of $N^{2/3}$ times. Specifically, we have $T_i = T_{i,N} \ge T_i^{\mathrm{exp}}(N) \ge \frac12\bar n_i \ge \frac{\kappa_i}{2}N^{2/3}.$ Applying Lemma~\ref{lem:dyadic-mse-pullorder} with
$n_0:=\left\lfloor \frac{\kappa_i}{2}N^{2/3}\right\rfloor$ gives
\[
\EE \left[(\hat\mu_i(N)-\mu_i)^2\ind\{\mathcal E_{\mathrm{cnt}}(N)\}\right] \lesssim N^{-2/3}.
\]

It remains to control the contribution of the bad event. Applying a union bound over $i\in[K]$, we obtain $ \PP(\mathcal E_{\mathrm{cnt}}(N)^c) \le K\exp(-c_{\mathrm{cnt}}N^{2/3})$ for some constant $c_{\mathrm{cnt}}>0$. Thus the bad event is exponentially unlikely in $N^{2/3}$, and its contribution can be absorbed into the expectation by a Cauchy--Schwarz argument and the sub-Gaussian assumption. Since the rewards are sub-Gaussian, they have finite fourth moments. We denote,
\[
\hat\mu_i(N)-\mu_i=\frac{S_{i,T_i}}{T_i}, \quad S_{i,n}:=\sum_{s=1}^n (Y_{i,s}-\mu_i),
\]
and $T_i\ge 1$ deterministically because of the initialization phase. Therefore
\[
\EE \left[(\hat\mu_i(N)-\mu_i)^4\right] = \sum_{m=1}^N \EE \left[\left(\frac{S_{i,m}}{m}\right)^4 \ind\{T_i=m\}\right] \le \sum_{m=1}^N
\EE \left[\left(\frac{S_{i,m}}{m}\right)^4\right].
\]

For each $m\ge 1$, the average $m^{-1}S_{i,m}$ is again sub-Gaussian with scale of order $\bar\nu/\sqrt m$, so $ \EE \left[\left(\frac{S_{i,m}}{m}\right)^4\right] \lesssim \frac{\bar\nu^4}{m^2}.$ Hence
\[
\EE \left[(\hat\mu_i(N)-\mu_i)^4\right] \lesssim \bar\nu^4 \sum_{m=1}^\infty m^{-2} <\infty.
\]
By Cauchy--Schwarz,
\[
\EE \left[(\hat\mu_i(N)-\mu_i)^2\ind\{\mathcal E_{\mathrm{cnt}}(N)^c\}\right] \le \EE \left[(\hat\mu_i(N)-\mu_i)^4\right]^{1/2} \PP(\mathcal E_{\mathrm{cnt}}(N)^c)^{1/2} = o(N^{-2/3}).
\]
Hence the bad event is negligible to the main $N^{-2/3}$ contribution from the good event. Therefore,
\[
\EE \left[(\hat\mu_i(N)-\mu_i)^2\right] = O(N^{-2/3}) \implies \sqrt{\EE \left[(\hat\mu_i(N)-\mu_i)^2\right]} = O(N^{-1/3}).
\]
Summing over the fixed set of arms gives the same $O(N^{-1/3})$ rate.

\paragraph{Regret bound.}
We now verify that the same exploration schedule that is sufficient for inference contributes no more than order $N^{-1/3}$ to the normalized regret. Decompose the regret as $R_N = R_N^{\mathrm{init}} + R_N^{\mathrm{explore}} + R_N^{\ALG},$ where the three terms correspond respectively to the initialization rounds, the explicit exploration branch, and the exploitation branch governed by $\ALG$. Let $\Delta_{\max}:=\max_{i\neq i^\star}\Delta_i.$ For the initialization term,
\[
R_N^{\mathrm{init}}\le K\Delta_{\max} \implies \frac{\EE[R_N^{\mathrm{init}}]}{N}=O(N^{-1}).
\]

For the exploration term, on round $t>K$, the expected regret contributed by explicit exploration is $x_t \sum_{i\neq i^\star}\Delta_i p_i^0.$
Hence $ \EE[R_N^{\mathrm{explore}}] = \sum_{t=K+1}^N x_t \sum_{i\neq i^\star}\Delta_i p_i^0,$
so
\[
\frac{\EE[R_N^{\mathrm{explore}}]}{N} = \frac{1}{N}\sum_{t=K+1}^N x_t \sum_{i\neq i^\star}\Delta_i p_i^0.
\]
Using $x_t\le c_x t^{-1/3}$ for all $t\ge t_\star$, and absorbing the finitely many terms $t<t_\star$ into a constant, the same integral comparison gives us
\[
\sum_{t=K+1}^N x_t = O \left(\sum_{t=1}^N t^{-1/3}\right) = O(N^{2/3}).
\]
Therefore
\[
\frac{\EE[R_N^{\mathrm{explore}}]}{N} = O(N^{-1/3}).
\]

Finally, interpret $\ALG$ as updating only on exploitation rounds. If $M_N^{\ALG}$ denotes the random number of exploitation rounds up to time $N$, then $M_N^{\ALG}\le N$, and conditional on $M_N^{\ALG}=m$, the exploitation branch is exactly a run of $\ALG$ for $m$ rounds. Hence, by Definition~\ref{def:ALG-Exp-control-clean},
\[
\EE[R_N^{\ALG}\mid M_N^{\ALG}=m] \le C_{\ALG}\sqrt{K m\log m} \le C_{\ALG}\sqrt{K N\log N}.
\]
Taking expectation gives
\[
\frac{\EE[R_N^{\ALG}]}{N} \le C_{\ALG}\sqrt{\frac{K\log N}{N}} = o(N^{-1/3}).
\]

Combining the three pieces yields
\[
\frac{\EE[R_N]}{N}=O(N^{-1/3}).
\]

\paragraph{Combining Inference and Regret}
Combining the inference and regret bounds, we obtain
\[
\mathcal J_N(\pi^{\rm SARP}) = \lambda \sum_{i=1}^K \sqrt{\EE \left[(\hat\mu_i(N)-\mu_i)^2\right]} + (1-\lambda)\frac{\EE[R_N]}{N} = O(N^{-1/3}).
\]
Equivalently, there exist constants $C<\infty$ and $N_0<\infty$ such that
\[
\mathcal J_N(\pi^{\rm SARP})\le C N^{-1/3}, \quad \forall N\ge N_0.
\]
This proves the theorem.
\end{proof}

\subsection{Proof of Theorem~\ref{thm:NARP}}

%\blue{The proof below is somewhat stronger than what is strictly needed for the \(N^{-1/3}\) rate. A weaker nondegeneracy argument would suffice for rate optimality, but we adopt the more detailed stabilization proof because it makes the mechanism of NARP explicit: after a checkpoint, the plug-in quantities become uniformly well behaved and the algorithm effectively runs the intended oracle-calibrated \(t^{-1/3}\) exploration rule. This is useful not only for the present theorem, but also because it isolates the ingredients needed for a future oracle-tracking refinement, which is the natural explanation for NARP's empirical improvement over SARP.}

\begin{proof}
Fix a horizon $N\ge 1$, and analyze the single policy $\pi^{\rm NARP}$ up to round $N$. For each arm $i\in[K]$, define $ \sigma_i^2:=\Var(X_{i,1}),$ which is finite by $\bar \nu-$sub-Gaussianity. If some arm has $\sigma_i=0$, then that arm is deterministic and its inference contribution is simpler to control. Thus, for the main argument, we restrict attention to the nondegenerate case $\sigma_i>0$ for all $i\in[K]$, and define
\[
\Delta_{\max}:=\max_{i\neq i^\star}\Delta_i,\quad \Delta_{\min}:=\min_{i\neq i^\star}\Delta_i,\quad \sigma_{\min}:=\min_{i\in[K]}\sigma_i.
\]
Also define the population rooted-Neyman quantities
\[
p_i^{\rm rN}:=\frac{\sigma_i^{2/3}}{\sum_{k=1}^K \sigma_k^{2/3}},\qquad
M:=\sum_{i\neq i^\star}\sigma_i (p_i^{\rm rN})^{-1/2},\qquad
L:=\sum_{i\neq i^\star}\Delta_i p_i^{\rm rN}.
\]

Unlike the proof of SARP, where the exploration counts are sums of independent Bernoulli variables with a fixed mean profile, the proof of NARP must first show that the data-driven plug-in exploration rule becomes reliable. We therefore begin with a stabilization step: after a deterministic checkpoint, every arm has already been sampled often enough that the empirical means, empirical standard deviations, and the plug-in summaries \(\hat M_{t-1}\) and \(\hat L_{t-1}\) are uniformly well behaved. Once this is established, the remainder of the argument parallels the logic of the SARP proof.

\paragraph{Checkpoint phase concentration.}

We first establish a deterministic lower bound on the minimum arm count rate. Let
\[
m_t:=\min_{i\in[K]} T_{i,t}, \quad \tilde m_t:=m_t-m_0.
\]
Since each arm is pulled exactly $m_0$ times during the warm start, we have $\tilde m_t\ge 0$ for all $t\ge m_0K$. We claim that there exists a constant $C_g<\infty$, depending only on $(K,m_0,\alpha)$, such that for all $t\ge m_0K+1$,
\begin{equation}\label{eq:count_rate}
    m_t\ge \alpha\sqrt t-C_g.
\end{equation}

To prove this, it suffice to show that $\tilde m_t\ge \alpha\sqrt t-C_g'$ for some finite constant $C_g'$.

Fix $t_0:=\max\{m_0K+1,\lceil (\alpha K)^2\rceil\}.$ For $t\le t_0$, the claim is trivial after enlarging the constant, since $\tilde m_t\ge 0$. Now partition the time axis after $t_0$ into blocks of length $K$:
\[
B_r:=\{t_r+1,\dots,t_r+K\},
\qquad
t_r:=t_0+rK,
\qquad r\ge 0.
\]
We will show that the deficit of the shifted minimum count relative to the target $\alpha\sqrt t$ cannot increase from one block endpoint to the next. Indeed, for every $r\ge 0$,
\[
\alpha(\sqrt{t_r+K}-\sqrt{t_r}) = \frac{\alpha K}{\sqrt{t_r+K}+\sqrt{t_r}} \le 1,
\]
where the last inequality holds because $t_r\ge t_0\ge (\alpha K)^2$. Thus the target level $\alpha\sqrt t$ increases by at most $1$ over any single block. Now fix a block $B_r$. Two cases can occur:

\begin{enumerate}
    \item If there exists some $u\in B_r$ such that $\tilde m_{u-1}\ge \lceil \alpha\sqrt u\rceil,$ then at that time the minimum count is already at least the target level. Since $\tilde m_t$ is nondecreasing in $t$, we have $\tilde m_{t_r+K}\ge \tilde m_{u-1}\ge \alpha\sqrt u.$ Because the target increases by at most $1$ over the remainder of the block,
    \[
    \alpha\sqrt{t_r+K}\le \alpha\sqrt u+1 \implies \alpha\sqrt{t_r+K}-\tilde m_{t_r+K}\le 1.
    \]

    \item Otherwise, for every $u\in B_r$, we have $\tilde m_{u-1}<\lceil \alpha\sqrt u\rceil.$ In this case every round in the block is forced. Since each forced round pulls a currently least-sampled arm, after $K$ such rounds every arm that was at the minimum level at the start of the block has been incremented at least once. Hence the minimum count increases by at least one:
\[
\tilde m_{t_r+K}\ge \tilde m_{t_r}+1.
\]
Using again that the target increases by at most $1$ over the block, we obtain
\[
\alpha\sqrt{t_r+K}-\tilde m_{t_r+K}
\le
\alpha\sqrt{t_r}-\tilde m_{t_r}.
\]

\end{enumerate}
    
For each block endpoint $t_r$, define the deficit $d_r:=\alpha\sqrt{t_r}-\tilde m_{t_r}.$ From the two cases above, we have shown that $d_{r+1}\le \max\{d_r,1\}.$ Hence, by induction, $d_r\le C_g':=\max\{d_0,1\}$  for all $r\ge 0.$ Equivalently,
\[
\tilde m_{t_r}\ge \alpha\sqrt{t_r}-C_g', \quad r\ge 0.
\]

Now we extend to the $t\ge t_0$ arbitrary case, and choose $r$ such that $t_r\le t<t_r+K=t_{r+1}.$ Since $\tilde m_t$ is nondecreasing in $t$, and the target increases by at most $1$ over any block,
\[
\alpha\sqrt t\le \alpha\sqrt{t_r+K}\le \alpha\sqrt{t_r}+1.
\]
Combining the two bounds gives
\[
\alpha\sqrt t-\tilde m_t \le \alpha\sqrt{t_r}+1-\tilde m_{t_r} = d_r+1 \le C_g'+1,
\]
which proves~\eqref{eq:count_rate}.

Next, we denote $s_N:=\lceil N^{1/2}\rceil.$ Define
\[
A_0:=\max\Bigl\{\frac{\bar\nu^2}{\Delta_{\min}^2},\frac{\bar\nu^4}{\sigma_{\min}^4}\Bigr\}, \quad
n_{\mathrm{stab}}(N):=\left\lceil C_{\mathrm{stab}}A_0\log(8KN^4)\right\rceil,
\]
where $C_{\mathrm{stab}}>0$ is a sufficiently large universal constant. These quantities are chosen so that the local concentration bounds below are strong enough to make the final stabilization event summable over the entire post-checkpoint timesteps. By the deterministic count floor in~\eqref{eq:count_rate},
\[
T_{i,s_N}\ge \alpha\sqrt{s_N}-C_g\ge \alpha N^{1/4}-C_g, \quad i\in[K].
\]
Since $N^{1/4}\gg \log N$, there exists $N_1<\infty$ such that for all $N\ge N_1$, $T_{i,s_N}\ge n_{\mathrm{stab}}(N),$ for all $i\in[K].$ By monotonicity of the counts, the same lower bound then holds for every $u\in\{s_N+1,\dots,N\}$. We next show that once every arm has accumulated at least \(n_{\mathrm{stab}}(N)\) samples by the checkpoint time, all empirical means and empirical standard deviations remain uniformly accurate throughout the post-checkpoint window with high probability. For each $u\in\{s_N+1,\dots,N\}$, define the local concentration event
\[
\mathcal E_{u-1}^{\mathrm{loc}} := \Bigl\{ \max_{i\in[K]} |\hat\mu_{i,u-1}-\mu_i| \le \tfrac14 \Delta_{\min}, \quad \max_{i\in[K]} |\hat\sigma_{i,u-1}-\sigma_i| \le \tfrac12 \sigma_{\min} \Bigr\}.
\]
For a fixed time $u$, the event $(\mathcal E_{u-1}^{\mathrm{loc}})^c$ is the union over arms of either a mean-concentration failure or a standard-deviation-concentration failure. Thus, applying Lemma~\ref{lem:subg-hoeff} and Corollary~\ref{lem:sd-concentration} to each arm and then taking a union bound over the $K$ arms, the choice of $n_{\mathrm{stab}}(N)$ ensures
\[
\PP\bigl((\mathcal E_{u-1}^{\mathrm{loc}})^c\bigr)\le N^{-4}, \quad u=s_N+1,\dots,N.
\]
Therefore, if we define the post-checkpoint stabilization event as $\mathcal E_{\mathrm{stab}}(N):=\bigcap_{u=s_N+1}^{N}\mathcal E_{u-1}^{\mathrm{loc}},$ then a further union bound over the at most $N$ post-checkpoint times gives
\[
\PP\bigl(\mathcal E_{\mathrm{stab}}(N)^c\bigr)\le N\cdot N^{-4}\le N^{-3}
\]

\paragraph{Global good event definition.}
With stabilized estimated parameters, we now work on $\mathcal E_{\mathrm{stab}}(N)$. Fix a non-forced round $u\in\{s_N+1,\dots,N\}$, that is, a round for which $T_{U_u,u-1}\ge \alpha\sqrt u+1.$ On $\mathcal E_{u-1}^{\mathrm{loc}}$, we have
\[
\frac12 \sigma_i \le \hat\sigma_{i,u-1}\le \frac32 \sigma_i, \quad i\in[K].
\]
Therefore
\[
\hat p_{i,u-1}^{\rm rN} = \frac{\hat\sigma_{i,u-1}^{2/3}}{\sum_{k=1}^K \hat\sigma_{k,u-1}^{2/3}} \ge \frac{(1/2)^{2/3}\sigma_i^{2/3}}{\sum_{k=1}^K (3/2)^{2/3}\sigma_k^{2/3}} = 3^{-2/3}p_i^{\rm rN},
\]
and similarly $\hat p_{i,u-1}^{\rm rN}\le 3^{2/3}p_i^{\rm rN}.$ Moreover, for every $i\neq i^\star$,
\[
\hat\mu_{i^\star,u-1}-\hat\mu_{i,u-1} \ge \Bigl(\mu_{i^\star}-\tfrac14\Delta_{\min}\Bigr) - \Bigl(\mu_i+\tfrac14\Delta_{\min}\Bigr) = \Delta_i-\tfrac12\Delta_{\min} \ge \tfrac12\Delta_{\min}>0.
\]
Hence
\[
\hat i^\star_{u-1}=i^\star, \text{ and} \quad \frac12\Delta_i\le \hat\Delta_{i,u-1}\le \frac32\Delta_i, \quad i\neq i^\star.
\]

Using these bounds, we obtain
\[
\frac12 3^{-1/3}M \le \hat M_{u-1} \le \frac32 3^{1/3}M, \quad \frac12 3^{-2/3}L \le \hat L_{u-1} \le \frac32 3^{2/3}L.
\]
In particular, $\hat L_{u-1}>0$ on $\mathcal E_{u-1}^{\mathrm{loc}}$, so the convention $x_u=1$ if $\hat L_{u-1}=0$ is irrelevant on this event. Therefore there exist constants $\underline b_x,\overline b_x>0$, depending only on the fixed instance and on $\lambda$, such that for all sufficiently large non-forced rounds $u$,
\[
\underline b_x u^{-1/3}\le x_u\le \overline b_x u^{-1/3}.
\]

We now convert this stabilized \(u^{-1/3}\) exploration rate into a high-probability lower bound on the realized post-checkpoint exploration counts. The argument has two steps: first, lower bound the predictable post-checkpoint exploration mass for each arm; second, show that the realized counts stay within a constant factor of that predictable mass.

For each arm $i\in[K]$, define
\[
Z_{u,i}:=\ind\{\text{round }u\text{ is non-forced, uses plug-in exploration, and selects arm }i\},
\]
and its predictable conditional mean
\[
p_{u,i}:=\EE[Z_{u,i}\mid \mathcal F_{u-1}] = \ind\{T_{U_u,u-1}\ge \alpha\sqrt u+1\} x_u \hat p_{i,u-1}^{\rm rN}.
\]
Hence, on $\mathcal E_{u-1}^{\mathrm{loc}}$, for every non-forced round $u$, $p_{u,i}\ge c_0 u^{-1/3}$ for some constant $c_0>0$.

Let
\[
\mathcal N_N:=\{u\in\{s_N+1,\dots,N\}:\text{round }u\text{ is non-forced}\}.
\]
Then define
\[
P_i^{\mathrm{post}}(N):=\sum_{u=s_N+1}^{N} p_{u,i} = \sum_{u\in\mathcal N_N} p_{u,i}.
\]
Since an arm can be force-selected at round $u\le N$ only while its count is below $\alpha\sqrt u+1\le \alpha\sqrt N+1$, each arm receives at most $\alpha\sqrt N+1$ forced pulls by time $N$. Hence the total number of forced rounds in $[s_N+1,N]$ is at most $K(\alpha\sqrt N+1)$. Therefore
\[
|\mathcal N_N|\ge (N-s_N)-K(\alpha\sqrt N+1).
\]

Now, on $\mathcal E_{\mathrm{stab}}(N)$, for every non-forced round $u\in\mathcal N_N$,
\[
p_{u,i}\ge c_0 u^{-1/3}\ge c_0 N^{-1/3},
\]
since $u\le N$. Summing over all non-forced rounds gives
\[
P_i^{\mathrm{post}}(N)
\ge
c_0 |\mathcal N_N| N^{-1/3}
\ge
c_0\bigl((N-s_N)-K(\alpha\sqrt N+1)\bigr)N^{-1/3}.
\]
Since $s_N=\lceil N^{1/2}\rceil$, the right-hand side is of order $N^{2/3}$, so there exists $c_->0$ such that
\[
P_i^{\mathrm{post}}(N)\ge c_- N^{2/3}
\]
for all sufficiently large $N$.

Likewise, on $\mathcal E_{\mathrm{stab}}(N)$, we have $p_{u,i}\le \overline b_x u^{-1/3}$ on every round $u$, so we can derive $ P_i^{\mathrm{post}}(N) \le c_+ N^{2/3}$ for some constant $c_+<\infty$. Next define $T_i^{\mathrm{post}}(N):=\sum_{u=s_N+1}^{N} Z_{u,i}$ as the actual count of non-forced pull, and the martingale deviation process
\[
D_{i,s}:=\sum_{u=s_N+1}^{s}(Z_{u,i}-p_{u,i}), \quad s=s_N+1,\dots,N.
\]
Its increments are bounded by $1$, and its predictable quadratic variation
\[
V_{i,N}:=\sum_{u=s_N+1}^{N}\EE\bigl[(Z_{u,i}-p_{u,i})^2\mid \mathcal F_{u-1}\bigr]
\]
satisfies $V_{i,N}\le P_i^{\mathrm{post}}(N)$ by simple algebra on indicator function. Hence
\[
\mathcal E_{\mathrm{stab}}(N)\cap \Bigl\{T_i^{\mathrm{post}}(N)\le \tfrac12 P_i^{\mathrm{post}}(N)\Bigr\}
\subseteq
\Bigl\{D_{i,N}\le -\tfrac12 c_- N^{2/3},\ \ V_{i,N}\le c_+ N^{2/3}\Bigr\}.
\]
Applying Lemma~\ref{lem:freed_ineq} yields
\[
\PP\Bigl(\mathcal E_{\mathrm{stab}}(N)\cap \{T_i^{\mathrm{post}}(N)\le \tfrac12 P_i^{\mathrm{post}}(N)\}\Bigr)
\le
e^{-cN^{2/3}}
\]
for some constant $c>0$.

Finally, define the count event $\mathcal E_{\mathrm{cnt}}(N) := \Bigl\{ T_i^{\mathrm{post}}(N)\ge \tfrac12 P_i^{\mathrm{post}}(N)\ \text{for all }i\in[K] \Bigr\}.$ A union bound yields
\[
\PP\bigl(\mathcal E_{\mathrm{stab}}(N)\cap \mathcal E_{\mathrm{cnt}}(N)^c\bigr) \le K e^{-cN^{2/3}}.
\]
Now define the global good event $\mathcal E:=\mathcal E_{\mathrm{stab}}(N)\cap \mathcal E_{\mathrm{cnt}}(N).$
Combining the previous bounds gives
\[
\PP(\mathcal E^c)=O(N^{-2}).
\]

Finally, on $\mathcal E$, for every arm $i$, $ T_i \ge T_i^{\mathrm{post}}(N) \ge \frac12 P_i^{\mathrm{post}}(N) \ge \frac12 c_- N^{2/3},$ for all sufficiently large $N$, hence we obtain the lower bound we need with high probability.

From this point onward, the regret and inference argument largely parallels the proof of SARP, and we include it only for completeness.

\paragraph{Regret bound.}

We decompose the pseudo-regret as
\[
R_N = R_N^{\mathrm{forced}} + R_N^{\mathrm{explore}} + R_N^{\ALG},
\]
where the three terms correspond respectively to the undersampled-set forced-exploration rule, the rooted-Neyman plug-in exploration branch, and the exploitation branch governed by $\ALG$.

First, the forced-exploration part contributes at most $K(\alpha\sqrt N+1)$ rounds by time $N$, deterministically. Hence
\[
\frac{\EE[R_N^{\mathrm{forced}}]}{N}=O(N^{-1/2}).
\]

Next, let
\[
Z_u:=\ind\{\text{round }u\text{ is non-forced and plug-in exploration is selected}\},
\qquad u=s_N+1,\dots,N.
\]
Then
\[
R_N^{\mathrm{explore}}\le \Delta_{\max}\sum_{u=s_N+1}^{N} Z_u.
\]
Therefore
\[
\EE[R_N^{\mathrm{explore}}]
\le
\Delta_{\max}\sum_{u=s_N+1}^{N} \EE[Z_u]
=
\Delta_{\max}\sum_{u=s_N+1}^{N}\EE\bigl[\EE[Z_u\mid \mathcal F_{u-1}]\bigr].
\]
Since
\[
\EE[Z_u\mid \mathcal F_{u-1}]
=
\ind\{T_{U_u,u-1}\ge \alpha\sqrt u+1\}x_u,
\]
and $\mathcal E_{u-1}^{\mathrm{loc}}\in\mathcal F_{u-1}$, we split
\begin{align*}
\EE[Z_u]
&=
\EE\bigl[\ind\{T_{U_u,u-1}\ge \alpha\sqrt u+1\}x_u \ind\{\mathcal E_{u-1}^{\mathrm{loc}}\}\bigr] \\
&\quad
+
\EE\bigl[\ind\{T_{U_u,u-1}\ge \alpha\sqrt u+1\}x_u \ind\{(\mathcal E_{u-1}^{\mathrm{loc}})^c\}\bigr].
\end{align*}
On $\mathcal E_{u-1}^{\mathrm{loc}}$, we have $x_u\le \overline b_x u^{-1/3}$, while always $x_u\le 1$. Hence
\[
\EE[Z_u]\le \overline b_x u^{-1/3}+\PP\bigl((\mathcal E_{u-1}^{\mathrm{loc}})^c\bigr)\le \overline b_x u^{-1/3}+N^{-4}.
\]
Summing over $u=s_N+1,\dots,N$, we get
\[
\EE[R_N^{\mathrm{explore}}]=O(N^{2/3}),
\]
and therefore
\[
\frac{\EE[R_N^{\mathrm{explore}}]}{N}=O(N^{-1/3}).
\]

Finally, by Definition~\ref{def:ALG-Exp-control-clean},
\[
\frac{\EE[R_N^{\ALG}]}{N}
\le
C_{\ALG}\sqrt{\frac{K\log N}{N}}
=
o(N^{-1/3}).
\]

Combining the three pieces, we conclude
\[
\frac{\EE[R_N]}{N}=O(N^{-1/3}).
\]

\paragraph{Inference MSE bound.}

Fix an arm $i\in[K]$. Under the independent-sequences model, let
\[
S_{i,n}:=\sum_{s=1}^n (Y_{i,s}-\mu_i),
\qquad n\ge 0.
\]
Then
\[
\hat\mu_i(N)-\mu_i=\frac{S_{i,T_i}}{T_i}.
\]

On the event $\mathcal E$, we have the deterministic lower bound $T_i\ge cN^{2/3}$ for some constant $c>0$. Therefore Lemma~\ref{lem:dyadic-mse-pullorder} applies with $n_0:=\lfloor cN^{2/3}\rfloor$, yielding
\[
\EE\Bigl[(\hat\mu_i(N)-\mu_i)^2 \ind\{\mathcal E\}\Bigr]\le C_0 N^{-2/3}
\]
for some constant $C_0<\infty$.

On the bad event $\mathcal E^c$, we use a fourth-moment bound. Since the rewards are $\bar\nu$-sub-Gaussian, all moments are finite. Hence, by Cauchy--Schwarz,
\[
\EE\Bigl[(\hat\mu_i(N)-\mu_i)^2 \ind\{\mathcal E^c\}\Bigr] \le \Bigl(\EE[(\hat\mu_i(N)-\mu_i)^4]\Bigr)^{1/2}\PP(\mathcal E^c)^{1/2}
=
O(N^{-1}).
\]
Therefore
\[
\EE\bigl[(\hat\mu_i(N)-\mu_i)^2\bigr]
=
O(N^{-2/3})+O(N^{-1})
=
O(N^{-2/3}).
\]

Summing over $i\in[K]$, we obtain
\[
\sum_{i=1}^{K}\sqrt{\EE\bigl[(\hat\mu_i(N)-\mu_i)^2\bigr]}
=
O(N^{-1/3}).
\]

\paragraph{Unconditional objective assembly.}

By the regret and inference bounds above,
\[
\lambda \sum_{i=1}^{K}\sqrt{\EE\bigl[(\hat\mu_i(N)-\mu_i)^2\bigr]} = O(N^{-1/3}), \quad (1-\lambda)\frac{\EE[R_N]}{N} = O(N^{-1/3}).
\]
Therefore
\[
\mathcal J_N(\pi^{\rm NARP}) = \lambda \sum_{i=1}^{K}\sqrt{\EE\bigl[(\hat\mu_i(N)-\mu_i)^2\bigr]} + (1-\lambda)\frac{\EE[R_N]}{N} = O(N^{-1/3}).
\]
This proves the theorem.
\end{proof}

\section{Virtual-Work Field Experiment and Dataset}\label{apx:virtual_data}

Additional details of the original experimental design are as follows. Assignment to the five conditions was randomized within organizational division, with division-specific assignment probabilities reflecting differences in the availability of senior managers. The interventions operated during weeks 2--5 of the internship and differed primarily in whether interaction was formal or informal and synchronous or asynchronous; the passive control introduced no additional structured interaction. Table~\ref{tab:virtual_conditions} describes the five conditions and their corresponding sample sizes.

\subsection{Treatment Conditions and Arm Construction}

Table~\ref{tab:virtual_conditions} describes the five experimental conditions, which we treat as the \(K=5\) arms of the pseudo-bandit environment. The original sample sizes refer to the number of interns assigned to each condition in the complete randomized sample. The analysis sample sizes refer to the number of interns in each condition with an observed final supervisor performance rating.

\begin{table}[htbp!]
    \centering
    \small
    \renewcommand{\arraystretch}{1.15}
    \caption{Experimental conditions in the virtual-work field experiment.}
    \label{tab:virtual_conditions}
    \begin{tabularx}{\textwidth}{ p{3.1cm} p{2.5cm} X c c
    }
        \toprule
        Experimental condition & Interaction structure & Description & Original \(N\) & Analysis \(N\) \\
        \midrule

        Passive control & N/A  & Interns were not assigned any virtual interactions treatment. & 518 & 504 \\

        Asynchronous Q\&A & Informal and asynchronous & Interns type a question on the discussion forum at the end of each week and are exogenously selected to receive a reply from a senior manager from their division  & 223 & 216 \\

        Intern group project & Formal and synchronous & Interns meet once per week over videoconference (i.e., Zoom) with the same group of three to four peers from their division to work on an Intern Group Challenge project  & 192 & 187 \\

        Intern-only virtual water cooler & Informal and synchronous & Interns meet between zero and four times (exogenously determined by the panel experimental design) over videoconference (i.e., Zoom) with a small group of three to four peers from their division to exchange advice and social information  & 218 & 211 \\

        Intern--senior-manager virtual water cooler & Informal and synchronous & Interns meet between zero and four times (exogenously determined by the panel experimental design) over videoconference (i.e., Zoom) with a small group of three to four peers and a senior manager from their division to exchange advice and social information  & 219 & 212 \\

        \midrule
        Total
        & & & 1{,}370 & 1{,}330 \\
        \bottomrule
    \end{tabularx}
\end{table}

Within the two virtual-water-cooler conditions, the original experiment included an additional week-by-week randomization of invitations. In each of weeks 2--5, available water-cooler slots were allocated through an independent lottery, so an intern assigned to a water-cooler condition could be invited to between zero and four sessions over the course of the internship. The authors use this variation to study treatment-dose effects and, for the intern--senior-manager condition, also examines whether interns and senior managers were demographically matched.

We do not use this within-condition variation to define additional arms. Instead, an intern's arm is determined by the original assignment to one of the five experimental conditions. In particular, we do not divide the water-cooler conditions according to the number of invitations received, actual attendance, demographic matching, or the identity of the participating senior manager. This arm construction follows an intent-to-treat interpretation: outcomes are grouped according to the experimental condition to which the intern was assigned, regardless of the number of sessions in which the intern subsequently participated.

\subsection{Analysis Sample and Reward Construction}

Our application uses the intern's final supervisor performance rating as the reward. The rating was assigned during the final week of the internship by the intern's direct supervisor. These direct supervisors were distinct from the senior managers who participated in the intern--senior-manager virtual water-cooler sessions. As we mentioned in Section~\ref{sec:application}, the original performance rating takes values \(1\), \(2\), and \(3\), where
\[
1=\text{outstanding performance},\qquad 2=\text{good performance},\qquad 3=\text{performance needing improvement}.
\]
Following the coding used in the original analysis, we reverse-code this
variable as
\[
Y=4-\text{Original Rating},
\]
so that $ Y\in\{1,2,3\},$ with larger values corresponding to better performance.

We restrict the sample to interns with an observed final performance rating. After this restriction and the associated data cleaning, the application sample contains \(1{,}330\) interns. Let $ D_j=\{Y_{j1},\ldots,Y_{jn_j}\}$ denote the collection of reverse-coded performance ratings observed in arm \(j\), where \(n_j\) is the post-cleaning sample size reported in Table~\ref{tab:virtual_conditions}. For \(r\in\{1,2,3\}\), define
\[
\widehat p_{jr} = \frac{1}{n_j} \sum_{i=1}^{n_j} \mathbf{1}\{Y_{ji}=r\}.
\]
The empirical reward distribution associated with arm \(j\) is therefore $ \widehat F_j = \sum_{r=1}^3 \widehat p_{jr}\delta_r = \frac{1}{n_j} \sum_{i=1}^{n_j} \delta_{Y_{ji}},$ where \(\delta_y\) denotes a point mass at \(y\). The corresponding empirical mean and standard deviation are
\[
\mu_j^{\mathrm{emp}} = \frac{1}{n_j} \sum_{i=1}^{n_j}Y_{ji} = \sum_{r=1}^3 r\widehat p_{jr}, \quad \left(\sigma_j^{\mathrm{emp}}\right)^2 = \frac{1}{n_j} \sum_{i=1}^{n_j} \left(Y_{ji}-\mu_j^{\mathrm{emp}}\right)^2 = \sum_{r=1}^3 \widehat p_{jr} \left(r-\mu_j^{\mathrm{emp}}\right)^2.
\]
Let $ j^\star \in \arg\max_{j\in[K]} \mu_j^{\mathrm{emp}}$ denote an empirically best arm, and define empirical gap of arm \(j\) by
\[
\Delta_j^{\mathrm{emp}} = \mu_{j^\star}^{\mathrm{emp}} - \mu_j^{\mathrm{emp}}.
\]

Table~\ref{tab:virtual_arm_summary} reports the empirical reward distribution and the corresponding bandit characteristics for each arm. Because the reward takes only three possible values, the columns \(\widehat p_{j1}\), \(\widehat p_{j2}\), and \(\widehat p_{j3}\) completely characterize \(\widehat F_j\). The empirical means, standard deviations, and gaps summarize the quantities that determine
the inference and regret components of our adaptive-sampling application. Conditional on the observed data, these five empirical distributions define the pseudo-bandit environment used in our application. When a policy selects arm \(j\), the resulting synthetic reward is sampled independently with replacement from \(\widehat F_j\). Thus, the application uses the original field-experiment data to calibrate the arm-specific reward distributions, but does not re-estimate the causal effects of the original interventions.

\begin{table}[htbp!]
    \centering
    \small
    \renewcommand{\arraystretch}{1.15}
    \caption{Empirical reward distributions and arm characteristics in the
    virtual-work application.}
    \label{tab:virtual_arm_summary}
    \begin{tabularx}{\textwidth}{X c c c c c c c}
        \toprule
        Arm & \(n_j\) & \(\widehat p_{j1}\) & \(\widehat p_{j2}\) & \(\widehat p_{j3}\) & \(\mu_j^{\mathrm{emp}}\)  & \(\sigma_j^{\mathrm{emp}}\) & \(\Delta_j^{\mathrm{emp}}\) \\
        \midrule

        Passive control & 504  & 0.042 & 0.296 & 0.663 & 2.621 & 0.565 & 0.053 \\

        Asynchronous Q\&A & 216  & 0.060 & 0.301 & 0.639 & 2.579 & 0.603 & 0.096 \\

        Intern group project & 187 & 0.086 & 0.321 & 0.594 & 2.508 & 0.649 & 0.167 \\

        Intern-only virtual water cooler & 211 & 0.052 & 0.284 & 0.664 & 2.611 & 0.585 & 0.063 \\

        Intern--senior-manager virtual water cooler & 212 & 0.042 & 0.241 & 0.715 & 2.675 & 0.552 & 0.000 \\

        \bottomrule
    \end{tabularx}

    \vspace{0.4em}
    \begin{minipage}{0.98\textwidth}
        \footnotesize
        \textit{Notes:}
        \(n_j\) is the number of interns with an observed reward in arm \(j\). The quantity \(\widehat p_{jr}\) is the empirical probability of reverse-coded rating \(Y=r\). The empirical standard deviation uses the finite empirical-distribution convention with denominator \(n_j\). The gap is \(\Delta_j^{\mathrm{emp}} =\mu_{j^\star}^{\mathrm{emp}}-\mu_j^{\mathrm{emp}}\); hence an empirically best arm has gap zero. Probabilities are rounded to six decimal places.
    \end{minipage}
\end{table}

Additional details about the original experimental design, including division-specific assignment probabilities, covariate-balance checks, week-by-week water-cooler invitation lotteries, demographic matching, treatment-dose variation, survey outcomes, and the implementation of the original intent-to-treat analysis, are provided by \citet{choudhury2021virtualwatercoolers}. We also refer readers to the original paper for its estimates of the interventions' effects on job offers, performance ratings, intern attitudes, and the proposed mechanisms underlying those effects. We do not reproduce those analyses here because our purpose is to use the observed arm-specific outcome distributions as an empirically calibrated environment for evaluating adaptive-sampling policies.

\section{Application: Behavioral Interventions to Encourage Exercise: A Megastudy}\label{apx:gym_app}
\subsection{Dataset Description}

We use deidentified data from the randomized field megastudy of \citet{milkman2021megastudies}, conducted in partnership with 24 Hour Fitness and released through the study's Open Science Framework repository. The megastudy was designed to make behavioral interventions directly comparable by evaluating them in the same population, using the same objectively measured outcome, and over the same four-week intervention duration.

The interventions were delivered through StepUp, a four-week digital program designed to encourage exercise. The study randomized 62{,}746 gym members to 54 experimental conditions, of whom 61{,}293 remained in the final analytic sample after the study's reported exclusions. The final sample included participants from 46 U.S. states, had a mean age of 39.13 years (SD = 13.25), and was 65.9\% female. Before the intervention, participants averaged 1.27 weekly gym visits over weeks $-4$ through $-1$ (SD = 1.48). These sample characteristics and the geographic coverage are reported by \citet{milkman2021megastudies} along with other demographic analysis.

The 54 conditions comprised a placebo control and 53 active interventions. The active conditions combined behavioral tools such as workout planning, text-message reminders, micro-incentives, return-after-missed-workout bonuses, social-norm information, commitment devices, and temptation bundling. The placebo-control group received 1{,}500 points at enrollment, worth US\$1.08, but did not receive a workout schedule or additional intervention content. Supplementary Table~1 of \citet{milkman2021megastudies} provides the complete list of conditions and their detailed stimuli, messages, and incentive schedules; the paper's Methods and Supplementary Information describe the original time-varying randomization and weighted estimation procedures.

Our application constructs three nested pseudo-bandit environments with $K\in\{5,9,13\}$. Each environment includes the placebo control, and active conditions are added in descending order of the bundled estimate of their effect on weekly gym visits. The $K=5$ environment contains the placebo control and the four highest effect treatments: Bonus for Returning after Missed Workouts$^{b}$, Higher Incentives$^{a}$, Exercise Social Norms Shared (High and Increasing), and Free Audiobook Provided. The corresponding empirical reward pools contain 9{,}931 participants.

The $K=9$ environment retains these five arms and adds the next four highest-ranked active conditions: Bonus for Returning after Missed Workouts$^{a}$, Planning Fallacy Described and Planning Revision Encouraged, Choice of Gain- or Loss-Framed Micro-Incentives, and Exercise Commitment Contract Explained. Its empirical reward pools contain 14{,}479 participants. Finally, the $K=13$ environment retains all nine of these arms and adds Free Audiobook Provided, Temptation Bundling Explained; Following Workout Plan Encouraged; Fitness Questionnaire with Decision Support \& Cognitive Reappraisal Prompt; and Values Affirmation. The full $K=13$ environment is reported in Table~\ref{tab:gym_arm_summary} and contains 18{,}292 reward-pool observations. We use these three nested arm sets to examine whether the policy comparisons persist as the action space and the number of arm means to be estimated increase. The $K=5$, $K=9$, and $K=13$ specifications therefore provide smaller, intermediate, and larger empirical environments.

\begin{table}[htbp!]
    \centering
    \small
    \renewcommand{\arraystretch}{1.12}
    \caption{Selected experimental conditions and empirical reward-pool characteristics for the full StepUp environment ($K=13$). The $K=5$ and $K=9$ environments are nested subsets described in the text.}
    \label{tab:gym_arm_summary}
    \begin{tabularx}{\textwidth}{X r r r r r}
        \toprule
        Experimental condition
        & Study $N_j$
        & Reward $n_j$
        & $\widehat\mu_j$
        & $\widehat s_j$
        & $\widehat\Delta_j$ \\
        \midrule

        Placebo Control
        & 4{,}992 & 4{,}606 & 0.000 & 1.086 & 0.403 \\

        Bonus for Returning after Missed Workouts$^{b}$
        & 1{,}633 & 1{,}497 & 0.403 & 1.136 & 0.000 \\

        Higher Incentives$^{a}$
        & 1{,}750 & 1{,}605 & 0.365 & 1.199 & 0.038 \\

        Exercise Social Norms Shared (High and Increasing)
        & 798 & 728 & 0.345 & 1.159 & 0.058 \\

        Free Audiobook Provided
        & 1{,}604 & 1{,}495 & 0.343 & 1.146 & 0.060 \\

        Bonus for Returning after Missed Workouts$^{a}$
        & 1{,}719 & 1{,}582 & 0.336 & 1.105 & 0.067 \\

        Planning Fallacy Described and Planning Revision Encouraged
        & 811 & 754 & 0.325 & 1.140 & 0.078 \\

        Choice of Gain- or Loss-Framed Micro-Incentives
        & 1{,}652 & 1{,}456 & 0.284 & 1.265 & 0.119 \\

        Exercise Commitment Contract Explained
        & 810 & 756 & 0.279 & 1.147 & 0.124 \\

        Free Audiobook Provided, Temptation Bundling Explained
        & 1{,}685 & 1{,}584 & 0.278 & 1.175 & 0.125 \\

        Following Workout Plan Encouraged
        & 805 & 684 & 0.268 & 1.191 & 0.135 \\

        Fitness Questionnaire with Decision Support \& Cognitive Reappraisal
        Prompt
        & 825 & 777 & 0.255 & 1.100 & 0.148 \\

        Values Affirmation
        & 824 & 768 & 0.243 & 1.150 & 0.160 \\

        \midrule
        Selected-pool total
        & 19{,}908 & 18{,}292 & & & \\
        \bottomrule
    \end{tabularx}

    \vspace{0.4em}
    \begin{minipage}{0.98\textwidth}
        \footnotesize
        \textit{Notes:}
        Study $N_j$ is the number of participants in the released final panel
        with intervention-period observations. Reward $n_j$ is the number who
        also have at least one available baseline observation in weeks
        $-52,\ldots,-1$ and therefore enter the empirical reward pool.
        The reward-pool mean $\widehat\mu_j$ is centered to the bundled paper
        coefficient; the placebo centering target is constructed as zero.
        The reported standard deviation $\widehat s_j$ uses denominator
        $n_j-1$, matching the analytical oracle, and
        $\widehat\Delta_j=\max_{\ell\in[13]}\widehat\mu_\ell-\widehat\mu_j$.
        Superscripts distinguish incentive amounts using the notation of the
        original study. The selection and ranking are descriptive. See
        Supplementary Table~1 of \citet{milkman2021megastudies} for all 54
        conditions and their full implementation details.
    \end{minipage}
\end{table}

\subsection{Analysis Sample and Reward Construction}

The original study's primary outcome was the number of days on which a participant visited the gym in a given week. Our pseudo-bandit does not use raw intervention-period attendance as its reward. Instead, it constructs a baseline-adjusted participant-level change and then shifts each arm to match the corresponding adjusted baseline.

For participant $i$ in arm $j$, let $V_{jiw}$ be the number of gym visits in week $w$ and define intervention-period attendance as
\[
    A_{ji}:=\frac{1}{4}\sum_{w=1}^{4}V_{jiw}.
\]
Let $\mathcal W^{\mathrm{pre}}_{ji}\subseteq\{-52,\ldots,-1\}$ be the set of weeks with an available finite baseline observation, and let $m_{ji}=|\mathcal W^{\mathrm{pre}}_{ji}|$. For participants with $m_{ji}\geq1$, define
\[
    B_{ji}:=\frac{1}{m_{ji}} \sum_{w\in\mathcal W^{\mathrm{pre}}_{ji}}V_{jiw}, \qquad  C_{ji}:=A_{ji}-B_{ji}.
\]
This rule uses the available baseline weeks and does not replace absent weeks with zeros. Across all 54 conditions, 56{,}292 of the 61{,}293 participants have at least one usable baseline observation. Within the selected $K=13$ conditions, the resulting reward pools contain 18{,}292 participants, as reported in Table~\ref{tab:gym_arm_summary}. Let $\overline C_j=n_j^{-1}\sum_{i=1}^{n_j}C_{ji}$ and let $\widetilde\beta_j$ denote the three-decimal weekly-visits coefficient supplied with the original study. We construct the reward
\[
    R_{ji}:=C_{ji}-\overline C_j+\widetilde\beta_j.
\]
Consequently, the empirical arm mean satisfies
\[
    \widehat\mu_j:=\frac{1}{n_j}\sum_{i=1}^{n_j}R_{ji} =\widetilde\beta_j.
\]
This hybrid construction preserves the observed within-arm dispersion of the simple participant-level changes while setting the arm means equal to the bundled estimates from the paper. It is not the outcome used in the original weighted fixed-effects regression, and it should not be interpreted as a re-estimation of that regression.

Let $D_j=\{R_{j1},\ldots,R_{jn_j}\}$ denote the resulting reward pool. Conditional on the observed data, arm $j$ is represented by the unweighted empirical distribution $ \widehat F_j:=\frac{1}{n_j}\sum_{i=1}^{n_j}\delta_{R_{ji}},$ where $\delta_r$ denotes a point mass at $r$. The plug-in sample variance used
by the analytical oracle is
\[
    \widehat s_j^2:=\frac{1}{n_j-1} \sum_{i=1}^{n_j}\left(R_{ji}-\widehat\mu_j\right)^2.
\]
Although the original paper weighted observations to account for assignment probabilities that varied across enrollment cohorts, our empirical reward pools do not apply cohort or inverse-propensity weights. Every eligible observation within a selected arm has equal sampling probability.

When a policy selects arm $I_t=j$, it receives a synthetic reward drawn with replacement according to $X_t\mid I_t=j,\mathcal D\sim\widehat F_j,$ independently across pulls conditional on the selected arms and the fixed empirical environment $\mathcal D=(D_1,\ldots,D_{K})$. Feasible policies observe only their own assignments and realized rewards. The complete dataset is available only to the evaluator for constructing the empirical environment and complete-information benchmark.

\subsection{Experiment Setup}

We evaluate the policies using $B= 1{,}000$ Monte Carlo replications and horizons $N$ ranging from 500 to $10{,}000$.  Our main tradeoff specification is $\lambda=0.1$. We compare uniform allocation, PORP, SARP, NARP, \texttt{ForcingBalance}, and the complete-information fixed-allocation oracle. SARP and NARP use Thompson Sampling as their exploitation component. As in the virtual-work application, this exercise is an empirically calibrated Monte Carlo comparison of adaptive-assignment policies. It neither re-estimates the causal effects in \citet{milkman2021megastudies} nor implies that the original megastudy should have used adaptive randomization.

\subsection{Results}

We first report the joint objective for each policy and horizon in Figure~\ref{fig:app_gym_jointloss}. Joint loss decreases with the budget in each of the three arm sets as expected, but the ordering of the feasible policies is not identical at smaller budgets. For $K=5$, PORP has the lowest feasible joint loss at every displayed horizon. For $K=9$, NARP is toe to toe with PORP at $N=1{,}000$ through $N=2{,}000$; PORP leads at $N=500$ and again from $N=3{,}000$ onward. For $K=13$, NARP leads from $N=1{,}000$ through $N=4{,}000$; PORP leads at $N=500$ and from $N=5{,}000$ onward. PORP and NARP both improve on uniform allocation at every displayed horizon for all three arm sets. At $N=10{,}000$, PORP is the best feasible policy in every environment.

\begin{figure}[htbp!]
    \centering
    \includegraphics[width=\linewidth]    {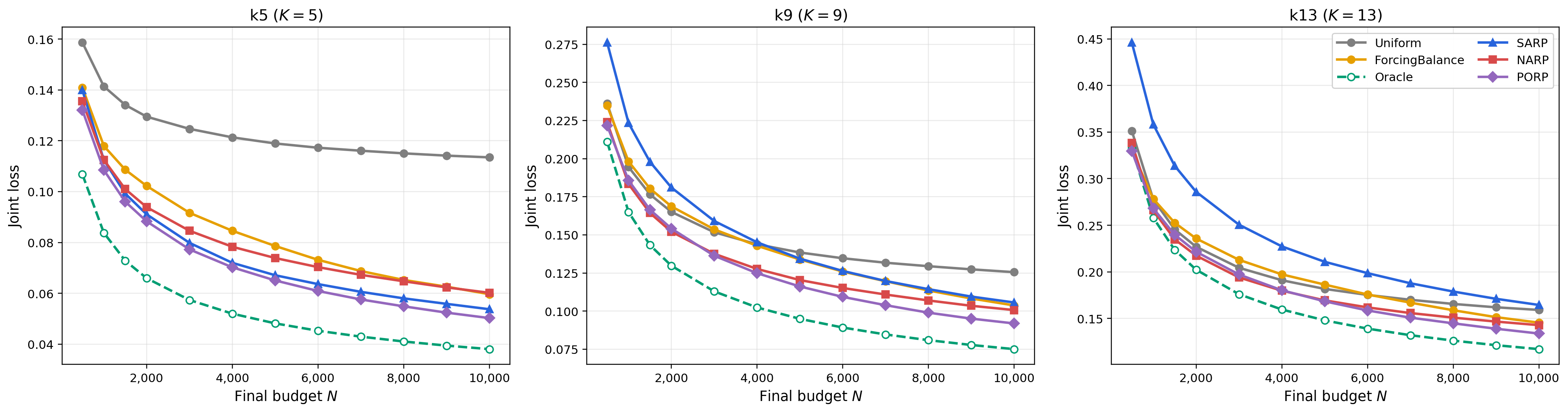}
    \caption{Joint inference--regret performance in the StepUp gym application. The panels report the joint objective for the nested $K=5$, $K=9$, and $K=13$ environments as the total budget $N$ increases. For each feasible sequential policy, the curve is the Monte Carlo mean across $B=1{,}000$ replications. The objective assigns weight $\lambda=0.1$ to summed RMSE and weight $1-\lambda=0.9$ to average pseudo-regret. The dashed oracle curve is a complete-information fixed-allocation benchmark evaluated analytically from the empirical arm moments rather than by Monte Carlo simulation.}
    \label{fig:app_gym_jointloss}
\end{figure}

Figures~\ref{fig:app_gym_decomp_k5}--\ref{fig:app_gym_decomp_k13} separately report summed RMSE and average pseudo-regret. Uniform allocation generally has the lowest summed RMSE among the feasible policies for each $K$, but it also has the highest average pseudo-regret. At the other extreme, \texttt{ForcingBalance} has the lowest endpoint pseudo-regret for each $K$ but incurs substantially higher summed RMSE. NARP retains more estimation precision than PORP while accepting more regret, which explains its advantage at the intermediate horizons for $K=9$ and $K=13$. As $N$ grows and estimation error declines, PORP's lower regret relative to NARP yields the more favorable weighted compromise. This tradeoff also shows why lower regret alone does not guarantee a lower joint objective. For example, SARP still has slightly higher endpoint joint loss than uniform allocation when $K=13$ because its increase in summed RMSE outweighs its regret reduction.

\begin{figure}[htbp!]
    \centering
    \includegraphics[width=\linewidth]
    {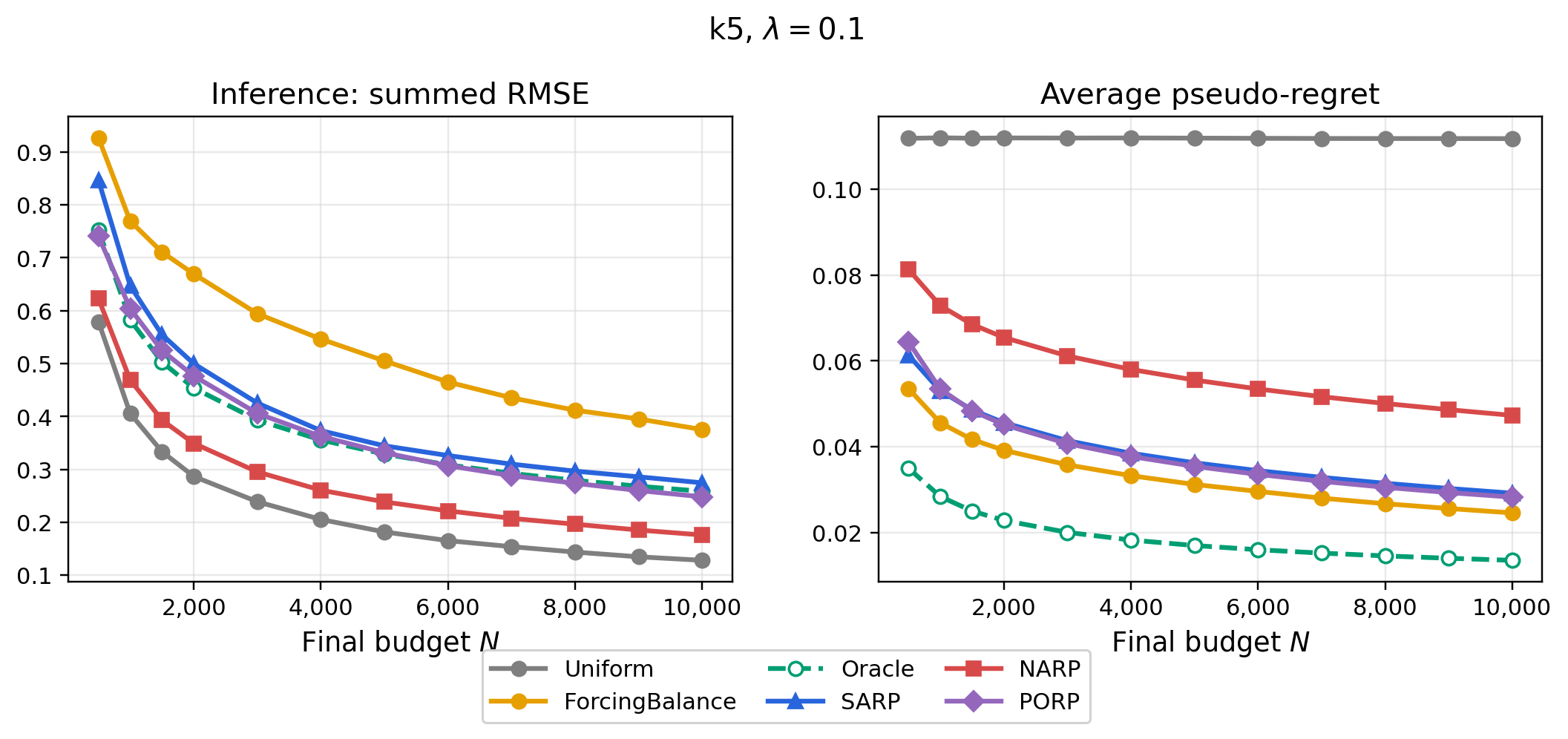}
    \caption{Decomposition of performance in the StepUp gym application for
    $K=5$, based on $B=1{,}000$ Monte Carlo replications for each feasible policy. The left panel reports summed RMSE across the five empirical-arm mean estimators, and the right panel reports cumulative pseudo-regret divided by $N$. These components enter the joint objective with weights $\lambda=0.1$ and $1-\lambda=0.9$, respectively. The dashed oracle curve is the analytical complete-information fixed-allocation benchmark.}
    \label{fig:app_gym_decomp_k5}
\end{figure}

\begin{figure}[htbp!]
    \centering
    \includegraphics[width=\linewidth]
    {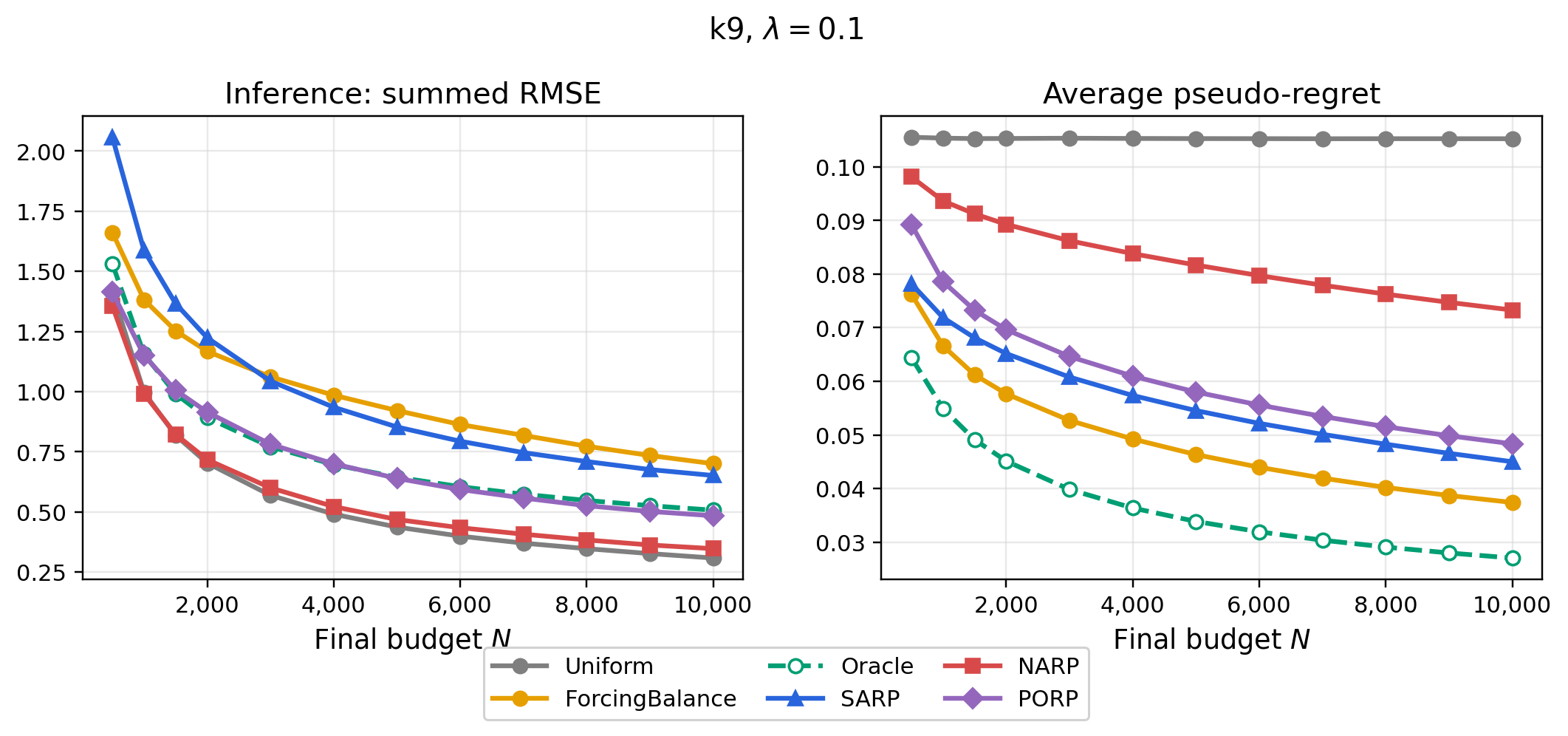}
    \caption{Decomposition of performance in the StepUp gym application for $K=9$, based on $B=1{,}000$ Monte Carlo replications for each feasible policy. The left panel reports summed RMSE across the nine empirical-arm mean estimators, and the right panel reports cumulative pseudo-regret divided by $N$. These components enter the joint objective with weights $\lambda=0.1$ and $1-\lambda=0.9$, respectively. The dashed oracle curve is the analytical complete-information fixed-allocation benchmark.}
    \label{fig:app_gym_decomp_k9}
\end{figure}

\begin{figure}[htbp!]
    \centering
    \includegraphics[width=\linewidth]
    {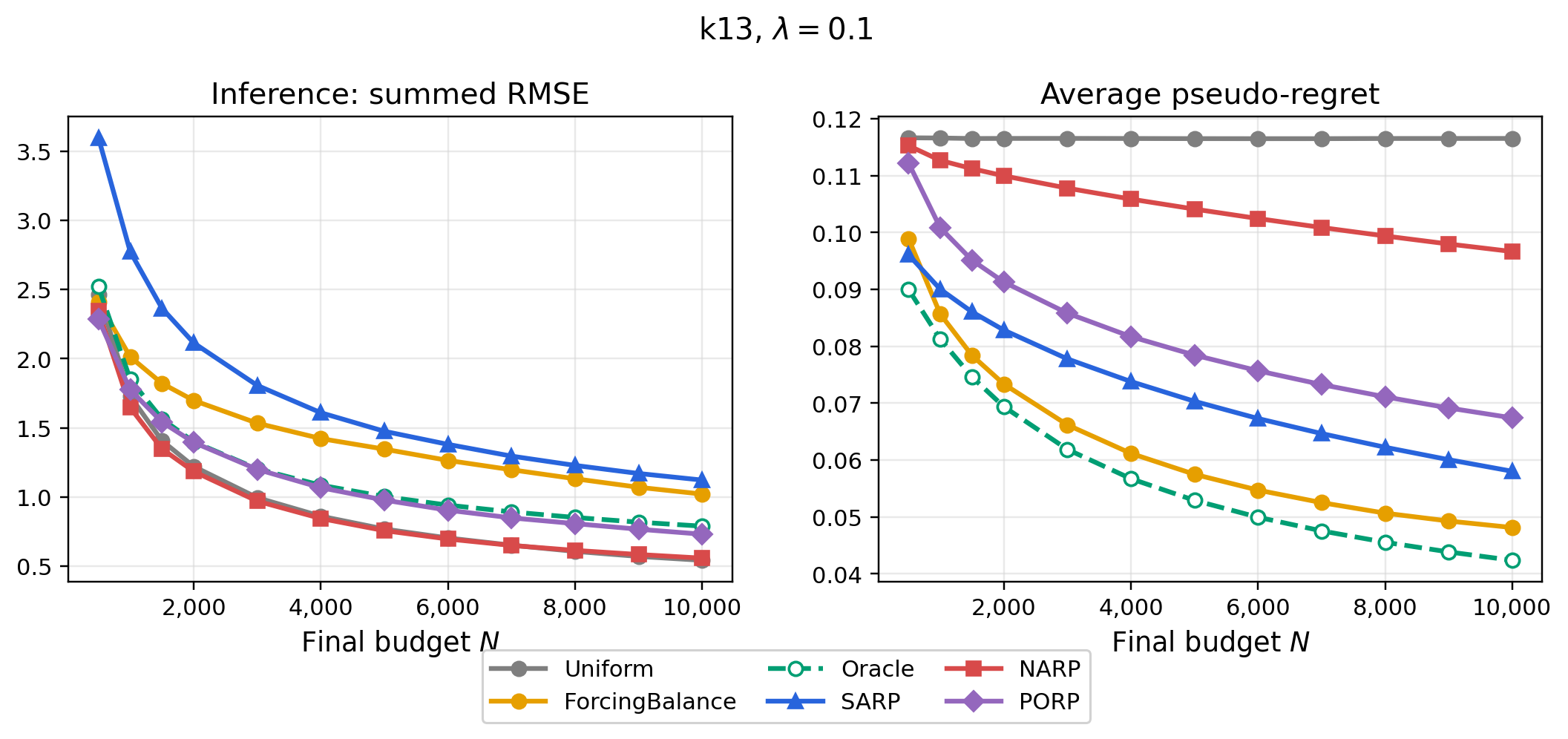}
    \caption{Decomposition of performance in the StepUp gym application for $K=13$, based on $B=1{,}000$ Monte Carlo replications for each feasible policy. The left panel reports summed RMSE across the thirteen empirical-arm mean estimators, and the right panel reports cumulative pseudo-regret divided by $N$. These components enter the joint objective with weights  $\lambda=0.1$ and $1-\lambda=0.9$, respectively. The dashed oracle curve is the analytical complete-information fixed-allocation benchmark.}
    \label{fig:app_gym_decomp_k13}
\end{figure}

Finally, Table~\ref{tab:app_gym_computation_time} reports computation time for the four adaptive policies in the $K=9$ environment. We use $K=9$ as a representative case because the same broad pattern appears for $K=5$ and $K=13$: the direct-update policies are less computationally expensive than the two policies that repeatedly solve numerical allocation equations. SARP is fastest because it uses a fixed exploration distribution and a closed-form exploration rate. NARP additionally updates the empirical means, variances, gaps, and rooted-Neyman allocation. PORP and \texttt{ForcingBalance} are more expensive because they repeatedly solve numerical allocation problems; in the $K=9$ benchmark, PORP is slightly faster than \texttt{ForcingBalance}.

\begin{table}[!t]
    \centering
    \begin{tabular}{lc}
        \toprule
        Policy & Average serial time ($\mu\mathrm{s}$/decision) \\
        \midrule
        SARP                    & 117.1 \\
        NARP                    & 144.1 \\
        PORP                    & 312.4 \\
        \texttt{ForcingBalance} & 324.9 \\
        \bottomrule
    \end{tabular}
    \caption{Average serial wall-clock time per sequential decision for $K=9$ and $\lambda=0.1$. Values average 200{,}000 timed decisions per policy, obtained from 20 paths of length $N=10{,}000$ after one untimed warm-up path. A timed decision includes policy selection, empirical reward sampling, state updating, and checkpoint bookkeeping; oracle evaluation, aggregation, and plotting are excluded.}
    \label{tab:app_gym_computation_time}
\end{table}

%\section{SARP/NARP Fixed Horizon Version}
%\input{sections/apx_fixedhorizon}

\end{document}